\documentclass[a4paper,10pt]{article}
\usepackage{amsfonts}
\usepackage{amsmath}
\usepackage{amsthm}
\usepackage{amsmath, graphics}
\usepackage{bbm}
\usepackage{amssymb}
\usepackage{appendix}
\usepackage{epsf}
\newcommand{\mathsym}[1]{{}}

\usepackage[active]{srcltx}
\usepackage{color}
\newtheorem{lemma}{Lemma}[section]
\newtheorem{theorem}[lemma]{Theorem}
\newtheorem{definiton}{Definition}[section]
\newtheorem{remark}{Remark}[section]
\newtheorem{corollary}{Corollary}[section]
\numberwithin{equation}{section} 
\title{\textbf{Zeta functions and regularized determinants related to the Selberg trace formula}}
\author{\begin{Large}Arash Momeni$^1$ and Alexei Venkov$^2$\end{Large}\\ \begin{small}
$^1$ Department of Statistical Physics and Nonlinear Dynamics\end{small}\\ \begin{small}Institute of Theoretical Physics, \end{small}\\
\begin{small}Clausthal University of Technology, 38678 Clausthal-Zellerfeld, \end{small}\\ \begin{small}Germany. E-Mail: arash.momeni@tu-clausthal.de \end{small}\\
\begin{small}$^2$Institute for Mathematics and Centre of Quantum Geometry QGM,\end{small}\\ \begin{small}University of Aarhus, 8000 Aarhus C,\end{small}\\ \begin{small}Denmark. E-Mail: venkov@imf.au.dk \end{small}\\}

\begin{document}

\maketitle

\begin{abstract}
For a general Fuchsian group of the first kind with an arbitrary unitary representation we define zeta functions related to the contributions of the identity, hyperbolic, elliptic and parabolic conjugacy classes in Selberg's trace formula. We present Selberg's zeta function in terms of a regularized determinant of the automorphic Laplacian. We also present the zeta function for the identity contribution in terms of a regularized determinant of the Laplacian on the two dimensional sphere. We express the zeta functions for the elliptic and parabolic contributions in terms of certain regularized determinants of one dimensional Schroedinger operator for harmonic oscillator.
We decompose the determinant of the automorphic Laplacian into a product of the determinants where each factor is a determinant representation of a zeta function related to Selberg's trace formula. Then we derive an identity connecting the determinants of the automorphic Laplacians on different Riemannian surfaces related to the arithmetical groups. Finally, by using the Jacquet-Langlands correspondence we connect the determinant of the automorphic Laplacian for the unit group of quaternions to the product of the determinants of the automorphic Laplacians for certain cogruence subgroups.
\end{abstract}
\tableofcontents
\section{Introduction}
Regularized determinants (RD) for Laplace operators on compact Riemannian manifolds have been extensively studied and used in mathematics and theoretical (mathematical) physics already for a long time. In mathematics, it was started in 1971 by Ray and Singer (see \cite{Ray}) where they introduced an analytic torsion for a compact manifold. The regularization of the determinant was done through the analytic continuation of the spectral zeta function which was defined and studied by Minakshisundaram and Pleiyel in 1949 (see \cite{Minakshisundaram}, and also see the paper by Krein \cite{Krein}).

In quantum field theory, RD of differential operators were defined and applied by the outstanding physicists Fock \cite{Fock} in 1937, Schwinger \cite{Schwinger} in 1951, Faddeev \cite{Faddeev-Popov} in 1967, Hawking \cite{Hawking} in 1977, and their students (for more references and interesting comments see also the book of Leon Takhtajan \cite{Takhtajan}).

The results are much more explicit mathematically when we consider compact Riemannian surfaces (see \cite{Hoker, Hoker1, Gilbert,Sarnak, Voros, Guillarmon}). The Selberg zeta function is important in this context (see also \cite{Koyama, Koyama2, Koyama3} 
and Efrat \cite{Efrat}). RD and the Selberg zeta function have some more interesting applications:
\begin{itemize}
\item[1] In the AdS/CFT correspondence \cite{Aros, Diaz}.
\item[2] To derive a semiclassical approximation to topological quantum invariants coming from the Chern-Simons action on 3d hyperbolic manifolds \cite{Freed}.
\item[3] To study a Ricci flow on non-compact surfaces \cite{Albin}.
\item[4] In 2+1 quantum gravity \cite{Maschot-Moore}, \cite{DJGV}.
\end{itemize}
Coming back to Riemannian surfaces, Efrat \cite{Efrat} and Koyama \cite{Koyama} considered the much more difficult situation of a surface (or 3 dim hyperbolic manifold) which is non compact but has finite area (volume) (see also \cite{Bytsenko0, Bytsenko, Bytsenko2, Bytsenko3, Deitmar, Kurokawa}). We consider manifolds with constant negative sectional curvature. In that case the corresponding Laplace operator has also a continuous spectrum with finite multiplicity. We still can define a regularized determinant in that situation on the basis of only discrete eigenvalues as Koyama did in strong arithmetical cases. Or, one can take into account also the resonances, the poles of the determinant of the scattering matrix, together with the eigenvalues of the Laplacian as Efrat did. In both cases we have some interesting problems to understand. In the first case if we deform the metric on the surface then generically, according to the Phillips and Sarnak conjecture \cite{Phillips-Sarnak}, the discrete spectrum of the corresponding Laplacian may have only finitely many eigenvalues and the determinant is trivial. In the second case, the set of the resonances is not canonically well defined when we talk about operators which are unitary equivalent to our Laplacian \cite{Balslev-venkov3, Balslev, Balslev-venkov, Balslev-venkov2, Bolte-Steiner, Brocker, Colin1, Colin, Deshouillers, Iwaniec}.

The Selberg zeta function for a Fuchsian group $\Gamma$ of the first kind is defined in the half-plane $\text{Re}(s)>1$ by an absolutely convergent infinite product given by \cite{Selberg}
\begin{equation}
Z(s)=\prod_{k=0}^\infty\prod_{\left\lbrace P\right\rbrace_\Gamma }(1-\mathcal N(P)^{-k-s})
\end{equation}
where $P$ runs over all primitive hyperbolic conjugacy classes in $\Gamma$ and $\mathcal N(P)>1$ denotes the norm of $P$. The Selberg zeta function is related to the contribution of the hyperbolic elements in the Selberg trace formula for a certain test function \cite{Alexei}. The zeta functions $Z_I(s)$, $Z_E(s)$ and $Z_P(s)$ are defined by requiring simillar relations with the identity, elliptic and parabolic contributions in the aformentioned trace formula \cite{Vigneras}, \cite{Koyama}.

In \cite{Koyama}, Koyama defined a regularized determinant $\det(A,s)$ for the automorphic Laplacian $A$ for the congruence subgroups with the trivial representation,
\begin{equation}
\det(A,s):={\det}_D(A-s(1-s)){\det}_C(A,s).
\end{equation}
In this formula ${\det}_D(A-s(1-s))$ is the regularized determinat related to the discrete spectrum of $A$ and ${\det}_C(A,s)$ is the regularized quasi-determinant related to the continuous spectrum of $A$. The RD related to the discrete spectrum is defined by using the zeta regularization method. In this method first, a spectral zeta function is assigned to the increasingly ordered sequence of real eigenvalues of $A$, $\lambda_0\leq\lambda_1\leq\ldots$,
\begin{equation}
\zeta(w,s):=\sum_{n=0}^\infty\dfrac{1}{(\lambda_n-s(1-s))^w},\quad Re(s)\gg0,\quad w\gg0.
\end{equation}
It turns out that for fixed $s$, $Re(s)\gg0$, the spectral zeta function $\zeta(w,s)$ has an analytic continuation to the whole complex $w$-plane as a meromorphic function which is holomorphic at $w=0$.
Then the RD of the Laplacian related to the discrete spectrum is defined by
\begin{equation}
{\det}_D(A-s(1-s)):=\exp(-\dfrac{\partial}{\partial w}\zeta(w,s)\vert_{w=0}).
\end{equation}
The quasi-determinant related to the continuous spectrum of the automorphic Laplacian is defined by
\begin{equation}
\dfrac{d}{ds}C(s)=\dfrac{d}{ds}\dfrac{1}{2s-1}\dfrac{d}{ds}\log {\det}_C(A,s)
\end{equation}
where $C(s)$ is the contribution of the continuous spectrum in the Selberg trace formula with a certain test function. Finally, applying the Selberg trace formula leads to the following identity, up to a nonzero holomorphic factor,
\begin{equation}
\det(A,s)=Z(s)Z_I(s)Z_E(s)Z_P(s).
\end{equation}

In an earlier paper Sarnak \cite{Sarnak} had done simillar work for cocompact groups. Later, Efrat \cite{Efrat} generalized this work to torsion free noncocompact groups twisted with a character $\theta$. In this work the RD of the Laplacian $\det(A-s(1-s))$ is defined by the zeta regularization method where in the spectral zeta function the resonances have also been taken into account. Again applying the Selberg trace formula, one obtains the following identity, up to a nonzero holomorphic factor,
\begin{equation}\label{f098}
\det(A-s(1-s))=(s-\frac{1}{2})^{-K_0}\varphi(s)Z^2(s,\theta)Z^2_I(s,\theta)Z^2_P(s,\theta).
\end{equation}
In this formula $\varphi(s)$ is the determinant of the scattering matrix, $K_0$ is a certain constant, $Z(s,\theta)$, $Z_I(s,\theta)$, and $Z_P(s,\theta)$ are zeta functions twisted with the character $\theta$.

Now we present all results and definitions of our paper in a compact form. For more details we refer to the corresponding lemmas, theorems and definitions. In this paper we consider a general Fuchsian group $\Gamma$ of the first kind twisted with an arbitrary unitary finite dimensional representation $\chi$. 

In subsection \ref{subsec:Selberg's zeta function}, we recall Selberg's zeta function for $\Gamma$ and $\chi$,
\begin{equation}
Z(s;\Gamma;\chi)=\prod_{k=0}^\infty\prod_{\left\lbrace P\right\rbrace_\Gamma }\det(1_V-\chi(P)\mathcal N(P)^{-k-s}).
\end{equation}
In definitions \ref{zi}, \ref{ze} and \ref{zp}, we introduce the zeta functions for the contributions of the identity, elliptic and parabolic elements as follows, 
\begin{equation}
Z_I(s;\Gamma;\chi)=\left((2\pi)^{s}\Gamma^2_2(s)\Gamma(s)^{-1}\right)^{\frac{n\vert F\vert}{2\pi}},
\end{equation}
\begin{equation}
Z_E(s;\Gamma;\chi)=\prod_{\left\lbrace R\right\rbrace _\Gamma}\prod_{l=0}^{\nu-1}\Gamma(\dfrac{s+l}{\nu})^{\frac{-n(\nu-1)+\alpha(R,l)}{\nu}},
\end{equation}
\begin{equation}
Z_P(s;\Gamma;\chi)=e^{-c(n,h)s}2^{-k(\Gamma; \chi)s}(s-\frac{1}{2})^{-\frac{k(\Gamma;\chi)}{2}}\Gamma(s+\frac{1}{2})^{-k(\Gamma;\chi)}.
\end{equation}
In subsection \ref{RdaL}, we define the complete Selberg zeta function as(see formula \eqref{complete Sel zeta})
\begin{equation}
\overset{\sim}{Z}(s;\Gamma;\chi)=Z(s;\Gamma,\chi)Z_I(s;\Gamma,\chi)Z_E(s;\Gamma,\chi)Z_P(s;\Gamma,\chi).
\end{equation}

In Theorem \ref{Had} we recall an extended version of Hadamard's theorem on the factorization of meromorphic functions. As a consequence of this theorem, in Lemma \ref{facSel} we present a factorization for Selberg's zeta function as an absolutely  convergent product given by
\begin{equation}
Z(s;\Gamma;\chi)=e^{Q(s)}(1-2s)^{-K_0}s^{n_0}\dfrac{\prod_{z\in S_3}(1-\frac{s}{z})}{\prod_{z\in S_4^-}(1-\frac{s}{z})}\dfrac{\prod_{z\in S_1\cup S_2\cup S_4^+}(1-\frac{s}{z})e^{(\frac{s}{z})+\frac{1}{2}(\frac{s}{z})^2}}{\prod_{z\in S_5}(1-\frac{s}{z})e^{(\frac{s}{z})}}
\end{equation}
In Lemma \ref{facSel-pointwise}, in the special case of congruence groups with a trivial representation, we derive a factorization for Selberg's zeta function $Z(s)$ as a pointwise convergent product given by
\begin{equation}
Z(s)=(1-2s)^{-K_0}s^{n_0}e^{Q(s)}\dfrac{\prod_{z\in S_1\cup S_2\cup S_3}(1-\frac{s}{z})\prod_{z\in S_4^+}(1-\frac{s}{z})e^{(\frac{s}{z})}}{\prod_{z\in S_5}(1-\frac{s}{z})\prod_{z\in S_4^-}(1-\frac{s}{z})}.
\end{equation}

In Theorem \ref{Mayer-transfer-op}, we recall the determinant expression of Selberg's zeta function for a finite index subgroup $\Gamma$ in $PSL(2,\mathbb Z)$ and a unitary representation $\chi$ in terms of Mayer's transfer operator (see formula \ref{maymay}),
\begin{equation}\label{op1}
\det(1-\mathcal L_s^{\Gamma,\chi})=Z(s;\Gamma;\chi).
\end{equation}

In Lemma \ref{splp} we express the zeta function for the contribution of the identity in terms of a regularized determinant of the Laplace operator on the two dimensional sphere,
\begin{equation}\label{op2}
Z_I(s;\Gamma;\chi)=\left[\sqrt{2\pi} \det(L_2+s)\right]^{-\frac{n\vert F\vert}{2\pi}}.
\end{equation} 
In Lemmas \ref{f105} and \ref{f106} we express the zeta functions for the elliptic and parabolic contributions in terms of regularized determinants of the Schroedinger operator for the harmonic oscillator respectively given by
\begin{equation}\label{op3}
Z_E(s;\Gamma;\chi):=\prod_{\left\lbrace R\right\rbrace _\Gamma}\prod_{l=0}^{\nu-1}\left( (2\pi)^{-\frac{1}{2}}\det(H_1+\dfrac{s+l}{\nu_R})\right)^{-\frac{-n(\nu-1)+\alpha(R,l)}{\nu}},
\end{equation}
\begin{equation}\label{op4}
Z_P(s;\Gamma;\chi):=e^{c(n,h)s}2^{-k(\Gamma; \chi)s}(s-\frac{1}{2})^{-\frac{k(\Gamma;\chi)}{2}}\det(H_1+s+\dfrac{1}{2})^{k(\Gamma;\chi)}.
\end{equation}

In section \ref{RdaL}, as a generalization of Efrat's work \cite{Efrat}, we define the regularized determinant of the automorphic Laplacian $A(\Gamma,\chi)$ by using the zeta regularization method. Then in Lemma \ref{hammer38}, we derive a generalized version of identity (\ref{f098}), namely
\begin{equation}\label{f198}
\det(A(\Gamma,\chi)-s(1-s))=e^{c+c's(s-1)}(s-\frac{1}{2})^{-K_0}\varphi(s)\overset{\sim}{Z}^2(s;\Gamma;\chi).
\end{equation}
In Corollary \ref{fun-eq-cs}, as a direct consequence of this identity, we come to the following functional equation for the complete Selberg zeta function,
\begin{equation}
\overset{\sim}{Z}(1-s;\Gamma;\chi)=\exp(\frac{-i\pi K_0}{2})\varphi(s)\overset{\sim}{Z}(s;\Gamma;\chi).
\end{equation}
For finite index subgroups $\Gamma\subset PSL(2,\mathbb Z)$, by inserting the determinant representations given in (\ref{op1}-\ref{op4}) into (\ref{f198}), up to a nonzero holomorphic factor, we get
\begin{eqnarray}\label{f24k}&&
\det(A-s(1-s))=(s-\dfrac{1}{2})^{-K_0-k(\Gamma;\chi)}\det\Phi(s)\det(L_2+s)^{-\frac{n\vert F\vert}{\pi}}\nonumber\\&&\prod_{\left\lbrace R\right\rbrace _\Gamma}\prod_{l=0}^{\nu-1} \det(H_1+\dfrac{s+l}{\nu_R})^{e(\Gamma;\chi)}\det(H_1+s+\dfrac{1}{2})^{2k(\Gamma;\chi)}\det(1-\mathcal L_s^{\Gamma,\chi})^2\nonumber\\&&
\end{eqnarray}
which is formulated in Theorem \ref{goaway}.

In \cite{Faddeev}, Faddeev introduced a compact operator on certain Banach spaces and used it for analytic continuation of the resolvent of the automorphic Laplacian to the whole complex plane. In a soon coming paper, we prove that for a Fuchsian group $\Gamma$ of the first kind with a unitary representation $\chi$, a generalized version of this operator denoted by $\mathcal H(s;\Gamma;\chi)$ fulfills the following identity, up to a nonzero holomorphic factor,
\begin{equation}\label{FF13}
\det(1-\mathcal H(s;\Gamma;\chi))=\overset{\sim}{Z}(s;\Gamma;\chi)
\end{equation}
where $\det$ denotes certain regularized determinant. For more details about the operator $\mathcal H(s;\Gamma;\chi)$ see \cite{Alexei2}. From (\ref{FF13}) and (\ref{f198}), up to a nonzero holomorphic factor, we get
\begin{equation}\label{f23k}
\det(A-s(1-s))=(s-\dfrac{1}{2})^{-K_0}\det\Phi(s)\det(1-\mathcal H(s;\Gamma;\chi))^2
\end{equation}
which is also formulated in Theorem \ref{goaway}. The factorizations (\ref{f24k}) and (\ref{f23k}) in Theorem \ref{goaway} are the first part of our main results.

In Theorem \ref{koli76} we derive an identity connecting certain regularized determinants defined on different hyperbolic surfaces. This is our second main result:
\begin{eqnarray}
&&\prod_{\psi_1\in(\Gamma_3\setminus\Gamma_1)^*}f(\Gamma_1;\psi_1;s)\det(A(\Gamma_1;\psi_1)-s(1-s))\times\nonumber\\&&{\det}^{-1}
\Phi(\Gamma_1;\psi_1;s)
\left[ \det(L_2+s)\right]^{\frac{\dim \psi_1\vert F\vert}{\pi}}\times\nonumber\\&&\det(H_1+s+\dfrac{1}{2})^{-2k(\Gamma_1;\psi_1)}\prod_{\left\lbrace R\right\rbrace _{\Gamma_1}}\prod_{l=0}^{\nu-1} \det(H_1+\dfrac{s+l}{\nu_R})^{e(\Gamma_1;\psi_1)}\nonumber\\&&=\prod_{\psi_1\in(\Gamma_3\setminus\Gamma_2)^*}f(\Gamma_2;\psi_2;s)\det(A(\Gamma_2;\psi_2)-s(1-s))\times\nonumber\\&&{\det}^{-1}
\Phi(\Gamma_2;\psi_2;s)
\left[ \det(L_2+s)\right]^{\frac{\dim \psi_2\vert F\vert}{\pi}}\times\nonumber\\&&\det(H_1+s+\dfrac{1}{2})^{-2k(\Gamma_2;\psi_2)}\prod_{\left\lbrace R\right\rbrace _{\Gamma_2}}\prod_{l=0}^{\nu-1} \det(H_1+\dfrac{s+l}{\nu_R})^{e(\Gamma_2;\psi_2)}.
\end{eqnarray}

In the last section, applying the results of Bolte and Johansson \cite{bolte}, we derive a determinant identity, connecting the determinant of the Laplacian for the unit group of quaternion with the tirivial representation to the product of determinants of the Laplacian for certain congruence subgroups with the trivial representation. This identity is presented in Theorem \ref{zzz45} as our last main result,
\begin{equation}
F(s)\det(A(\mathcal O_{max}^1)-s(1-s))=\prod_{m\vert n}\det(A(\Gamma_0(m))-s(1-s))^{\beta(\frac{n}{m})}.
\end{equation}

We would like to thank Dieter Mayer for reading of the text and also for his important remarks.
\section{Zeta functions}
\label{sec:Zeta functions}
\subsection{Selberg's zeta function}
\label{subsec:Selberg's zeta function}
Let $\Gamma$ be a Fuchsian group of the first kind with a unitary representation $\chi$ of degree of nonsingularity $k(\Gamma;\chi)$ (see formula \eqref{ksk}). The Selberg zeta function for $\Gamma$ and $\chi$ is defined in the domain $\text{Re}(s)>1$ by an absolutely convergent infinite product given by (see \cite{Selberg}, \cite{Alexei}) 
\begin{equation}\label{Selber zeta function}
Z(s;\Gamma;\chi)=\prod_{k=0}^\infty\prod_{\left\lbrace P\right\rbrace_\Gamma }\det(1_V-\chi(P)\mathcal N(P)^{-k-s})
\end{equation}
where $P$ runs over all primitive hyperbolic conjugacy classes of $\Gamma$ and $\mathcal N(P)>1$ denotes the norm of $P$. By definition every hyperbolic element $P$ of the group $\Gamma$ can be conjugated, by an element from $PSL(2,\mathbb R)$, to a 2 by 2 matrix of the form
\begin{equation}
\left( \begin{array}{cc}
\rho&0\\
0&\rho^{-1}\\
\end{array}
\right)
\end{equation}
with $\rho>1$ which defines the norm of $P$ by $\mathcal N(P)=\rho^2$. 

The Selberg trace formula (see Theorem \ref{Sel-tr}) provides a huge amount of information about Selberg's zeta function. Let us choose the test function $h$ in the left hand side of the trace formula in (\ref{trace formula}) as follows,
\begin{equation}\label{h1}
h(r^2+\frac{1}{4})=\dfrac{1}{r^2+\frac{1}{4}+s(s-1)}-\dfrac{1}{r^2+\beta^2},\quad (\beta>\frac{1}{2},\quad \Re(s)>1).
\end{equation}
According to equation (\ref{g(u)}), the function $g$ is given by
\begin{equation}
g(u)=\dfrac{1}{2s-1}e^{-(s-\frac{1}{2})\vert u\vert}-\frac{1}{2\beta}e^{-\beta\vert u\vert}.
\end{equation}
Then the following identity holds \cite{Alexei},
\begin{equation}\label{logarithmic derivitive of Z(s)}
\dfrac{d}{ds}H(s;\Gamma;\chi)=\dfrac{d}{ds}\dfrac{1}{2s-1}\dfrac{d}{ds}\log Z(s;\Gamma;\chi)
\end{equation}
where $H(s;\Gamma;\chi)$ , given in \eqref{HHH}, is the contribution of the hyperbolic elements in the trace formula with the test function given in \eqref{h1}. The $s$-dependence of $H(s;\Gamma;\chi)$ comes from the $s$-dependence of the test function.

Based on (\ref{logarithmic derivitive of Z(s)}) it is proved that $Z(s;\Gamma;\chi)$ has an analytic (meromorphic) continuation to the whole complex $s$-plane and satisfies the functional equation (see \cite{Alexei} and the references there),
\begin{equation}\label{Sel1}
Z(1-s;\Gamma;\chi)=\Psi(s;\Gamma;\chi)\varphi(s;\Gamma;\chi)Z(s;\Gamma;\chi)
\end{equation}
where $\varphi(s;\Gamma;\chi)$ is the determinant of the scattering matrix and $\Psi(s;\Gamma;\chi)$ is a known function. We note that for cocompact groups $\varphi(s;\Gamma;\chi)\equiv 1$. 

The nontrivial zeros of $Z(s;\Gamma;\chi)$ are related to the eigenvalues of the automorphic Laplacian $A(\Gamma;\chi)$ and the resonances \cite{Alexei}. We distinguish the following sets of zeros of Selberg's zeta function:
\begin{itemize}
\item[\textbf{1}] The discrete set $S_1$ is defined to be the set of points $s$ which are the spectral parameters of cusp forms (see \eqref{ldr}). These are located symmetric relative to the real axis on the line $\text{Re}(s)=\frac{1}{2}$ and $s\neq\frac{1}{2}$ or in the interval $(\frac{1}{2},1]$. For  such a point $s$, $s(1-s)$ is an eigenvalue of the automorphic Laplacian $A(\Gamma;\chi)$. The multiplicity of each zero at $s$ is equal to the multiplicity of the corresponding eigenvalue $s(1-s)$ \cite{Phillips-Lax, Alexei}.
\item[\textbf{2}] The discrete set $S_2$ containing the poles of the determinant of the  scaterring matrix $\varphi(s;\Gamma;\chi)$, lying in the half plane $\text{Re}(s)<\frac{1}{2}$. The elements of $S_2$ are called resonances. The multiplicity of a zero at a point $s\in S_2$ is equal to the order of the pole of $\varphi(s;\Gamma;\chi)$ at the point $s$ \cite{Alexei}.
\item[\textbf{3}] The set $S_3$ containing finitely many poles of the determinant of the scattering matrix $\varphi(s;\Gamma;\chi)$ in the interval $(\frac{1}{2},1]$. For a point $s\in S_3$, $s(1-s)$ is an eigenvalue of the automorphic Laplacian. The multiplicity of a zero at $s\in S_3$ is equal to the multiplicity of the eigenvalue $s(1-s)$ \cite{Alexei}.
\item[\textbf{4}] The set $S_4^+$ including negative integers $-j$, $j\in\mathbb N$ with multiplicity $n_j>0$ given by \cite{Alexei}
\begin{equation}\label{n-j-elliptic}
n_j:=\dfrac{\vert F\vert\dim \chi}{\pi}(j+\frac{1}{2})-\sum_{\left\lbrace R\right\rbrace_\Gamma}\sum_{k=1}^{m-1}\dfrac{tr_V\chi(R^k)}{m\sin\frac{k\pi}{m}}\sin(\dfrac{k\pi(2j+1)}{m}).
\end{equation}
We note that for values of $j$ for which $n_j>0$, $n_j$ is the multiplicity of the zero of Selberg's zeta function at $s=-j$ whereas for $n_j<0$ the point $s=-j$ is a pole of order $-n_j$.
\end{itemize}
The poles of Selberg's zeta function are the following \cite{Alexei}:
\begin{itemize}
\item[\textbf{1}] The point $s=\frac{1}{2}$ with multiplicity $\frac{1}{2}(k(\Gamma;\chi)- \text{tr} \Phi(\frac{1}{2};\Gamma;\chi))$ where \\$\Phi(s;\Gamma;\chi)$ is the scattering matrix (see formula \ref{sctnat}).
\item[\textbf{2}] The set $S_5$ of trivial poles including the points $(-j+\frac{1}{2})$, $j\in\mathbbm N$ with multiplicity $k(\Gamma;\chi)$.
\item[\textbf{3}] The set $S_4^-$ of finitely many poles at the negative integers $s=-j$ of order $-n_j>0$, given in \eqref{n-j-elliptic}. 
\end{itemize}
The point $s=0$ is a pole of order $-n_0$ if $n_0<0$ and it is a zero of multiplicity $n_0$ if $n_0>0$
\begin{remark}
We assume that all these sets contain repeated elements according to the corresponding multiplicity of a zero or an order of a pole.
\end{remark}

Next we write Selberg's zeta function as a product running over its zeros and poles. This factorization is a consequence of an extended version of Hadamard's theorem. First we recall this theorem from \cite{Titchmarsh}.
\begin{theorem}\label{Had}
Let $f(s)$ be a meromorphic function on $\mathbb C$ of finite order $\rho$ with zeros $a_1,a_2,\ldots$ and poles $b_1,b_2,\ldots$ where the only limiting points of these sequences are at infinity and $f(0)\neq0$.
\begin{itemize}
\item[\textbf{1}] There exist integers $p_1$ and $p_2$ not exceeding $\rho$ such that the products formed with zeros and poles of $f(s)$ respectively given by (see \cite{Titchmarsh} page 284g)
\begin{equation}
P_1(s)=P_1(s,p_1)=\prod_{n=1}^\infty\mathcal P(\dfrac{s}{a_n},p_1)
\end{equation}
and
\begin{equation}
P_2(s)=P_2(s,p_2)=\prod_{n=1}^\infty\mathcal P(\dfrac{s}{b_n},p_2),
\end{equation}
are absolutely convergent for all values of $s$. Here the primary factors $\mathcal P$ are defined by (see \cite{Titchmarsh} page 246).
\begin{equation}
\mathcal P(u,p)=\begin{cases}
(1-u)\exp{\sum_{i=1}^p \dfrac{u^i}{i}},& p\in \mathbb N,\\1-u,& p=0.
\end{cases}
\end{equation}
\item[\textbf{2}] Let $P_1(s)=P_1(s,q_1)$ and $P_2(s)=P_2(s,q_2)$ be canonical that is, $q_1$ and $q_2$ are the minimum values of $p_1$ and $p_2$ for which the corresponding products $P_1(s,p_1)$ and $P_2(s,p_2)$ are absolutely convergent. Let $\rho_1$ and $\rho_2$ denote the exponents of convergence of the zeros of $P_1(s)$ and $P_2(s)$ respectively. That is,
\begin{equation}
\rho_1=\inf\left\lbrace \alpha\vert \sum\vert a_n\vert^{-\alpha}<\infty\right\rbrace 
\end{equation}
and
\begin{equation}
\rho_2=\inf\left\lbrace \alpha\vert \sum\vert b_n\vert^{-\alpha}<\infty\right\rbrace.
\end{equation}
Then (see \cite{Titchmarsh}, page 250)
\begin{itemize}
\item[\textbf{a)}] $q_1=[\rho_1]$ if $\rho_1$ is not an integer.
\item[\textbf{b)}] $q_1=\rho_1$ if $\sum\vert a_n\vert^{-\rho_1}$ is divergent.
\item[\textbf{c)}] $q_1=\rho_1-1$ if $\sum\vert a_n\vert^{-\rho_1}$ is convergent. 
\end{itemize}
and 
\begin{itemize}
\item[\textbf{d)}] $q_2=[\rho_2]$ if $\rho_2$ is not an integer.
\item[\textbf{e)}] $q_2=\rho_2$ if $\sum\vert b_n\vert^{-\rho_2}$ is divergent.
\item[\textbf{f)}] $q_2=\rho_2-1$ if $\sum\vert b_n\vert^{-\rho_2}$ is convergent. 
\end{itemize}
\item[\textbf{3}] The function $f(s)$ has the following representation (see \cite{Titchmarsh}, page 284),
\begin{equation}
f(s)=e^{Q(s)}\dfrac{P_1(s)}{P_2(s)}
\end{equation}
where  $Q(s)$ is a polynomial of order at most $\rho$.
In the case of a function $f$ with a zero at the point $s=0$ of multiplicity $m$, this representation is modified to
\begin{equation}\label{hada-rep}
f(s)=s^{m}e^{Q(s)}\dfrac{P_1(s)}{P_2(s)}.
\end{equation}
\end{itemize}
\end{theorem}
\begin{lemma}\label{facSel}
The Selberg zeta function can be written as the following absolutely convergent product,
\begin{equation}
Z(s;\Gamma;\chi)=e^{Q(s)}(1-2s)^{-K}s^{n_0}\dfrac{\prod_{z\in S_3}(1-\frac{s}{z})}{\prod_{z\in S_4^-}(1-\frac{s}{z})}\dfrac{\prod_{z\in S_1\cup S_2\cup S_4^+}(1-\frac{s}{z})e^{(\frac{s}{z})+\frac{1}{2}(\frac{s}{z})^2}}{\prod_{z\in S_5}(1-\frac{s}{z})e^{(\frac{s}{z})}}
\end{equation}
where $K=\frac{1}{2}(k(\Gamma;\chi)-\text{tr} \Phi(\frac{1}{2};\Gamma;\chi))$, $n_0$ is given in \eqref{n-j-elliptic} with $j=0$, and $Q(s)$ is a polynomial of degree at most two. Moreover, $S_1$, $S_2$, $S_3$, $S_4^+$, $S_4^-$, and $S_5$ are the aforementioned sets of zeros and poles of Selberg's zeta function (see pages 7 and 8). 
\end{lemma}
\begin{proof}
Selberg proved that Selberg's zeta function is a meromorphic function of order $2$ with the aforementioned set of zeros and poles \cite{Selberg}. We remove the finitely many poles of Selberg's zeta function in the set $S_4^-\cup\left\lbrace 1/2\right\rbrace $ by multiplying it by the polynomial $(1-2s)^{K}\prod_{z\in S_4^-}(1-\frac{s}{z})$. 
Note that $S_4^-$ has finitely many elements and the pole $s=1/2$ is of order $K$. 
Then we remove the finitely many zeros of Selberg's zeta function in the set $S_3$ by dividing it by the polynomial $\prod_{z\in S_3}(1-\frac{s}{z})$. 
Therefore we obtain a function
\begin{equation}
f(s):=(1-2s)^{K_0}\dfrac{\prod_{z\in S_4^-}(1-\frac{s}{z})}{\prod_{z\in S_3}(1-\frac{s}{z})}Z(s;\Gamma;\chi)
\end{equation}
with the set of zeros $S_1\cup S_2\cup S_4^+$ and set of poles $S_5$. Multiplying and dividing a function by polynomials dosen't change its order, hence $f(s)$ is of order $2$. 
Thus we can apply the previous theorem to $f(s)$. According to the formula \eqref{hada-rep}, $f(s)$ can be represented as
\begin{equation}
f(s)=s^{n_0}e^{Q(s)}\dfrac{P_1(s)}{P_2(s)}
\end{equation}
where $Q(s)$ is a polynomial of degree at most $2$.
Hence, Selberg's zeta function is represented as
\begin{equation}\label{rep001}
Z(s;\Gamma;\chi)=e^{Q(s)}s^{n_0}(1-2s)^{-K}\dfrac{\prod_{z\in S_3}(1-\frac{s}{z})}{\prod_{z\in S_4^-}(1-\frac{s}{z})}\dfrac{P_1(s)}{P_2(s)}.
\end{equation}
If $n_0>0$ then it is the multiplicity of the zero of $Z(s;\Gamma;\chi)$ at $s=0$ and if $n_0<0$ then $-n_0$ is the order of the pole of $Z(s;\Gamma;\chi)$ at $s=0$. According to the first part of the previous theorem, the products $P_1(s)$ and $P_2(s)$ are given by
\begin{equation}\label{p2450}
P_1(s)=\prod_{a_n\in S_1\cup S_2\cup S_4^+}\mathcal P(\frac{s}{a_n},p_1),\qquad P_2(s)=\prod_{b_n\in S_5}\mathcal P(\frac{s}{b_n},p_2).
\end{equation}
The exponent of convergence of $P_1(s)$ is given by 
\begin{equation}
\rho_1=\inf\left\lbrace \alpha\vert \sum_{a_n\in S_1\cup S_2}\vert a_n\vert^{-\alpha}+\sum_{a_n\in S_4^+}\vert a_n\vert^{-\alpha}<\infty\right\rbrace.
\end{equation}
We are going to investigate the convergence of each sum in the above expression. Let $a_j=\sigma_j+i\gamma_j\in S_1\cup S_2$. Then 
\begin{equation}
\sum_{a_j\in S_1\cup S_2}\vert a_j\vert^{-\alpha}=\sum_{a_j\in S_1\cup S_2}\vert \sigma_j^2+\gamma_j^2\vert^{-\alpha/2}<\sum_{a_j\in S_1\cup S_2}\vert\gamma_j\vert^{2(-\alpha/2)}.
\end{equation}
From \eqref{mnb4} and \eqref{mnb5} we get  $\vert\gamma_j\vert^2\sim j$. Hence
\begin{equation}
\sum_{a_j\in S_1\cup S_2}\vert\gamma_j\vert^{2(-\alpha/2)}\sim\sum_{a_j\in S_1\cup S_2}j^{-\alpha/2}
\end{equation}
which is convergent for $\alpha>2+\delta$, $\forall\delta>0$.
For the sum over $S_4^+$, obviously we have 
\begin{equation}\label{mig1}
\sum_{a_n\in S_4^+}\vert a_n\vert^{-2-\delta}<\infty,\qquad \forall\delta>0.
\end{equation}
Thus $\rho_1=2$.
The sum over $S_4^+$ in \eqref{mig1} is not convergent for $\delta=0$, for this reason the sum 
\begin{equation}
\sum_{a_n\in S_1\cup S_2\cup S_4^+}\vert a_n\vert^{-\alpha}
\end{equation} 
is not convergent for $\alpha=2$.
Hence, according to part 2b of the previous theorem the product $P_1(s)$ in \eqref{p2450} becomes a canonical product with the choice $p_1=2$.
The convergent exponent of $P_2(s)$ is determined by the elements of the set $S_5$,
\begin{equation}
\rho_2=\inf\left\lbrace \alpha\vert \sum_{a_n\in S_5}\vert a_n\vert^{-\alpha}=k\sum_{n\in \mathbb N}\vert -n+\frac{1}{2}\vert^{-\alpha}\right\rbrace<\infty. 
\end{equation}
The series 
\begin{equation}
\sum_{n\in \mathbb N}\frac{1}{\vert -n+\frac{1}{2}\vert^{1+\delta}}
\end{equation}
is convergent for all $\delta>0$, hence $\rho_1=1$. Since this series is divegent for $\delta=0$, according to part 2e of the previous theorem for $p_2=1$ the product $P_2(s)$ is canonical. 
By inserting the canonical products $P_1(s)$ and $P_2(s)$ given in \eqref{p2450} with $p_1=2$ and $p_2=1$ respectively into formula \eqref{rep001} we get the desired factorization.
\end{proof}
\begin{lemma}\label{facSel-pointwise}
Let $Z(s)$ be Selberg's zeta function for a congruence subgroup with the trivial representation $\chi=1$. Then $Z(s)$ can be factorized into a conditional convergent infinite product given by
\begin{equation}
Z(s)=(1-2s)^{-K}s^{n_0}e^{Q(s)}\dfrac{\prod_{z\in S_1\cup S_2\cup S_3}(1-\frac{s}{z})\prod_{z\in S_4^+}(1-\frac{s}{z})e^{(\frac{s}{z})}}{\prod_{z\in S_5}(1-\frac{s}{z})\prod_{z\in S_4^-}(1-\frac{s}{z})}.
\end{equation}
Here we used the same notations as in the previous lemma.
\end{lemma}
\begin{proof}
We note that the product over $S_4^+$ is absolutely convergent because of the exponential factor.
Also the products over $S_3$ and $S_4^-$ are obviously absolutely convergent because these sets contain finitely many elements. 
Therefore, it is enough to investigate the convergence of the product,
\begin{equation}
Z_1(s):=\dfrac{\prod_{s_j\in S_1\cup S_2}(1-\frac{s}{s_j})
}{\prod_{s_j\in S_5}(1-\frac{s}{s_j})}.
\end{equation}
To this end, we take logarithmic derivatives of both sides of $Z_1(s)$,
\begin{equation}
\dfrac{d}{ds}\log Z_1(s)=\sum_{s_j\in S_1\cup S_2}\dfrac{1}{s-s_j}-\sum_{s_j\in S_5}\dfrac{1}{s-s_j}.
\end{equation}
To prove the convergence of $Z_1(s)$, one must prove the convergence of this series. 
For $s_j\in S_5$, $s_j=-j+\frac{1}{2}$ with multiplicity $k=k(\Gamma;\chi)$ and therefore,
\begin{equation}\label{delbar}
\dfrac{d}{ds}\log Z_1(s)=\sum_{s_j\in S_1\cup S_2}\dfrac{1}{s-s_j}-\sum_{j\in\mathbbm N}\dfrac{k}{s+j-\frac{1}{2}}.
\end{equation}
Let $S'$ be the set such that $S'\cup\overline{S'}=S_1\cup S_2$ where $\overline{S'}$ denotes the set of complex conjugate elements of $S'$ . Then (\ref{delbar}) can be written as,
\begin{equation}
\dfrac{d}{ds}\log Z_1(s)=\sum_{s_j\in S'}(\dfrac{1}{s-s_j}+\dfrac{1}{s-\overline{s}_j})-\sum_{j\in\mathbbm N}\dfrac{k}{s+j-\frac{1}{2}}
\end{equation}
or
\begin{equation}\label{daleh}
\dfrac{d}{ds}\log Z_1(s)=\sum_{j\in \mathbb N}\dfrac{2s-2\sigma_j}{(s-\sigma_j)^2+\gamma_j^2}-\sum_{j\in\mathbb N}\dfrac{k}{s+j-\frac{1}{2}}
\end{equation}
where $s_j=\sigma_j+i\gamma_j$.
By adding and subtracting 
\begin{equation}
\dfrac{2s-2\sigma_j}{1+\gamma_j^2}
\end{equation}
from the term in the first sum, we can extract its divergent part. Hence the right hand side of (\ref{daleh}) can be written as,
\begin{equation}\label{kalkate}
\sum_{j\in \mathbb N}\dfrac{(2s-2\sigma_j)(1-(s-\sigma_j)^2)}{(1+\gamma_j^2)((s-\sigma_j)^2+\gamma_j^2)}+\sum_{j\in \mathbb N}\dfrac{2s-2\sigma_j}{1+\gamma_j^2}-\sum_{j\in\mathbb N}\dfrac{k}{s+j-\frac{1}{2}}.
\end{equation}
According to the Weyl-Selberg asymptotic formula (see formulas \eqref{asn}, \eqref{asm} in Appendix), for large $j$ the following estimates hold,
\begin{equation}\label{kom}
\gamma_j^2\underset{j\rightarrow\infty}{\sim} j\quad\text{for $s_j=\sigma_j+i\gamma_j$ an eigenvalue}
\end{equation}
and 
\begin{equation}\label{kom2}
\gamma_j\log\gamma_j\underset{j\rightarrow\infty}{\sim} j\quad\text{for $s_j=\sigma_j+i\gamma_j$ a resonance}.
\end{equation}
Thus the first sum and also the part of the second sum running over the resonances in (\ref{kalkate}) are absolutely convergent.
Note that to conclude the later we also used the fact that for each congruence group, $\sigma_j$ is bounded from the left \cite{Selberg}????.
Therefore it remains to prove the convergence of the following term,
\begin{equation}\label{popo}
\mathfrak S:={\sum}'_{j\in \mathbb N}\dfrac{2s-2\sigma_j}{1+\gamma_j^2}-\sum_{j\in\mathbb N}\dfrac{k}{s+j-\frac{1}{2}}
\end{equation}
where $\sum'$ means the sum runs over the eigenvalues $s_j=\sigma_j+i\gamma_j$. Since $\sigma_j=\frac{1}{2}$, the sum in (\ref{popo}) reduces to
\begin{equation}
\mathfrak S:=\sum_{j=1}^\infty\dfrac{2s-1}{1+\gamma_j^2}-\sum_{j=1}^\infty\dfrac{k}{s+j-\frac{1}{2}}.
\end{equation}
According to the estimates in (\ref{kom}) the convergence of $\mathfrak S$ is equivalent to the convergence of the sum,
\begin{equation}
\mathfrak S':=\sum_{j=1}^\infty\dfrac{2s-1}{1+j}-\sum_{j=1}^\infty\dfrac{k}{s+j-\frac{1}{2}}.
\end{equation}

We regularize this sum by considering it as a limit of a suitable sequence of sums,
\begin{equation}
\mathfrak S'=\lim_{N\rightarrow\infty}\left(\sum_{j=1}^ N\dfrac{2s-1}{1+j}-\sum_{j=1}^{t(N)}\dfrac{k}{s+j-\frac{1}{2}}\right).
\end{equation}
Adding and removing $\frac{2s-1}{j}$ and $\frac{k}{j}$ from the first and the second sum respectively, we get
\begin{eqnarray}
&&\mathfrak S'=\lim_{N\rightarrow\infty}\left(-\sum_{j=1}^ N\dfrac{2s-1}{j(1+j)}+\sum_{j=1}^{t(N)}\dfrac{k(s-\frac{1}{2})}{(s+j-\frac{1}{2})j}\right)+\nonumber\\&&\lim_{N\rightarrow\infty}\left(\sum_{j=1}^ N\dfrac{2s-1}{j}-\sum_{j=1}^ {t(N)}\dfrac{k}{j}\right).
\end{eqnarray}
The first limit is convergent. The second limit above is convergent if we put
\begin{equation}
t(N)=N^{\dfrac{2s-1}{k}}
\end{equation}
where to obtain $t(N)$ we used the asymptotic behavior for the harmonic series
\begin{equation}
\sum_{j=1}^t\dfrac{1}{j}\sim \ln(t).
\end{equation}
\end{proof}
We recall the determinant expression of Selberg's zeta function in terms of Mayer's transfer operator.
Let $PSL(2,\mathbb Z)$ be the projective modular group which is generated by the elements
\begin{equation}
S=\left(\begin{array}{cc}
0&-1\\
1&0
\end{array}\right),\quad
T=\left(\begin{array}{cc}
1&1\\
0&1
\end{array}\right).
\end{equation}

Let $\Gamma$ be a subgroup of $PSL(2,\mathbb Z)$ of finite index $\mu=\left[PSL(2,\mathbb Z):\Gamma \right]<\infty $ and $U^{\chi_\Gamma}$ be the representation of $PSL(2,\mathbb Z)$ induced from the representation $\chi$ of $\Gamma$. More precisely, for $g\in PSL(2,\mathbb Z)$,
\begin{equation}
U^{\chi_\Gamma}(g)_{i,j}=\omega_\Gamma(r_igr_j^{-1})
\end{equation}
where
\begin{equation}
\omega_\Gamma(g)=\left\{ \begin{array}{ll}\chi(g)&g\in \Gamma\\0&g\not\in\Gamma
\end{array}\right.
\end{equation}
and $\left\lbrace r_i\right\rbrace_{i=1}^\mu$ is a set of right coset representatives of $\Gamma$ in $PSL(2,\mathbb Z)$.

Let
\begin{equation}
D=\left\lbrace z\in\mathbb C \ \vert \ \vert z-1\vert<\frac{3}{2}\right\rbrace
\end{equation}
and let $B(D)$ denote the Banach space 
\begin{equation}
B(D)=\left\lbrace f:D\rightarrow\mathbb C \, \vert \,f \text{
holomorphic on $D$ and continuous on}\, \overline{D}\right\rbrace
\end{equation}
with the sup norm $\Vert f\Vert=\sup_{z\in\overline{D}}\vert f\vert$.

Mayer's transfer operator for the group $\Gamma$ with a unitary representation $\chi$, in the domain $Re(s)>\frac{1}{2}$, is defined by \cite{Chang}
\begin{eqnarray}\label{maymay}
&\mathcal L_s^{\Gamma,\chi}:\bigoplus_{i=1}^{2\mu\dim\chi} B(D)\rightarrow\bigoplus_{i=1}^{2\mu\dim\chi} B(D)&\nonumber\\
&\mathcal L_s^{\Gamma,\chi}\vec f(z,\epsilon)=\sum_{n=1}^\infty(\dfrac{1}{z+n})^{2s}U^{\chi_\Gamma}(ST^{n\epsilon})\vec f(\dfrac{1}{z+n},-\epsilon),\qquad\epsilon\in\left\lbrace\pm1\right\rbrace.&
\end{eqnarray}
It turns out that Mayer's transfer operator is a nuclear operator of order zero and its Fredholm determinant is defined in the sense of Grothendieck by \cite{Chang}
\begin{equation}\label{Fredholm determinant}
\det(1-\mathcal L_s^{\Gamma,\chi})=\exp\left( -\sum_{n=1}^\infty\dfrac{\text{tr}{\left[ \mathcal L_s^{\Gamma,\chi}\right]}^n}{n}\right).
\end{equation}

There is a close connection between this Fredholm determinant and Selberg's zeta function which is a striking feature of Mayer's transfer operator \cite{Chang}:
\begin{theorem}\label{Mayer-transfer-op}
Let $Z(s;\Gamma;\chi)$ be Selberg's zeta function for a finite index subgroup $\Gamma$ in $PSL(2,\mathbb Z)$ with a unitary representation $\chi$ and $\mathcal L_s^{\Gamma,\chi}$ be Mayer's transfer operator for $\Gamma$ and $\chi$ as defined in (\ref{maymay}). Then the following identity holds, 
\begin{equation}
\det(1-\mathcal L_s^{\Gamma,\chi})=Z(s;\Gamma;\chi),\quad Re(s)>\dfrac{1}{2}
\end{equation}
where the determinant is defined in the sense of Grothendieck in (\ref{Fredholm determinant}) and the analytic continuation of both sides to the whole complex plane coincides. 
\end{theorem}
\subsection{The zeta function for the identity contribution}
In \cite{Vigneras}, Vigneras defined a zeta function related to the identity contribution in Selberg's trace formula for modular group with the trivial representation. In the general case of a Fuchsian group $\Gamma$ of the first kind with a unitary representation $\chi$, she also suggested a definition for the zeta function for identity but without presenting the calculations leading to this definition. In this section, for the aformentioned general case, we derive the defintion of the zeta function for identity, mentioned by Vigneras.

The Selberg zeta function is related to the contribution of the hyperbolic elements in the Selberg trace formula through the differential equation (\ref{logarithmic derivitive of Z(s)}). In the same way a zeta function $Z_I(s;\Gamma;\chi)$ corresponding to the identity contribution $I$ in the Selberg trace formula can be defined by requiring the following differential equation:
\begin{equation}\label{logarithmic derivitive of Z_X1}
\dfrac{d}{ds}I(s;\Gamma;\chi)=\dfrac{d}{ds}\dfrac{1}{2s-1}\dfrac{d}{ds}\log Z_I(s;\Gamma;\chi).
\end{equation}
Here, $I(s;\Gamma;\chi)$ is the contribution of the identity in the Selberg trace formula as given in (\ref{contribution of identity}) with the test function $h(r^2+\frac{1}{4})$ given in (\ref{h1}). 
We note that equation (\ref{logarithmic derivitive of Z_X1}) determines $Z_I(s;\Gamma;\chi)$ up to a factor 
\begin{equation}\label{fac}
e^{c_1s(s-1)+c_2}
\end{equation}
where $c_1$ and $c_2$ are two arbitrary constants.

We are going to find a solution of (\ref{logarithmic derivitive of Z_X1}) for $Re(s)>\frac{1}{2}$. By inserting formula (\ref{h1}) into (\ref{contribution of identity}), we get
\begin{equation}
I(s;\Gamma;\chi)=\dfrac{n\vert F\vert}{4\pi}\int_{-\infty}^\infty r \tanh(\pi r)\left\lbrace \dfrac{1}{r^2+(s-\frac{1}{2})^2}-\dfrac{1}{r^2+\beta^2}\right\rbrace dr
\end{equation}
where $n=\dim\chi$.
This integral can be calculated by using Cauchy's integral formula. Let $N\in\mathbb N$ and 
\begin{equation}\label{integrand}
\mathcal I(s;N):=\dfrac{n\vert F\vert}{4\pi}\int_{C_N}z \tanh(\pi z)\left\lbrace \dfrac{1}{z^2+(s-\frac{1}{2})^2}-\dfrac{1}{z^2+\beta^2}\right\rbrace dz
\end{equation}
where the integration is performed counter clockwise along the path $C_N$, consisting of the interval $[-N,N]$ and the semi-circle $C'_N$ of radius $N$ with center at the origin, lying in the upper half-plane $Re(z)>0$.

In the domain $Re(s)>\frac{1}{2}$, in the upper half-plane $Re(z)>0$, the integrand of the integral (\ref{integrand}),
\begin{equation}
f(z):=z\tanh(\pi z)\left\lbrace \dfrac{1}{z^2+(s-\frac{1}{2})^2}-\dfrac{1}{z^2+\beta^2}\right\rbrace ,
\end{equation}
has only simple poles at $z_k:=i(k+\frac{1}{2})$, $k\in\mathbb N\cup\left\lbrace 0\right\rbrace$ coming from $\tanh (\pi z)$ and simple poles at $z=i(s-\frac{1}{2})$ and $z=i\beta$ coming from the denominaters. Thus according to Cauchy's integral formula we have
\begin{equation}\label{Cauchy90}
\mathcal I(s;N)=\dfrac{n\vert F\vert}{4\pi}\left[ 2\pi i \underset{z=i(s-\frac{1}{2})}{\text{Res}}f(z)+2\pi i \underset{z=i\beta}{\text{Res}}f(z)+2\pi i\sum_{k=0}^{N-1}  \underset{z=z_k}{\text{Res}}f(z)\right] 
\end{equation}
where $\underset{z=u}{\text{Res}}f(z)$ denotes the residue of the function $f(z)$ at the point $z=u$.

Then by a simple calculation we get
\begin{equation}\label{r300}
2\pi i \underset{z=i(s-\frac{1}{2})}{\text{Res}}f(z)=-\pi \tan (\pi (s-\frac{1}{2})),
\end{equation}
\begin{equation}\label{r400}
2\pi i \underset{z=i\beta}{\text{Res}}f(z)=-\pi \tan (\pi \beta)
\end{equation}
and 
\begin{equation}\label{r500}
2\pi i\underset{z=z_k}{\text{Res}}f(z)=-2\left\lbrace  \dfrac{k+\frac{1}{2}}{-(k+\frac{1}{2})^2+(s-\frac{1}{2})^2}-\dfrac{k+\frac{1}{2}}{-(k+\frac{1}{2})^2+\beta^2}\right\rbrace.
\end{equation}

Now by separating the contributions of the segment $[-N,N]$ and the semi-circle $C_N'$ in the integral $\mathcal I(s,N)$, the formula \eqref{Cauchy90} can be written as
\begin{equation}\label{lkj356222}
\begin{split}
&\dfrac{n\vert F\vert}{4\pi}\int_{-N}^N f(z) dz+\dfrac{n\vert F\vert}{4\pi}\int_{C_N'}f(z) dz=\\&\dfrac{n\vert F\vert}{4\pi}\left[ 2\pi i \underset{z=i(s-\frac{1}{2})}{\text{Res}}f(z)+2\pi i \underset{z=i\beta}{\text{Res}}f(z)+2\pi i\sum_{k=0}^{N-1}  \underset{z=z_k}{\text{Res}}f(z)\right] .
\end{split}
\end{equation} 
Since 
\begin{equation}
 f(z)=O(z^{-3})\quad \text{as} \quad\vert z\vert \rightarrow \infty
\end{equation}
the second integral over $C'_N$ vanishes as $N\rightarrow\infty$. Hence, by inserting \eqref{r300}, \eqref{r400}, \eqref{r500} into \eqref{lkj356222} and then taking the limit $N\rightarrow\infty$, we get
\begin{equation}\label{hovo}
I(s,\Gamma;\chi)=-\dfrac{n\vert F\vert}{4\pi}\pi \tan (\pi (s-\frac{1}{2}))-\dfrac{n\vert F\vert}{4\pi}\pi \tan (\pi \beta)+G(s)
\end{equation}
where $G(s)$ is an absolutely convergent sum given by
\begin{equation}
G(s):=-2\dfrac{n\vert F\vert}{4\pi}\sum_{k=0}^\infty \left\lbrace  \dfrac{k+\frac{1}{2}}{-(k+\frac{1}{2})^2+(s-\frac{1}{2})^2}-\dfrac{k+\frac{1}{2}}{-(k+\frac{1}{2})^2+\beta^2}\right\rbrace.
\end{equation}

By inserting (\ref{hovo}) into (\ref{logarithmic derivitive of Z_X1}), we get
\begin{equation}
-\dfrac{n\vert F\vert}{4\pi}\pi \dfrac{d}{ds}\tan (\pi (s-\frac{1}{2}))+\dfrac{d}{ds}G(s)=\dfrac{d}{ds}\dfrac{1}{2s-1}\dfrac{d}{ds}\log Z_I(s;\Gamma;\chi).
\end{equation}
Then double integration of both sides of this identity leads to
\begin{equation}\label{dom1}
\log Z_I(s;\Gamma;\chi)=-\dfrac{n\vert F\vert}{4\pi}\int_{\frac{1}{2}}^{s}(2s'-1)\pi \tan (\pi (s'-\frac{1}{2}))ds'+F(s)+c_1(s-\frac{1}{2})^2+c_2
\end{equation}
where
\begin{equation}
F(s)=\int_{\frac{1}{2}}^{s}(2s'-1)G(s')ds'
\end{equation}
and $c_1$, $c_2$ are some constants. 
Since $G(s)=G(1-s)$, it can be easily verified that 
\begin{equation}\label{F-F}
F(s)=F(1-s).
\end{equation}
Next, by changing the variable, $s'\rightarrow s'+\frac{1}{2}$, we get
\begin{equation}
\log Z_I(s;\Gamma;\chi)=-\dfrac{n\vert F\vert}{2\pi}\int_{0}^{s-\frac{1}{2}}\pi s'\tan (\pi s')ds'+F(s)+c_1(s-\frac{1}{2})^2+c_2.
\end{equation}

Now changing $s\rightarrow 1-s$, we get
\begin{equation}
\log Z_I(1-s;\Gamma;\chi)=\dfrac{n\vert F\vert}{2\pi}\int_{0}^{s-\frac{1}{2}}\pi s'\tan (\pi s')ds'+F(s)+c_1(s-\frac{1}{2})^2+c_2
\end{equation}
where we used (\ref{F-F}).
Hence,
\begin{equation}\label{balghoor}
\dfrac{Z_I(s;\Gamma;\chi)}{Z_I(1-s;\Gamma;\chi)}=\exp\left( -\dfrac{n\vert F\vert}{\pi}\int_{0}^{s-\frac{1}{2}}\pi s'\tan (\pi s')ds'\right). 
\end{equation}
\begin{lemma}
The following expression holds,
\begin{equation}\label{k987}
\int_{0}^{s-\frac{1}{2}}\pi s'\tan (\pi s')ds'=-(s-\dfrac{1}{2})\log 2\pi+\dfrac{1}{2}\log\dfrac{\Gamma(s)}{\Gamma(1-s)}+\log\dfrac{\Gamma_2(1-s)}{\Gamma_2(s)}
\end{equation}
where $\Gamma_2(s)$ denotes the double Barnes gamma function.
\end{lemma}
\begin{proof}
To prove this assertion we need the Kinkelin formula (see \cite{Vigneras}, page 239),
\begin{equation}\label{eqq7000}
\int_0^s \pi x\cot\pi x dx=s\log2\pi-\log\dfrac{\Gamma_2(1-s)}{\Gamma_2(1+s)}.
\end{equation}
The double Barnes gamma function satisfies the following identity (see \cite{Handbook}, page 144),
\begin{equation}
\dfrac{1}{\Gamma_2(1+s)}=\dfrac{\Gamma(s)}{\Gamma_2(s)}.
\end{equation}
Inserting this into \eqref{eqq7000} we get
\begin{equation}\label{a7896}
\int_0^s \pi x\cot\pi x dx=s\log2\pi-\log\Gamma(s)-\log\dfrac{\Gamma_2(1-s)}{\Gamma_2(s)}.
\end{equation}
On the other hand the integral above can be written as
\begin{equation}\label{dso31}
\int_0^s \pi x\cot\pi x dx=\int_0^{\frac{1}{2}} \pi x\cot\pi x dx+\int_{\frac{1}{2}}^s \pi x\cot\pi x dx.
\end{equation}
But we have (see \cite{Gradshteyn}, page 417)
\begin{equation}
\int_0^{\frac{1}{2}} \pi x\cot\pi x dx=\dfrac{\log 2}{2},
\end{equation}
hence \eqref{dso31} reduces to
\begin{equation}
\int_0^s \pi x\cot\pi x dx=\dfrac{\log 2}{2}+\int_{\frac{1}{2}}^s \pi x\cot\pi x dx.
\end{equation}
By variable changing $x\rightarrow x+\frac{1}{2}$ in the integral in the right hand side of the formula above, we get 
\begin{equation}
\int_0^s \pi x\cot\pi x dx=\dfrac{\log 2}{2}-\int_0^{s-\frac{1}{2}}\pi (x+\frac{1}{2})\tan\pi x dx
\end{equation}
or
\begin{equation}\label{j90000}
\int_0^s \pi x\cot\pi x dx=\dfrac{\log 2}{2}-\int_0^{s-\frac{1}{2}}\frac{\pi}{2} \tan\pi x dx-\int_0^{s-\frac{1}{2}}\pi x\tan\pi x dx.
\end{equation}
But (see \cite{Handbook}, page 122 and page 138)
\begin{equation}
-\int_0^{s-\frac{1}{2}}\frac{\pi}{2} \tan\pi x dx=\dfrac{1}{2}\log \sin \pi s=\dfrac{1}{2}\log \dfrac{\pi}{\Gamma(1-s)\Gamma(s)},
\end{equation}
hence \eqref{j90000} reduces to 
\begin{equation}
\int_0^s \pi x\cot\pi x dx=\dfrac{\log 2\pi}{2}+\dfrac{1}{2}\log \dfrac{1}{\Gamma(1-s)\Gamma(s)}-\int_0^{s-\frac{1}{2}}\pi x\tan\pi x dx.
\end{equation}
Now inserting this into \eqref{a7896} we get
\begin{gather}
\int_0^{s-\frac{1}{2}}\pi x\tan\pi x dx=\notag\\\dfrac{\log 2\pi}{2}+\dfrac{1}{2}\log \dfrac{1}{\Gamma(1-s)\Gamma(s)}-s\log2\pi+\log\Gamma(s)+\log\dfrac{\Gamma_2(1-s)}{\Gamma_2(s)}
\end{gather}
which leads to the desired result by some simple algebraic manipulations.
\end{proof}
Inserting \eqref{k987} into (\ref{balghoor}), we get
\begin{equation}
\dfrac{Z_I(s;\Gamma;\chi)}{Z_I(1-s;\Gamma;\chi)}=\left[ (2\pi)^{2s-1}\dfrac{\Gamma(1-s)}{\Gamma(s)}\dfrac{\Gamma_2^2(s)}{\Gamma_2^2(1-s)}\right] ^{\dfrac{n\vert F\vert}{2\pi}}
\end{equation}
or
\begin{equation}
\dfrac{Z_I(s;\Gamma;\chi)}{Z_I(1-s;\Gamma;\chi)}=\dfrac{\left( (2\pi)^{s}\Gamma_2(s)^2\Gamma(s)^{-1}\right)^{\frac{n\vert F\vert}{2\pi}} }{\left( (2\pi)^{1-s}\Gamma_2^2(1-s)\Gamma(1-s)^{-1}\right) ^{\frac{n\vert F\vert}{2\pi}}}.
\end{equation}
On the basis of arithmetical motivations, this suggests the following choice for $Z_I(s;\Gamma;\chi)$,
\begin{definiton}\label{zi}
The zeta function for the contribution of the identity element in the Selberg trace formula for a Fuchsian group $\Gamma$ of the first kind with an $n$-dimensional unitary representation $\chi$ is given by
\begin{equation}\label{dodod}
Z_I(s;\Gamma;\chi)=\left((2\pi)^{s}\Gamma^2_2(s)\Gamma(s)^{-1}\right)^{\frac{n\vert F\vert}{2\pi}} 
\end{equation}
where $\vert F\vert$ denotes the measure of the fundamental domain of $\Gamma$ and $\Gamma_2(s)$ is the double Barnes gamma function. 
\end{definiton}
\subsection{The zeta function for the elliptic contribution}
In \cite{Fischer}, Fischer considered the Laplacian for automorphic forms with multiplier system. 
Then he defined a zeta function related to the contribution of elliptic elements in the corresponding resolvent trace formula. 
In the case of weight zero for congruence groups with a trivial representation, Koyama defined the zeta functions for elliptic elements by using the original Selberg trace formula \cite{Koyama}.
In this section we generalize Koyama's approach to the case of a Fuchsian group of the first kind with an arbitrary unitary representation.

The zeta function $Z_E(s;\Gamma;\chi)$ corresponding to the elliptic contribution $E$ in the Selberg trace formula is defined as a solution of the following differential equation,
\begin{equation}\label{elel}
\dfrac{d}{ds}E(s;\Gamma;\chi)=\dfrac{d}{ds}\dfrac{1}{2s-1}\dfrac{d}{ds}\log Z_E(s;\Gamma;\chi)
\end{equation}
where 
\begin{equation}\label{elel2}
\begin{split}
E(s;\Gamma;\chi)=\dfrac{1}{2}\sum_{\left\lbrace R\right\rbrace _\Gamma}\sum_{m=1}^{\nu-1}\dfrac{\text{tr}_V\chi^m(R)}{\nu \sin\pi m/\nu}\int_{-\infty}^\infty\dfrac{\exp(-2\pi rm/\nu)}{1+\exp(-2\pi r)}\times\\\left( \dfrac{1}{r^2+ (s-\frac{1}{2})^2}-\dfrac{1}{r^2+\beta^2}\right) dr
\end{split}
\end{equation}
is the contribution of the elliptic elements given by (\ref{contribution of elliptic}) with the test function $h(r^2+\frac{1}{4};s)$ as in (\ref{h1}). We consider a solution of \eqref{elel} satisfying,
\begin{equation}\label{elel3}
\dfrac{d}{ds}\log Z_E(s;\Gamma;\chi)=(2s-1)\Omega(s;\Gamma;\chi)
\end{equation}
where
\begin{equation}\label{elel5}
\Omega(s;\Gamma;\chi)=\dfrac{1}{2}\sum_{\left\lbrace R\right\rbrace _\Gamma}\sum_{m=1}^{\nu-1}\dfrac{\text{tr}_V\chi^m(R)}{\nu\sin\pi m/\nu}\int_{-\infty}^\infty\dfrac{\exp(-2\pi rm/\nu)}{1+\exp(-2\pi r)}\dfrac{1}{r^2+ (s-\frac{1}{2})^2} dr
\end{equation}
Purely for technical reasons, instead of identity \eqref{elel3} we consider the following identity
\begin{equation}\label{hijo4}
\dfrac{d}{ds}\log Z_E(s;\Gamma;\chi)-\dfrac{d}{da}\log Z_E(a;\Gamma;\chi)=(2s-1)\Omega(s;\Gamma;\chi)-(2a-1)\Omega(a;\Gamma;\chi)
\end{equation}
where $a>\frac{1}{2}$. We note that according to \eqref{elel3} and \eqref{elel5},
\begin{equation}\label{sos45}
\lim_{a\rightarrow\infty}\dfrac{d}{da}\log Z_E(a;\Gamma;\chi)=\lim_{a\rightarrow\infty}(2a-1)\Omega(a;\Gamma;\chi)=0,
\end{equation}
hence \eqref{hijo4} reduces to \eqref{elel3} in the limit $a\rightarrow\infty$. In the following first we calculate the right hand side of \eqref{hijo4} and then we let $a\rightarrow\infty$.

Thus we have
\begin{equation}
\begin{split}
\dfrac{d}{ds}\log Z_E(s;\Gamma;\chi)-\dfrac{d}{da}\log Z_E(a;\Gamma;\chi)=\\\dfrac{1}{2}\sum_{\left\lbrace R\right\rbrace _\Gamma}\sum_{m=1}^{\nu-1}\dfrac{\text{tr}_V\chi^m(R)}{\nu\sin\pi m/\nu}\int_{-\infty}^\infty\dfrac{\exp(-2\pi rm/\nu)}{1+\exp(-2\pi r)}\dfrac{2s-1}{r^2+(s-\frac{1}{2})^2} dr\\-\dfrac{1}{2}\sum_{\left\lbrace R\right\rbrace _\Gamma}\sum_{m=1}^{\nu-1}\dfrac{\text{tr}_V\chi^m(R)}{\nu\sin\pi m/\nu}\int_{-\infty}^\infty\dfrac{\exp(-2\pi rm/\nu)}{1+\exp(-2\pi r)}\dfrac{2a-1}{r^2+(a-\frac{1}{2})^2} dr
\end{split}
\end{equation}
or
\begin{equation}\label{mil17}
\begin{split}
&\dfrac{d}{ds}\log Z_E(s;\Gamma;\chi)-\dfrac{d}{da}\log Z_E(a;\Gamma;\chi)=\\&-\dfrac{1}{2i}\sum_{\left\lbrace R\right\rbrace _\Gamma}\sum_{m=1}^{\nu-1}\dfrac{\text{tr}_V\chi^m(R)}{\nu \sin\pi m/\nu}\times\\&\int_{-\infty}^\infty \dfrac{\exp(-2\pi rm/\nu)}{1+\exp(-2\pi r)}\left( \dfrac{1}{r+i(s-\frac{1}{2})}-\dfrac{1}{r+i(a-\frac{1}{2})}\right)dr\\&+\dfrac{1}{2i}\sum_{\left\lbrace R\right\rbrace _\Gamma}\sum_{m=1}^{\nu-1}\dfrac{\text{tr}_V\chi^m(R)}{\nu \sin\pi m/\nu}\times\\&\int_{-\infty}^\infty \dfrac{\exp(-2\pi rm/\nu)}{1+\exp(-2\pi r)}\left( \dfrac{1}{r-i(s-\frac{1}{2})}-\dfrac{1}{r-i(a-\frac{1}{2})}\right)dr.
\end{split}
\end{equation}
We calculate these integrals in the domain $Re(s)>\frac{1}{2}$ by using Cauchy's integral formula. Let
\begin{equation}\label{integdd}
\mathcal I_{N,\pm}(s,a,m):=\int_{C_{N,\pm}}f_\pm(z)dz
\end{equation}
where
\begin{equation}
f_\pm(z):=\dfrac{\exp(-2\pi zm/\nu)}{1+\exp(-2\pi z)}\left( \dfrac{1}{z\pm i(s-\frac{1}{2})}-\dfrac{1}{z\pm i(a-\frac{1}{2})}\right),
\end{equation}
$C_{N,-}$ and $C_{N,+}$ denote countours consisting of the line element from $-N$ to $N$ along the real axis and the semi-circle from $N$ to $-N$ lying in the lower and upper half plane, respectively.
\begin{lemma}\label{lemNN}
Let 
\begin{equation}\label{su35}
\begin{split}
&\mathcal E_N:=-\dfrac{1}{2i}\sum_{\left\lbrace R\right\rbrace _\Gamma}\sum_{m=1}^{\nu-1}\dfrac{\text{tr}_V\chi^m(R)}{\nu \sin\pi m/\nu}\mathcal I_{N,+}(s,a,m)\\&+\dfrac{1}{2i}\sum_{\left\lbrace R\right\rbrace _\Gamma}\sum_{m=1}^{\nu-1}\dfrac{\text{tr}_V\chi^m(R)}{\nu \sin\pi m/\nu}\mathcal I_{N,-}(s,a,m).
\end{split}
\end{equation}
Then 
\begin{equation}
\lim_{N\rightarrow\infty}\mathcal E_N=\dfrac{d}{ds}\log Z_E(s;\Gamma;\chi)-\dfrac{d}{da}\log Z_E(a;\Gamma;\chi).
\end{equation}
\end{lemma}
\begin{proof}
First we decompose the integral $\mathcal I_{N,\pm}(s,a,m)$ as the following
\begin{equation}\label{dododod}
\mathcal I_{N,\pm}(s,a,m)=\int_{-N}^Nf_\pm(z)dz+\int_{C'_{N,\pm}} f_\pm(z)dz
\end{equation}
where $C'_{N,+}$, $C'_{N,-}$ denote the semi-circles of radius $N$ centered at the origion lying in the upper and lower half-planes, respectively. 
For the points $z=N\exp(i \theta)$ with $0\leq\theta<2\pi$ lying on $C'_{N,\pm}$, we have
\begin{equation}
\vert f_\pm(z)\vert=O(N^{-2}),\qquad\text{as}\quad N\rightarrow\infty.
\end{equation}
Thus in the limit $N\rightarrow\infty$, the integral over $C'_{N,\pm}$ vanishes. 
Hence, inserting the right hand side of \eqref{dododod} into \eqref{su35} and taking the limit $N\rightarrow\infty$ yields the right hand side of \eqref{mil17} which completes the proof.
\end{proof}
Now we are going to calculate the integrals $\mathcal I_{N,\pm}(s,a,m)$ in \eqref{integdd}. The integrands $f_\pm(z)$ have inside the countours $C_{N,\pm}$ simple poles at the points $z_n=\pm i(n+\frac{1}{2})$, $n=0,1,2,\ldots$, with residues given by
\begin{equation}
\underset{z=z_n}{\text{Res}}f_{\pm}=\pm\dfrac{1}{2\pi i}\exp\left(\mp 2\pi i(n+\frac{1}{2})\dfrac{m}{\nu}\right)\left(\dfrac{1}{s+n}-\dfrac{1}{a+n}\right).
\end{equation}
But according to Cauchy's integral formula, the integrals in \eqref{integdd} are equal to the sum of the residues of the integrands $f_\pm$ at the corresponding poles inside the countours $C_{N,\pm}$ times $\pm2\pi i$.
Therefore we get
\begin{equation}\label{kado}
\mathcal I_{N,\pm}(s,a,m)=\sum_{n=0}^{N-1}\exp\left(\mp 2\pi i(n+\frac{1}{2})\dfrac{m}{\nu}\right)\left(\dfrac{1}{s+n}-\dfrac{1}{a+n}\right).
\end{equation}
where $N-1$ is the number of poles inside the countours $C_{N,\pm}$.
According to Lemma \ref{lemNN}, by inserting \eqref{kado} into \eqref{su35} and then taking the limit $N\rightarrow\infty$, we get
\begin{equation}\label{kado12}
\begin{split}
&\dfrac{d}{ds}\log Z_E(s;\Gamma;\chi)-\dfrac{d}{da}\log Z_E(a;\Gamma;\chi)=\\&-\dfrac{1}{2i}\sum_{\left\lbrace R\right\rbrace _\Gamma}\sum_{m=1}^{\nu-1}\dfrac{\text{tr}_V\chi^m(R)}{\nu \sin\pi m/\nu}\sum_{n=0}^{\infty}\exp\left(-2\pi i(n+\frac{1}{2})\dfrac{m}{\nu}\right)\left(\dfrac{1}{s+n}-\dfrac{1}{a+n}\right)\\&+\dfrac{1}{2i}\sum_{\left\lbrace R\right\rbrace _\Gamma}\sum_{m=1}^{\nu-1}\dfrac{\text{tr}_V\chi^m(R)}{\nu \sin\pi m/\nu}\sum_{n=0}^{\infty}\exp\left(2\pi i(n+\frac{1}{2})\dfrac{m}{\nu}\right)\left(\dfrac{1}{s+n}-\dfrac{1}{a+n}\right)
\end{split}
\end{equation}
Replace $n\rightarrow \nu n+l$ where the new $n$ runs over $\mathbb N\cup\left\lbrace 0\right\rbrace $ and $0\leq l\leq\nu-1$. Then \eqref{kado12} can be written as
\begin{equation}\label{aliyaret}
\begin{split}
&\dfrac{d}{ds}\log Z_E(s;\Gamma;\chi)-\dfrac{d}{da}\log Z_E(a;\Gamma;\chi)=\\&-\sum_{\left\lbrace R\right\rbrace _\Gamma}\dfrac{1}{\nu}\dfrac{1}{2i}\sum_{m=1}^{\nu-1}\dfrac{\text{tr}_V\chi^m(R)}{\nu \sin\pi m/\nu}\times\\&\sum_{l=0}^{\nu-1}\exp\left(-2\pi i(l+\frac{1}{2})\frac{m}{\nu}\right)\sum_{n=0}^{\infty}\left(\dfrac{1}{n+\frac{s+l}{\nu}}-\dfrac{1}{n+\frac{a+l}{\nu}}\right)\\&+\sum_{\left\lbrace R\right\rbrace _\Gamma}\dfrac{1}{\nu}\dfrac{1}{2i}\sum_{m=1}^{\nu-1}\dfrac{\text{tr}_V\chi^m(R)}{\nu \sin\pi m/\nu}\exp\left(2\pi i(l+\frac{1}{2})\dfrac{m}{\nu}\right)\sum_{n=0}^{\infty}\left(\dfrac{1}{n+\frac{s+l}{\nu}}-\dfrac{1}{n+\frac{a+l}{\nu}}\right)
\end{split}
\end{equation}
The di-gamma function $\psi(z)$ is defined by
\begin{equation}
\psi(z)=\sum_{n=0}^\infty\left( \dfrac{1}{n+1}-\dfrac{1}{n+z}\right)-\gamma. 
\end{equation}
where $\gamma$ is Euler's constant.
The sum over $n$ in formula \eqref{aliyaret} can be written in terms of di-gamma function,
\begin{equation}\label{aliyaret2}
\begin{split}
&\dfrac{d}{ds}\log Z_E(s;\Gamma;\chi)-\dfrac{d}{da}\log Z_E(a;\Gamma;\chi)=\\&\sum_{\left\lbrace R\right\rbrace _\Gamma}\dfrac{1}{\nu}\dfrac{1}{2i}\sum_{m=1}^{\nu-1}\dfrac{\text{tr}_V\chi^m(R)}{\nu \sin\pi m/\nu} \sum_{l=0}^{\nu-1}\exp\left(-2\pi i(l+\frac{1}{2})\frac{m}{\nu}\right)\left( \psi(\frac{s+l}{\nu})-\psi(\frac{a+l}{\nu})\right) \\&-\sum_{\left\lbrace R\right\rbrace _\Gamma}\dfrac{1}{\nu}\dfrac{1}{2i}\sum_{m=1}^{\nu-1}\dfrac{\text{tr}_V\chi^m(R)}{\nu \sin\pi m/\nu}\sum_{l=0}^{\nu-1}\exp\left(2\pi i(l+\frac{1}{2})\dfrac{m}{\nu}\right)\left( \psi(\frac{s+l}{\nu})-\psi(\frac{a+l}{\nu})\right).
\end{split}
\end{equation}
\begin{lemma}
The following identity holds,
\begin{equation}
\begin{split}
&\sum_{m=1}^{\nu-1}\dfrac{\text{tr}_V\chi^m(R)}{\nu \sin\pi m/\nu}\exp\left(2\pi i(l+\frac{1}{2})\frac{m}{\nu}\right)=\\&-\sum_{m=1}^{\nu-1}\dfrac{\overline{\text{tr}}_V\chi^m(R)}{\nu \sin\pi m/\nu}\exp\left(-2\pi i(l+\frac{1}{2})\frac{m}{\nu}\right)
\end{split}
\end{equation}
where $\overline{\text{tr}}_V$ denotes the complex conjugate of $\text{tr}_V$.
\end{lemma}
\begin{proof}
It is enough to change $m$ to $\nu-m$ in the sum in the left hand side and to use the the facts that $\chi$ is a unitary representation and $R^{-m}=R^{\nu-m}$.
\end{proof}
Applying this lemma, \eqref{aliyaret2} can be written as
\begin{equation}\label{aliyaret3}
\begin{split}
&\dfrac{d}{ds}\log Z_E(s;\Gamma;\chi)-\dfrac{d}{da}\log Z_E(a;\Gamma;\chi)=\\&\sum_{\left\lbrace R\right\rbrace _\Gamma}\dfrac{1}{\nu}\dfrac{1}{2i}\sum_{m=1}^{\nu-1}\left[ \dfrac{\text{tr}_V\chi^m(R)}{\nu\sin\pi m/\nu}+\dfrac{\text{tr}_V\overline{\chi}^m(R)}{\nu \sin\pi m/\nu}\right]\times\\& \sum_{l=0}^{\nu-1}\exp\left(-2\pi i(l+\frac{1}{2})\frac{m}{\nu}\right)\left( \psi(\frac{s+l}{\nu})-\psi(\frac{a+l}{\nu})\right).
\end{split}
\end{equation}
\begin{lemma}
For $1\leq m<\nu$, the following limit holds,
\begin{equation}
\lim_{a\rightarrow\infty}\sum_{l=0}^{\nu-1}\exp\left(-2\pi i(l+\frac{1}{2})\frac{m}{\nu}\right)\psi(\frac{a+l}{\nu})=0
\end{equation}
\end{lemma}
\begin{proof}
First we show that
\begin{equation}\label{j876}
\sum_{l=0}^{\nu-1}\exp\left(-2\pi i(l+\frac{1}{2})\frac{m}{\nu}\right)=0.
\end{equation}
We have
\begin{equation}
\exp\left(-\pi i\frac{m}{\nu}\right)\sum_{l=0}^{\nu-1}\exp\left(-2\pi i \frac{m}{\nu}l\right)=\exp\left(-\pi i\frac{m}{\nu}\right)\sum_{l=1}^{\nu}\exp\left(-2\pi i \frac{m}{\nu}l\right).
\end{equation}
From the identities (see \cite{Gradshteyn}, page 30),
\begin{equation}
\sum_{l=1}^\nu\sin l x=\sin \dfrac{\nu+1}{2}x\sin\dfrac{\nu x}{2} \text{cosec} \dfrac{x}{2},\qquad\sum_{l=1}^\nu\cos lx=\cos \dfrac{\nu+1}{2}x\sin\dfrac{\nu x}{2} \text{cosec}\dfrac{x}{2}
\end{equation}
with $x=-2\pi i\frac{m}{\nu}$, it follows that
\begin{equation}
\sum_{l=1}^{\nu}\exp\left(-2\pi i \frac{m}{\nu}l\right)=0
\end{equation}
Hence, formula \eqref{j876} is proved. 
According to \cite{Bernardo}, the asymptotic behavior of the digamma function at infinity is given by
\begin{equation}
\psi(z)=\log z+O(\dfrac{1}{z}),\qquad \vert z\vert\rightarrow\infty.
\end{equation}
Hence,
\begin{equation}
\begin{split}
&L:=\lim_{a\rightarrow\infty}\sum_{l=0}^{\nu-1}\exp\left(-2\pi i(l+\frac{1}{2})\frac{m}{\nu}\right)\psi(\frac{a+l}{\nu})=\\&
\lim_{a\rightarrow\infty}\sum_{l=0}^{\nu-1}\exp\left(-2\pi i(l+\frac{1}{2})\frac{m}{\nu}\right)\log (\frac{a+l}{\nu}).
\end{split}
\end{equation}
Since
\begin{equation}
\log (\frac{a+l}{\nu})=\log (a+l)-\log \nu,
\end{equation}
according to \eqref{j876} we get
\begin{equation}\label{gg34}
L=\lim_{a\rightarrow\infty}\sum_{l=0}^{\nu-1}\exp\left(-2\pi i(l+\frac{1}{2})\frac{m}{\nu}\right)\log (a+l).
\end{equation}
On the other hand by inserting 
\begin{equation}
\log(a+l)=\log a+\log(1+\dfrac{l}{a})
\end{equation}
into \eqref{gg34} we get
\begin{equation}
\begin{split}
&L=\lim_{a\rightarrow\infty}\sum_{l=0}^{\nu-1}\exp\left(-2\pi i(l+\frac{1}{2})\frac{m}{\nu}\right)\log (1+\dfrac{l}{a})=\\&\sum_{l=0}^{\nu-1}\exp\left(-2\pi i(l+\frac{1}{2})\frac{m}{\nu}\right)\log (1)=0
\end{split}
\end{equation}
which completes the proof.
\end{proof}
According to this Lemma and formula \eqref{sos45}, by taking the limit of both sides of \eqref{aliyaret3} we get
\begin{equation}\label{ujn}
\begin{split}
&\dfrac{d}{ds}\log Z_E(s;\Gamma;\chi)=\sum_{\left\lbrace R\right\rbrace _\Gamma}-\dfrac{1}{\nu^2}\sum_{m=1}^{\nu-1}\left[ \dfrac{i\text{tr}_V\chi^m(R)}{2\sin\pi m/\nu}+\dfrac{i\text{tr}_V\overline{\chi}^m(R)}{2\sin\pi m/\nu}\right]\times\\& \sum_{l=0}^{\nu-1}\exp\left(-2\pi i(l+\frac{1}{2})\frac{m}{\nu}\right)\psi(\frac{s+l}{\nu}).
\end{split}
\end{equation}
The following lemma has been proved by Fischer (see \cite{Fischer}, page 67).
\begin{lemma}
Let $\exp(-\frac{2\pi i}{\nu}\alpha_{p}(R))$ with $\alpha_{p}(R)\in\left\lbrace 0,\ldots, \nu-1\right\rbrace $, $p=1,\ldots, n$ denote the eigenvalues of $\chi(R)$. Furthermore, let
\begin{equation}
\alpha_{p}(R,l)=(l+\alpha_{p}(R))\mod \nu,\quad \widehat{\alpha}_{p}(R,l)=(l-\alpha_{p}(R))\mod \nu
\end{equation}
where $l\in\left\lbrace 0,1,\ldots,\nu-1\right\rbrace $.
Then the following identities hold,
\begin{equation}
\sum_{m=1}^{\nu-1}\dfrac{i\text{tr}_V\chi(R^m)}{2 \sin\pi m/\nu}\exp\left(-2\pi i(l+\frac{1}{2})\frac{m}{\nu}\right)=\frac{1}{2}n(\nu-1)-\sum_{p=1}^n\alpha_{p}(R,l),
\end{equation}
\begin{equation}
\sum_{m=1}^{\nu-1}\dfrac{i\text{tr}_V\chi(R^m)}{2 \sin\pi m/\nu}\exp\left(2\pi i(l+\frac{1}{2})\frac{m}{\nu}\right)=-\frac{1}{2}n(\nu-1)+\sum_{p=1}^n\widehat{\alpha}_{p}(R,l).
\end{equation}
\end{lemma}
According to this lemma, equation (\ref{ujn}) can be written as
\begin{equation}\label{p13}
\dfrac{d}{ds}\log Z_E(s;\Gamma;\chi)=\sum_{\left\lbrace R\right\rbrace _\Gamma}-\dfrac{1}{\nu^2}\sum_{l=0}^{\nu-1}\left(n(\nu-1)-\alpha(R,l)\right)\psi(\frac{s+l}{\nu})
\end{equation}
where
\begin{equation}\label{fouz}
\alpha(R,l):=\sum_{p=1}^n(\widehat{\alpha}_{p}(R,l)+\alpha_{p}(R,l)).
\end{equation}
Finally, since
\begin{equation}
\psi(\frac{s+l}{\nu})=\nu\frac{d}{ds}\log\Gamma(\frac{s+l}{\nu}),
\end{equation}
formula (\ref{p13}) can be written as
\begin{equation}
\dfrac{d}{ds}\log Z_E(s;\Gamma;\chi)=\sum_{\left\lbrace R\right\rbrace _\Gamma}\sum_{l=0}^{\nu-1}\dfrac{-n(\nu-1)+\alpha(R,l)}{\nu}\frac{d}{ds}\log\Gamma(\frac{s+l}{\nu}).
\end{equation}
This formula gives us a solution of (\ref{elel}) for $Z_E(s;\Gamma;\chi)$. Thus the zeta function for the elliptic contributions can be defined by
\begin{definiton}\label{ze}
The zeta function for the contribution of the elliptic elements in the Selberg trace formula for a Fuchsian group $\Gamma$ of the first kind with an $n$-dimensional unitary representation $\chi$ is given by
\begin{eqnarray}
&&Z_E(s;\Gamma;\chi):\mathbb C\backslash(-\infty,0]\rightarrow\mathbb C,\\&&Z_E(s;\Gamma;\chi):=\prod_{\left\lbrace R\right\rbrace _\Gamma}\prod_{l=0}^{\nu-1}\Gamma(\dfrac{s+l}{\nu})^{\frac{-n(\nu-1)+\alpha(R,l)}{\nu}}.
\end{eqnarray}
Here $R$ denotes a representative of an elliptic conjugacy class of $\Gamma$ of order $\nu$ and $\alpha(R,l)$ is given in (\ref{fouz}).
\end{definiton}
\subsection{The zeta function for the parabolic contribution}
The zeta function for the parabolic contribution $Z_P(s;\Gamma;\chi)$ is defined in the same way as the zeta functions for the other contributions.
According to Theorem \ref{Sel-tr} the contribution of the parabolic elements in the Selberg trace formula for a Fuchsian group $\Gamma$ of the first kind with an $n$-dimensional unitary representation $\chi$ and a test function $h$ is given by
\begin{eqnarray}
&&P(h;\Gamma;\chi)=-\left( k(\Gamma; \chi)\ln2+c(n,h)\right) g(0)\nonumber\\&&
-\dfrac{k(\Gamma;\chi)}{2\pi}\int_{-\infty}^\infty\psi(1+ir)h(r^2+\frac{1}{4})dr+\dfrac{k(\Gamma;\chi)}{4}h(\dfrac{1}{4}).
\end{eqnarray}
where
\begin{equation}\label{ccoocc}
c(n,h):=\sum_{\alpha=1}^h\sum_{l=k_\alpha+1}^n\ln\vert 1-\exp(2\pi i\theta_{\alpha l})\vert.
\end{equation}
For the pair of functions, 
\begin{equation}\label{h111}
h(r^2+\frac{1}{4})=\dfrac{1}{r^2+\frac{1}{4}+s(s-1)}-\dfrac{1}{r^2+\beta^2},\ \ (\beta>\frac{1}{2},\ Re(s)>1).
\end{equation}
and 
\begin{equation}\label{g111}
g(u)=\dfrac{1}{2s-1}e^{-(s-\frac{1}{2})\vert u\vert}-\frac{1}{2\beta}e^{-\beta\vert u\vert}
\end{equation}
this definition is reduced to
\begin{eqnarray}\label{voe}
&&P(s;\Gamma;\chi)=-\left( k(\Gamma; \chi)\ln2+c(n,h)\right)\left[ \dfrac{1}{2s-1}-\frac{1}{2\beta}\right] \nonumber\\&&
-\dfrac{k(\Gamma;\chi)}{2\pi}\int_{-\infty}^\infty\psi(1+ir)\left[ \dfrac{1}{r^2+\frac{1}{4}+s(s-1)}-\dfrac{1}{r^2+\beta^2}\right] dr\nonumber\\&&+\dfrac{k(\Gamma;\chi)}{4}\left[\dfrac{1}{(s-\frac{1}{2})^2}-\dfrac{1}{\beta^2}\right].
\end{eqnarray}

We define the zeta function for parabolic elements as a solution of the differential equation
\begin{equation}\label{pelpel}
\dfrac{d}{ds}P(s;\Gamma;\chi)=\dfrac{d}{ds}\dfrac{1}{2s-1}\dfrac{d}{ds}\log Z_P(s;\Gamma;\chi).
\end{equation}
Koyama solved this equation for $Z_P(s;\Gamma;\chi)$ in the case of the congruence groups with the trivial representation $\chi=1$ \cite{Koyama}.
Assuming $Re(s)>\frac{1}{2}$, we are going to find a solution $Z_P(s;\Gamma;\chi)$ of the differential equation \eqref{pelpel} in the more general case of a Fuchsian group of the first kind $\Gamma$ with an arbitrary unitary representation $\chi$.

Inserting (\ref{voe}) into (\ref{pelpel}) we get
\begin{eqnarray}
&&\dfrac{d}{ds}\dfrac{1}{2s-1}\dfrac{d}{ds}\log Z_P(s;\Gamma;\chi)=-\left( k(\Gamma; \chi)\ln2+c(n,h)\right) \dfrac{d}{ds}\dfrac{1}{2s-1}\nonumber\\&&
-\dfrac{k(\Gamma;\chi)}{2\pi}\int_{-\infty}^\infty\psi(1+ir)\left( \dfrac{d}{ds}\dfrac{1}{r^2+\frac{1}{4}+s(s-1)}\right)dr\nonumber\\&&+\dfrac{k(\Gamma;\chi)}{4}\left( \dfrac{d}{ds}\dfrac{1}{(s-\frac{1}{2})^2}\right).
\end{eqnarray}
By integrating over $s$ we get
\begin{eqnarray}\label{ooo}
&&\dfrac{d}{ds}\log Z_P(s;\Gamma;\chi)=-\left( k(\Gamma; \chi)\ln2+c(n,h)\right)\nonumber\\&&
-(s-\frac{1}{2})\dfrac{k(\Gamma;\chi)}{\pi}\int_{-\infty}^\infty \dfrac{\psi(1+ir)}{r^2+\frac{1}{4}+s(s-1)}dr\nonumber\\&&+\dfrac{k(\Gamma;\chi)}{2}\dfrac{1}{s-\frac{1}{2}}.
\end{eqnarray}
The integral in this formula can be calculated by using Cauchy's integral formula (see \cite{Alexei2}, page 83),
\begin{equation}
\begin{split}
&-\dfrac{k(\Gamma;\chi)}{\pi}\int_{-\infty}^\infty \dfrac{\psi(1+ir)}{r^2+\frac{1}{4}+s(s-1)}dr=\\&-\dfrac{k(\Gamma;\chi)}{s-\frac{1}{2}}\dfrac{d}{ds}\log\Gamma(s+\frac{1}{2})+2k(\Gamma;\chi) \sum_{j=1}^\infty\dfrac{1}{(s-\frac{1}{2})^2-j^2}+c.
\end{split}
\end{equation}
Hence, (\ref{ooo}) can be written as
\begin{eqnarray}
&&\dfrac{d}{ds}\log Z_P(s;\Gamma;\chi)=-\left( k(\Gamma; \chi)\ln2+c(n,h)\right)\nonumber\\&&
-k(\Gamma; \chi)\dfrac{d}{ds}\log\Gamma(s+\frac{1}{2})+k(\Gamma;\chi)\dfrac{d}{ds}\sum_{j=1}^\infty\dfrac{1}{(s-\frac{1}{2})^2-j^2}+c(s-\frac{1}{2})\nonumber\\&&+\dfrac{k(\Gamma;\chi)}{2}\dfrac{1}{s-\frac{1}{2}}.
\end{eqnarray}
This equation determines the function $Z_P(s;\Gamma;\chi)$ up to a holomorphic nonzero factor:
\begin{definiton}\label{zp}
The zeta function for the contribution of the parabolic elements in the Selberg trace formula for a Fuchsian group $\Gamma$ of the first kind with an $n$-dimensional unitary representation $\chi$ on the entire complex plane is given by
\begin{equation}
Z_P(s;\Gamma;\chi):=e^{-c(n,h)s}2^{-k(\Gamma; \chi)s}(s-\frac{1}{2})^{-\frac{k(\Gamma;\chi)}{2}}\Gamma(s+\frac{1}{2})^{-k(\Gamma;\chi)}
\end{equation}
Here $c(n,h)$ is given in (\ref{ccoocc}) and $k(\Gamma;\chi)$ is the degree of non-singularity (see formula \eqref{ksk} in the Appendix) of the representation.
\end{definiton}
In \cite{Fischer}, Fischer defined a zeta function for parabolic elements by using the resolvent trace formula related to the Laplacian for automorphic forms with multiplier system. 
In his definiton, the zeta funtion for parabolic elements is indeed related to the contribution of parabolic elements and also continuous spectrum in the resolvent trace formula 
(see \cite{Fischer}, page 102). 
This is different from our definiton where the zeta function for parabolic elements is related only to the parabolic contribution (see formula \eqref{pelpel}).

\section{Regularized determinants and zeta functions}
\subsection{Determinant expression of the zeta function for the identity}
In this subsection we express the zeta function for the identity contribution for a general Fuchsian group $\Gamma$ of the first kind  with an arbitrary representation $\chi$ in terms of a regularized determinant connected to the Laplace operator on the two dimensional sphere.
For a cocompact group with a trivial representation this has been done by Voros in \cite{Cartier-Voros}, \cite{Voros}. 

The regularized determinant is defined by the zeta regularization method. We recall briefly this method. Let
\begin{equation}
\lambda_0,\lambda_1,\lambda_2,\ldots
\end{equation}
be the sequence of all eigenvalues including their multiplicities of an operator $\mathcal K$. 
The following zeta function can be assigned to this sequence
\begin{equation}\label{fori}
Z(z,w)=\sum_{k=0}^\infty\dfrac{1}{(\lambda_k+w)^z}.
\end{equation}
We assume that this series converges in a domain of $z$ and admits an analytic continuation to a domain containing $z=0$ such that $Z(z,w)$ is holomorphic at $z=0$. 
Then the regularized determinant of $\mathcal K$ is defined by
\begin{equation}\label{detdet}
\det(\mathcal K-\lambda):=\exp(-Z'(0,-\lambda))
\end{equation}
where $Z'$ denotes the first derivative of $Z$ with respect to $z$.

Let $\Delta_2$ denotes the Laplacian on the two dimensional sphere. The set of eigenvalues of this operator is given by
\begin{equation}
\left\lbrace \lambda_l=l(l+1) \ \vert \ l=0,1,2,\ldots\right\rbrace 
\end{equation}
with multiplicity $2l+1$. We put,
\begin{equation}
L_2:=\sqrt{\Delta_2+\frac{1}{4}}-\frac{1}{2}.
\end{equation}
We are going to assign a regularized determinant to this operator.
The set of eigenvalues of $L_2$ is given by
\begin{equation}
\left\lbrace \lambda_l=l \ \vert \ l\in \mathbb N\cup\left\lbrace 0\right\rbrace \right\rbrace 
\end{equation}
with multiplicity $2l+1$.
The zeta function associated to this eigenvalues is given by
\begin{equation}
Z(z,w)=\sum_{k=0}^\infty\dfrac{2k+1}{(k+w)^z},\quad Re(z)\gg0\,,Re(w)\gg0.
\end{equation}
This sum can be written as,
\begin{equation}\label{lamas}
Z(z,w)=2\sum_{m,n=0}^\infty\dfrac{1}{(m+n+w)^z}-\sum_{k=0}^\infty\dfrac{1}{(k+w)^z}.
\end{equation}
or in terms of the Barnes multiple zeta functions,
\begin{equation}\label{kas}
Z(z,w)=2\zeta_2(z,w\vert \ 1,1)- \zeta_1(z,w \vert 1).
\end{equation}
The general Barnes multiple zeta function is defined by \cite{Ruij},
\begin{equation}
\zeta_N(z,w\vert \ a_1,a_2,\ldots, a_N)=\sum_{m_1,m_2,\ldots,m_N=0}^\infty\dfrac{1}{(w+m_1a_1+m_2a_2+\ldots,m_Na_N)^z}
\end{equation}
which is convergent for $Re(w)>0$, $Re(z)>N$. 
We also note that $\zeta_1(z,w \vert 1)$ coincides with the Hurwitz zeta function $\zeta(z,w)$,
\begin{equation}\label{porjani}
\zeta(z,w)=\zeta_1(z,w \vert 1).
\end{equation}
We summarize the analytic properties of the Barnes multiple zeta function in the following lemma \cite{Barnes};
\begin{lemma}
The Barnes multiple zeta function $\zeta_N(z,w\vert \ a_1,a_2,\ldots, a_N)$ has a meromorphic continuation to the whole complex $z$-plane with simple poles at $z=1,2,\ldots, N$ 
and is holomorphic at $w=0$.
Moreover,
\begin{equation}\label{pop1}
\Psi_N(w\vert \ a_1,a_2,\ldots, a_N):=\partial_z\zeta_N(z,w\vert \ a_1,a_2,\ldots, a_N)\vert_{z=0}
\end{equation}
has an analytic continuation to the whole complex $w$-plane.
\end{lemma}
Hence, $Z(z,w)$ as given in (\ref{kas}) is holomorphic at $z=0$
and we can define the regularized determinant of $L_2$ as follows
\begin{equation}\label{pop3}
\det(L_2+s):=\exp(-Z'(0,s)).
\end{equation}
By inserting the right hand side of (\ref{kas}) into this definition we get
\begin{equation}
\det(L_2+s):=\exp\left(-2\zeta_2'(0,s\vert \ 1,1)+\zeta_1'(0,s \vert 1)\right)=\exp\left(-2\Psi_2(s\vert \ 1,1)+\Psi_1(s \vert 1)\right)
\end{equation}
or equivalently,
\begin{equation}\label{dorothy}
\det(L_2+s):=\dfrac{\exp\Psi_1(s \vert 1)}{[\exp\Psi_2(s \vert 1,1)]^2}.
\end{equation}
The Barnes multiple gamma functions are defined as (see for example \cite{Ruij}, page 119)
\begin{equation}\label{pop2}
\Gamma_N(w\vert \ a_1,a_2,\ldots, a_N)=\exp(\Psi_N(w\vert \ a_1,a_2,\ldots, a_N)).
\end{equation}
According to this definition the identity (\ref{dorothy}) can be written as
\begin{equation}\label{ghod}
\det(L_2+s)=\dfrac{\Gamma_1(s \ \vert\ 1)}{\Gamma^2_2(s \ \vert \ 1,1)}.
\end{equation}
On the other hand we note that (see \cite{Barnes}, page 363),
\begin{equation}
\Gamma^{-1}_2(s \ \vert \ 1,1)=\Gamma^{-1}_2(s)(2\pi)^{-s/2}
\end{equation}
where $\Gamma_2(s)$ is the Barnes double gamma function
and also (see \cite{Ruij}, formula 3.27),
\begin{equation}\label{sag}
\Gamma_1(s \ \vert\ 1)=\Gamma(s)(2\pi)^{-\frac{1}{2}}.
\end{equation}
Hence, (\ref{ghod}) can be written as
\begin{equation}
\det(L_2+s)=(2\pi)^{-s-\frac{1}{2}}\Gamma(s)\Gamma^{-2}_2(s).
\end{equation}
Comparing this with (\ref{dodod}) we arrive at
\begin{lemma}\label{splp}
Let $\Delta_2$ be the Laplace operator on the two dimensional sphere and $L_2:=\sqrt{\Delta_2+\frac{1}{4}}-\frac{1}{2}$. Moreover, let  $Z_I(s;\Gamma;\chi)$ be the zeta function for the identity contribution for the Fuchsian group $\Gamma$ of the first kind with an $n$-dimensional representation $\chi$ and $\vert F\vert$ be the measure of the fundamental domain of $\Gamma$. Then the following identity holds,
\begin{equation}\label{sol123}
Z_I(s;\Gamma;\chi)=\left[\sqrt{2\pi} \det(L_2+s)\right]^{-\frac{n\vert F\vert}{2\pi}}
\end{equation}
where the determinant is defined in (\ref{pop3}).
\end{lemma}
\subsection{Determinant expressions of the zeta functions for the elliptic and parabolic contributions}
We are going to find some determinant representations for the zeta functions for the parabolic and elliptic contributions. To this end we must write the gamma function as a determinant of an operator.
Consider the Schroedinger operator for the one dimensional harmonic oscillator
\begin{equation}
H=-\dfrac{\hbar^2}{2m}\dfrac{d^2}{dx^2}+\dfrac{m\omega^2}{2}x^2.
\end{equation}
The set of eigenvalues of $H$ is given by
\begin{equation}
\left\lbrace E_n=(n+\dfrac{1}{2})\hbar\omega \ \vert \ n\in\mathbb N\cup {0}\right\rbrace.
\end{equation}
Then we put
\begin{equation}
H_1:=\dfrac{H}{\hslash\omega}-\dfrac{1}{2}.
\end{equation}
The set of eigenvalues of $H_1$ is $\mathbb N\cup\left\lbrace 0\right\rbrace$. The zeta function (\ref{fori}) assigned to this sequence of eigenvalues is the Hurwitz zeta function,
\begin{equation}
\zeta(z,w)=\sum_{n=0}^\infty\left( \dfrac{1}{n+w}\right)^z
\end{equation}
with a well known analytic continuation to the whole complex plane in both variables $z$ and $w$.
As in the previous section we define
\begin{equation}\label{detrz}
\det(H_1+\lambda)=\exp(-\partial_z\zeta(z,\lambda)\vert_{z=0}).
\end{equation}

Thus from (\ref{porjani}), (\ref{pop1}), (\ref{pop2}), (\ref{sag}) we conclude that
\begin{equation}
\det(H_1+\lambda)=\dfrac{1}{\Gamma_1(\lambda\vert 1)}=\dfrac{(2\pi)^{\frac{1}{2}}}{\Gamma(\lambda)}.
\end{equation}
By comparing this determinant with the definitions of the zeta functions $Z_E$ and $Z_P$ we obtain the desired determinant expressions:
\begin{lemma}\label{f105}
Let $Z_E(s;\Gamma;\chi)$ be the zeta function for the elliptic contribution for a Fuchsian group $\Gamma$ of the first kind with an $n$-dimensional 
unitary representation $\chi$ of degree of non-singularity $k(\Gamma;\chi)$, let $H$ be the Schroedinger operator for the one dimensional harmonic oscillator 
and $H_1:=\dfrac{H}{\hslash\omega}-\dfrac{1}{2}$. 
Then the following determinat expression holds
\begin{equation}\label{sol1234}
Z_E(s;\Gamma;\chi)=\prod_{\left\lbrace R\right\rbrace _\Gamma}\prod_{l=0}^{\nu-1}\left( (2\pi)^{-\frac{1}{2}}\det(H_1+\dfrac{s+l}{\nu_R})\right)^{-\frac{-n(\nu-1)+\alpha(R,l)}{\nu}}.
\end{equation}
In this identity $R$ denotes a representative of an elliptic conjugacy class in $\Gamma$ of order $\nu$, $\alpha(R,l)$ is given in (\ref{fouz}) and the determinant is defined by (\ref{detrz}).
\end{lemma}
\begin{lemma}\label{f106}
Let $Z_P(s;\Gamma;\chi)$ be the zeta function for the parabolic contribution for a Fuchsian group $\Gamma$ of the first kind with an $n$-dimensional unitary representation $\chi$ of degree of non-singularity $k(\Gamma;\chi)$, let $H$ be the Schroedinger operator for the one dimensional harmonic oscillator and $H_1:=\dfrac{H}{\hslash\omega}-\dfrac{1}{2}$. Then the following determinat expression holds
\begin{equation}\label{sol12345}
Z_P(s;\Gamma;\chi)=e^{-c(n,h)s}2^{-k(\Gamma; \chi)s}(2\pi)^{-\frac{k(\Gamma;\chi)}{2}}(s-\frac{1}{2})^{-\frac{k(\Gamma;\chi)}{2}}\det(H_1+s+\dfrac{1}{2})^{k(\Gamma;\chi)}.
\end{equation}
In this identity $c(n,h)$ is given in (\ref{ccoocc}) where $h$ denotes the number of inequivalent cusps of $\Gamma$ and the determinant is defined by (\ref{detrz}).
\end{lemma}
\subsection{Regularized determinant of the automorphic Laplacian}
\label{RdaL}
In this subsection we generalize Efrat's  work \cite{Efrat} on the regularized determinant for the shifted automorphic Laplacian. 
Consider the automorphic Laplacian $A(\Gamma;\chi)$ for a general Fuchsian group $\Gamma$ of the first kind with an $n$-dimensional unitary representation $\chi$ 
and the sets $S_1$, $S_2$, and $S_3$ as defined in subsection 3.1.
For $s\in S_1$, $s(1-s)$ is an eigenvalue of $A(\Gamma;\chi)$, corresponding to a cusp form as eigenfunction.
For $s\in S_3$, $s(1-s)$ is an eigenvalue of $A(\Gamma;\chi)$ where the corresponding eigenfunction is the residue of an Eisenstein series at its pole $s\in S_3$. 
We denote by $S'_3$ the set of zeros of the determinant of the scattering matrix in the interval $[0,\frac{1}{2})$ and hence
\begin{equation}
S'_3=\left\lbrace 1-s \,\vert\, s\in S_3\right\rbrace.
\end{equation}
We also recall that $S_2$ is the set of resonances in the half-plane $Re(s)<\frac{1}{2}$.

A spectral zeta function is defined by
\begin{equation}\label{spectral zeta function}
\zeta(w,s):=\sum_{\sigma\in S}\dfrac{1}{(\sigma(1-\sigma)-s(1-s))^w},\quad s\gg0,\quad w\gg0
\end{equation}
where $S=S_1\cup S_2\cup S'_3$. As we recalled in formula \eqref{detdet}, the regularity of $\zeta(w,s)$ at $w=0$ is crucial for defining a regularized determinant. First we describe the analytic continuation of $\zeta(w,s)$ to a domain in the complex $w$-plane including $w=0$. Then we prove the regularity of $\zeta(w,s)$ at $w=0$. 

Let
\begin{equation}\label{thetat}
\theta(t):=\sum_{\sigma\in S}e^{-\sigma(1-\sigma)t},\quad t>0.
\end{equation}
Then the zeta function $\zeta(w,s)$ can be represented as the Mellin transform of $\theta(t)$ (similar to \cite{Efrat}, page 446),
\begin{equation}\label{meromorphic continuation spectral zeta function}
\zeta(w,s)=\dfrac{1}{\Gamma(w)}\int_0^\infty \theta(t)e^{s(1-s)t}t^w\dfrac{dt}{t},\quad s\gg0,\quad w\gg0.
\end{equation}
Based on this identity, the analytic continuation of $\zeta(s,w)$ is obtained from the asymptotics of $\theta(t)$ as $t\rightarrow\infty$ and $t\rightarrow0$. It is easy to see that:
\begin{lemma}\label{kool341}
The following asymptotic holds,
\begin{equation}
\theta(t)=O(1),\qquad t\rightarrow\infty.
\end{equation}
\end{lemma}
\begin{lemma}\label{kool340}
For some constants $\alpha, \beta, \gamma, \delta$, the function $\theta(t)$ has the following asymptotic behavior for $t\rightarrow0$:
\begin{equation}
\theta(t)=\dfrac{\alpha}{t}+\beta\dfrac{\log t}{\sqrt{t}}+\dfrac{\gamma}{\sqrt{t}}+\delta+O(\sqrt{t}\log t ),\quad t\rightarrow0
\end{equation}
\end{lemma}
\begin{proof}
Let
\begin{equation}
h(r^2+\dfrac{1}{4}):=e^{-(r^2+\frac{1}{4})t},\quad t>0
\end{equation}
and
\begin{equation}
g(u)=\dfrac{1}{2\sqrt{\pi t}}e^{-\frac{u^2}{4t}-\frac{t}{4}}.
\end{equation}
Then by inserting this pair of functions into the Selberg trace formula defined in Theorem \ref{Sel-tr}, we get a version of Selberg's trace formula given by
\begin{equation}\label{tf24}
\sum_{\sigma\in S_1\cup S_3\cup S'_3}e^{-\sigma(\sigma-1)t}+C_1(t)+C_2(t)=I(t)+H(t)+E(t)+P_1(t)+P_2(t)+P_3(t)
\end{equation}
where 
\begin{equation}\label{concon}
C_1(t)=-\frac{1}{2\pi}\int_{-\infty}^{\infty}\dfrac{\varphi'}{\varphi}(\frac{1}{2}+ir;\Gamma;\chi)e^{-(r^2+\frac{1}{4})t}dr,
\end{equation}
\begin{equation}
C_2(t)=\dfrac{K_0}{2}e^{-\frac{t}{4}},
\end{equation}
\begin{equation}
I(t)=\dfrac{n\vert F\vert}{2\pi}\int_{-\infty}^\infty r \tanh(\pi r) e^{-(r^2+\frac{1}{4})t}dr,
\end{equation}
\begin{equation}
H(t)=\dfrac{1}{\sqrt{\pi t}}e^{-\frac{t}{4}}\sum_{\left\lbrace P\right\rbrace _\Gamma}\sum_{k=1}^\infty \dfrac{\text{tr}_V\chi^k(P)logN(P)}{N(P)^{\frac{k}{2}}-N(P)^{-\frac{k}{2}}}e^{-\frac{(k\log N(P))^2}{4t}},
\end{equation}
\begin{equation}
E(t)=\sum_{\left\lbrace R\right\rbrace _\Gamma}\sum_{m=1}^{\nu-1}\dfrac{\text{tr}_V\chi^k(R)}{\nu \sin\pi m/\nu}\int_{-\infty}^\infty \dfrac{exp(-2\pi rm/\nu)}{1+exp(-2\pi r)}e^{-(r^2+\frac{1}{4})t}dr
\end{equation}
\begin{equation}\label{pppp123}
P_1(t)=-\dfrac{k(\Gamma;\chi)}{\pi}\int_{-\infty}^\infty\psi(1+ir)e^{-(r^2+\frac{1}{4})t}dr,
\end{equation}
\begin{equation}\label{hji01}
P_2(t)=-2\left( k(\Gamma; \chi)\ln2+\sum_{\alpha=1}^h\sum_{l=k_\alpha+1}^n\ln\vert 1-\exp(2\pi i\theta_{\alpha l})\vert\right) \dfrac{1}{2\sqrt{\pi t}}e^{-\frac{t}{4}},
\end{equation}
\begin{equation}
P_3(t)=\dfrac{k(\Gamma;\chi)}{2}e^{-\frac{t}{4}}.
\end{equation}
Note that in \eqref{tf24}, the sum runs over spectral parameters $\sigma\in S_1\cup S_3\cup S_3'$ which is twice the sum over the eigenvalues in \eqref{trace formula}. This is why all contributions in \eqref{tf24} are twice the corresponding contributions in \eqref{trace formula}. 

By a simple calculation, from Theorem 3.5.5 and formula 3.5.2 in \cite{Alexei2}, one can conclude the following formula, 
\begin{equation}
-\dfrac{\varphi'}{\varphi}(\frac{1}{2}+ir;\Gamma;\chi)=\sum_{\sigma\in S_2}\dfrac{1-2\beta}{(r-\gamma)^2+(\frac{1}{2}-\beta)^2}+\sum_{\sigma\in S_3}\dfrac{1-2\beta}{r^2+(\frac{1}{2}-\beta)^2}+c_0
\end{equation}
where $\sigma=\beta+i\gamma$ and $c_0$ is a constant.
Inserting this into (\ref{concon}) and using Cauchy's integration method one gets
\begin{equation}
C_1(t)=\sum_{\sigma\in S_2}e^{-\sigma(1-\sigma)t}-\sum_{\sigma\in S_3}e^{-\sigma(1-\sigma))t}.
\end{equation}
Hence the Selberg trace formula (\ref{tf24}) reduces to
\begin{equation}\label{loebat}
\theta(t)=\sum_{\sigma\in S_1\cup S_2\cup S'_3}e^{-\sigma(1-\sigma)t}=I(t)+H(t)+E(t)+P_1(t)+P_2(t)+P_3(t)-C_2(t).
\end{equation}
By investigating the asymptotics of each term in the right hand side of this identity as $t\rightarrow0$, we obtain the estimate of $\theta(t)$ as $t\rightarrow0$. According to Venkov (see \cite{Alexei2}, page 78), the following estimates hold
\begin{eqnarray}\label{esi}
&&I(t)=\dfrac{n\vert F\vert}{2\pi}\dfrac{1}{t}+\underset{t\rightarrow0}{O(1)},\quad H(t)=\underset{t\rightarrow0}{o(1)},\quad E(t)=\underset{t\rightarrow0}{O(1)},\quad \nonumber\\&&C_2(t)=\underset{t\rightarrow0}{O(1)},\quad 
P_1(t)=\underset{t\rightarrow0}{O(\dfrac{\log t}{\sqrt{t}})},\nonumber\\&&\quad
P_2(t)=\underset{t\rightarrow0}{O(\dfrac{1}{\sqrt{t}})},\quad P_3(t)=\underset{t\rightarrow0}{O(1)}.
\end{eqnarray}
To obtain the desired result we need a more precise estimates for $P_1(t)$ and $P_2(t)$. For $P_2(t)$ we obtain
\begin{equation}\label{esi2}
P_2(t)=\dfrac{c}{\sqrt{t}}+O(\sqrt{t}), \qquad t\rightarrow0.
\end{equation}
Now we calculate the desired estimate for $P_1(t)$. Similar calculations has been done in \cite{Alexei3}.
The digamma function $\psi(1+ir)$ can be written in terms of Gamma function,
\begin{equation}
\psi(1+ir)=\dfrac{1}{i}\dfrac{d}{dr}\log\Gamma(1+ir).
\end{equation}
Inserting this into \eqref{pppp123} we get
\begin{equation}
P_1(t)=-\dfrac{k(\Gamma;\chi)}{\pi i}e^{-\frac{t}{4}}\int_{-\infty}^\infty\dfrac{d}{dr}\log\Gamma(1+ir)e^{-tr^2}dr.
\end{equation}
Then integrating by parts leads for $P_1(t)$ to
\begin{equation}\label{force2}
P_1(t)=\dfrac{2k(\Gamma;\chi)i}{\pi}e^{-\frac{t}{4}}t\int_{-\infty}^\infty re^{-tr^2}\log\Gamma(1+ir)dr.
\end{equation}
The following estimate holds (part 12.33 of \cite{Watson}):
\begin{equation}
\ln\Gamma(z)=(z-\dfrac{1}{2})\ln z-z+\dfrac{1}{2}\ln2\pi+\dfrac{1}{12}\dfrac{1}{z+1}+O(z^{-2}),\qquad \Re z>0
\end{equation}
Inserting this estimate into \eqref{force2}, we get
\begin{equation}\label{su678}
P_1(t)=\dfrac{2k(\Gamma;\chi)i}{\pi}t\left[ A(t)+B(t)+C(t)+D(t)\right]+O(t\log t),\qquad t\rightarrow0
\end{equation}
where
\begin{equation}
A(t)=\int_{-\infty}^\infty re^{-tr^2}(\frac{1}{2}+ir)\log(1+ir)dr,
\end{equation}
\begin{equation}\label{efb2}
B(t)=-\int_{-\infty}^\infty re^{-tr^2}(1+ir)dr=-i\int_{-\infty}^\infty r^2e^{-tr^2}dr=\dfrac{-i}{4t}\sqrt{\dfrac{\pi}{t}},
\end{equation}
\begin{equation}\label{efb5}
C(t)=\dfrac{1}{2}\ln2\pi\int_{-\infty}^\infty re^{-tr^2}dr=0,
\end{equation}
\begin{equation}
D(t)=\dfrac{1}{12}\int_{-\infty}^\infty re^{-tr^2}\dfrac{1}{2+ir}dr.
\end{equation}
The estimate $O(t\log t)$ in \eqref{su678}, is obtained from the estimate of the following integral as $t\rightarrow 0$,
\begin{equation}
\int_{-\infty}^\infty \vert r e^{-tr^2}(1+ir)^{-2}\vert dr
\end{equation}
By a simple caculation this integral can be written as
\begin{equation}
e^t\int_{1}^{\infty}\dfrac{e^{-tu}}{u}du
\end{equation}
As $t\rightarrow+0$, this integral behaves like $-\log t$ which leads to the desired result.
For $A(t)$ the following estimate holds (see \cite{Efrat}, page 447)
\begin{equation}\label{efb1}
A(t)=c_1\dfrac{\ln t}{t\sqrt{t}}+c_2\dfrac{1}{t\sqrt{t}}+ c_3\dfrac{1}{t}+O(\dfrac{1}{\sqrt{t}}),\qquad t\rightarrow 0
\end{equation}
where $c_1$, $c_2$, and $c_3$ are some constants.
The integral $D(t)$ can be written as
\begin{equation}
D(t)=\dfrac{-i}{6}\int_0^\infty\dfrac{r^2e^{-tr^2}}{4+r^2}dr
\end{equation}
According to \cite{Gradshteyn} (see there page 338), for some constant $c_4$ we get
\begin{equation}\label{efb3}
D(t)=\dfrac{c_4}{\sqrt{t}}+O(1),\qquad t\rightarrow0.
\end{equation}
Finally, by inserting \eqref{efb2}, \eqref{efb5}, \eqref{efb1}, and \eqref{efb3} into \eqref{su678}, we get the desired estimate for $P_1(t)$, namely
\begin{equation}\label{esi3}
P_1(t)=a_1\dfrac{\ln t}{\sqrt{t}}+a_2\dfrac{1}{\sqrt{t}}+a_3+O(t\log t),\qquad t\rightarrow0
\end{equation}
where $a_1$, $a_2$, and $a_3$ are some constants.
The identity \eqref{loebat} together with the estimates \eqref{esi}, \eqref{esi2}, and \eqref{esi3} complete the proof.
\end{proof}
\begin{lemma}
For a fixed real $s\gg0$ the zeta function (\ref{spectral zeta function}) has an analytic (meromorphic) continuation to the half-plane $\text{Re}(w)>-\dfrac{1}{2}$ and it is holomorphic at $w=0$. 
\end{lemma}
\begin{proof}
According to Lemma \ref{kool341} and Lemma \ref{kool340}, there exist a function $f(t)$ with asymptotics
\begin{equation}\label{fffppf}
f(t)=\underset{t\rightarrow0}{O(\sqrt{t}\log t)},\quad f(t)=\underset{t\rightarrow\infty}{O(1)}
\end{equation}
such that 
\begin{equation}\label{t-asy}
\theta(t)=\dfrac{\alpha}{t}+\beta\dfrac{\log t}{\sqrt{t}}+\dfrac{\gamma}{\sqrt{t}}+\delta+f(t).
\end{equation}
Inserting (\ref{t-asy}) into (\ref{meromorphic continuation spectral zeta function}) and calculating the integrals (see \cite{Efrat}, page 448) one gets,
\begin{eqnarray}
&&\zeta(w,s)=\alpha\dfrac{\Gamma(w-1)}{\Gamma(w)}\left[s(s-1) \right]^{1-w}\nonumber\\&&+\gamma\dfrac{\Gamma(w-\frac{1}{2})}{\Gamma(w)}\left[s(s-1) \right]^{\frac{1}{2}-w}+\dfrac{\delta}{\left[s(s-1)\right]^w}\nonumber\\&&
+\beta\dfrac{\Gamma(w-\frac{1}{2})}{\Gamma(w)}\left[s(s-1) \right]^{\frac{1}{2}-w}\left( \dfrac{\Gamma'}{\Gamma}(w-\dfrac{1}{2})-\log(s(s-1))\right)\nonumber\\&&+\dfrac{1}{\Gamma(w)}\int_0^\infty f(t)e^{s(1-s)t}t^{w-1}dt.
\end{eqnarray}
Because of the estimates in (\ref{fffppf}) the integral in the last formula is convergent for $\text{Re}(w)>-\frac{1}{2}$. On the other hand $\dfrac{1}{\Gamma(w)}$ is regular and vanishing at $w=0$. Hence, $\zeta(w,s)$ is holomorphic at $w=0$.
\end{proof}

The regularized determinant ${\det}_R$ of the automorphic Laplacian $A=A(\Gamma;\chi)$ is defined as
\begin{equation}\label{determinant of automorphic laplacian}
\det(A-s(1-s)):=\exp{\left( -\frac{\partial}{\partial w}\zeta(w,s)\vert_{w=0}\right)}.
\end{equation}
This determinant is closely related to Selberg's zeta function. To derive this relation, consider the trace formula (\ref{trace formula}) with the test function $h$ as given in (\ref{h1}),
\begin{eqnarray}\label{eop}
&&\sum_{\sigma\in S_1\cup S_3\cup S'_3}\left( \dfrac{1}{\sigma(1-\sigma)-s(1-s)}-\dfrac{1}{\sigma(1-\sigma)+\beta^2-\frac{1}{4}}\right)+C_1(s)+C_2(s)\nonumber\\&&=2I(s)+2H(s)+2E(s)+2P(s) 
\end{eqnarray}
where $H(s)$, $I(s)$, $E(s)$, and $P(s)$ are the contributions of hyperbolic, identity, elliptic and parabolic elements for the aforementioned test function and they are the same as the corresponding contributions in formulas \eqref{logarithmic derivitive of Z(s)}, \eqref{logarithmic derivitive of Z_X1}, \eqref{elel} and \eqref{pelpel}, respectively. Moreover,
\begin{equation}
C_1(s):=-\frac{1}{2\pi}\int_{-\infty}^{\infty}\dfrac{\varphi'}{\varphi}(\frac{1}{2}+ir;\Gamma;\chi)\left( \dfrac{1}{r^2+\frac{1}{4}-s(1-s)}-\dfrac{1}{r^2+\beta^2}\right) dr
\end{equation}
and
\begin{equation}\label{babas}
C_2(s):=\dfrac{K_0}{2}\left( \dfrac{1}{(s-\frac{1}{2})^2}-\dfrac{1}{\beta^2}\right).
\end{equation}
We note that all contributions in \eqref{eop} are twice the contributions in Selberg's trace formula \eqref{trace formula} because in the right hand side of \eqref{eop}, the sum is over the spectral parameters $\sigma\in S_1\cup S_3\cup S_3'$ while in \eqref{trace formula} the sum is over the eigenvalues.
According to Venkov (see \cite{Alexei2}, p.84), for some constant $c$, the following holds
\begin{eqnarray}
&&C_1(s)=-\dfrac{1}{2s-1}\dfrac{d}{ds}\log\phi(s)\nonumber\\&&+\sum_{\sigma\in S_2}\left( \dfrac{1}{\sigma(1-\sigma)-s(1-s)}-\dfrac{1}{\sigma(1-\sigma)+\beta^2-\frac{1}{4}}\right)\nonumber\\&&-\sum_{\sigma\in S_3}\left( \dfrac{1}{\sigma(1-\sigma)-s(1-s)}-\dfrac{1}{\sigma(1-\sigma)+\beta^2-\frac{1}{4}}\right)+c.
\end{eqnarray}
Now, inserting this and (\ref{babas}) into (\ref{eop}), we get
\begin{eqnarray}
&&\sum_{\sigma\in S_1\cup S_2\cup S'_3}\left( \dfrac{1}{\sigma(1-\sigma)-s(1-s)}-\dfrac{1}{\sigma(1-\sigma)+\beta^2-\frac{1}{4}}\right)=\nonumber\\&&\dfrac{1}{2s-1}\dfrac{d}{ds}\log\phi(s)+2I(s)+2H(s)+2E(s)+2P(s)\nonumber\\&&-\dfrac{K_0}{2}\left( \dfrac{1}{(s-\frac{1}{2})^2}-\dfrac{1}{\beta^2}\right). 
\end{eqnarray}
Next, by differentiating both sides of this we get
\begin{eqnarray}\label{kokoko}
&&\sum_{\sigma\in S_1\cup S_2\cup S'_3}\dfrac{-(2s-1)}{\left[ \sigma(1-\sigma)-s(1-s)\right] ^2}=\dfrac{d}{ds}\dfrac{1}{2s-1}\dfrac{d}{ds}\log\phi(s)+\nonumber\\&&2\dfrac{d}{ds}I(s)+2\dfrac{d}{ds}H(s)+2\dfrac{d}{ds}E(s)+2\dfrac{d}{ds}P(s)-\dfrac{K_0}{2}\dfrac{d}{ds} \dfrac{1}{(s-\frac{1}{2})^2}.
\end{eqnarray}
On the other hand, according to (\ref{determinant of automorphic laplacian}) and \eqref{spectral zeta function}, the following identity holds
\begin{equation}
\dfrac{d}{ds}\dfrac{1}{2s-1}\dfrac{d}{ds}\log\det(A-s(1-s))=\sum_{\sigma\in S_1\cup S_2\cup S'_3}\dfrac{-(2s-1)}{\left[\sigma(1-\sigma)+s(1-s)\right]^2}.
\end{equation}
Therefore (\ref{kokoko}) can be written as
\begin{eqnarray}
&&\dfrac{d}{ds}\dfrac{1}{2s-1}\dfrac{d}{ds}\log\det(A-s(1-s))=\dfrac{d}{ds}\dfrac{1}{2s-1}\dfrac{d}{ds}\log\phi(s)+\nonumber\\&&2\dfrac{d}{ds}\dfrac{1}{2s-1}\dfrac{d}{ds}\log Z_I(s)+2\dfrac{d}{ds}\dfrac{1}{2s-1}\dfrac{d}{ds}\log Z(s)+\\&&2\dfrac{d}{ds}\dfrac{1}{2s-1}\dfrac{d}{ds}\log Z_E(s)+2\dfrac{d}{ds}\dfrac{1}{2s-1}\dfrac{d}{ds}\log Z_P(s)-\dfrac{K_0}{2}\dfrac{d}{ds}\left( \dfrac{1}{(s-\frac{1}{2})^2}\right).\nonumber
\end{eqnarray}
where in the right hand side we used formulas (\ref{logarithmic derivitive of Z(s)}), (\ref{logarithmic derivitive of Z_X1}), (\ref{elel}) and (\ref{pelpel}). This formula leads to the desired identity formulated as a lemma.
\begin{lemma}\label{hammer38}
Let $A=A(\Gamma;\chi)$ be the automorphic Laplacian for a Fuchsian group $\Gamma$ of the first kind with a unitary representation $\chi$. Let $Z(s):=Z(s;\Gamma;\chi)$ be Selberg's zeta function, $Z_I(s)=Z_I(s;\Gamma;\chi)$, $Z_E(s):=Z_E(s;\Gamma;\chi)$, and $Z_P(s):=Z_P(s;\Gamma;\chi)$ be the zeta functions for the contributions of the identity, elliptic, and parabolic classes, respectively, and also let $\varphi(s)=\varphi(s;\Gamma;\chi)$ be the determinant of the scattering matrix $\Phi(s)=\Phi(s;\Gamma;\chi)$. Then the following identity holds
\begin{equation}\label{determinant expression of Laplacian}
\det(A-s(1-s))=e^{(c_1s(s-1)+c_2)}(s-\dfrac{1}{2})^{-K_0}\varphi(s)Z_I^2(s)Z_E^2(s)Z_P^2(s)Z^2(s).
\end{equation}
In this formula, $K_0:=\text{tr}\Phi(\frac{1}{2};\Gamma;\chi)$, and the determinant of the Laplacian is defined in (\ref{determinant of automorphic laplacian}). Finally, the constants $c_1$ and $c_2$ can be determined by the asymptotic behavior of both sides as $s\rightarrow\infty$.
\end{lemma}
The complete Selberg zeta function for a Fuchsian group $\Gamma$ of the first kind and a unitary representation $\chi$ is defined by
\begin{equation}\label{complete Sel zeta}
\overset{\sim}{Z}(s;\Gamma;\chi):=Z_I(s;\Gamma;\chi)Z_E(s;\Gamma;\chi)Z_P(s;\Gamma;\chi)Z(s;\Gamma;\chi).
\end{equation}
Thus identity (\ref{determinant expression of Laplacian}) can be written as
\begin{equation}\label{determinant expression of Laplacian2}
\det(A-s(1-s))=e^{(c_1s(s-1)+c_2)}(s-\dfrac{1}{2})^{-K_0}\varphi(s)\overset{\sim}{Z}^2(s;\Gamma;\chi).
\end{equation}
\begin{corollary}\label{fun-eq-cs}
The complete Selberg zeta function fulfills the following functional equation
\begin{equation}
\overset{\sim}{Z}(1-s;\Gamma;\chi)=\exp(\frac{-i\pi K_0}{2})\varphi(s)\overset{\sim}{Z}(s;\Gamma;\chi).
\end{equation}
\end{corollary}
\begin{proof}
The regularized determinant ${\det}_{R}(A-s(1-s))$ is invariant under $s\rightarrow 1-s$. Hence, from the equation \eqref{determinant expression of Laplacian2} we get
\begin{equation}\label{polo201}
(s-\dfrac{1}{2})^{-K_0}\varphi(s)\overset{\sim}{Z}^2(s;\Gamma;\chi)=(-s+\dfrac{1}{2})^{-K_0}\varphi(1-s)\overset{\sim}{Z}^2(1-s;\Gamma;\chi).
\end{equation}
The determinant of the scattering matrix fulfills the following ientity \cite{Alexei}
\begin{equation}
\varphi(s)\varphi(1-s)=1.
\end{equation}
Hence, \eqref{polo201} can be written as
\begin{equation}\label{polo2012}
\varphi(s)^2\overset{\sim}{Z}^2(s;\Gamma;\chi)=(-1)^{-K_0}\overset{\sim}{Z}^2(1-s;\Gamma;\chi)
\end{equation}
or 
\begin{equation}
\overset{\sim}{Z}^2(1-s;\Gamma;\chi)=\exp(-i\pi K_0)\varphi(s)^2\overset{\sim}{Z}^2(s;\Gamma;\chi).
\end{equation}
This equation determines the following, up to the sign:
\begin{equation}
\overset{\sim}{Z}(1-s;\Gamma;\chi)=\exp(\frac{-i\pi K_0}{2})\varphi(s)\overset{\sim}{Z}(s;\Gamma;\chi).
\end{equation}
\end{proof}

In \cite{Faddeev}, Faddeev introduced a compact operator on certain Banach spaces and applied it for analytical continuation of the resolvent of the automorphic Laplacian $R(s)$ to the whole complex plane. In a soon coming paper we prove that for a generalized version of this operator, denoted by $\mathcal H(s;\Gamma;\chi)$, for a Fuchsian group $\Gamma$ of the first kind with a unitary representation $\chi$ the following identity, up to a nonzero holomorphic factor, holds
\begin{equation}
\det(1-\mathcal H(s;\Gamma;\chi))=\overset{\sim}{Z}(s)
\end{equation}
where $\det$ denotes certain regularized determinant. We use this identity in Theorem \ref{goaway}. For more details about the operator $\mathcal H(s;\Gamma;\chi)$ see \cite{Alexei2}. 

By inserting the determinant expression of each element in the right hand side of (\ref{determinant expression of Laplacian}), the regularized determinant of the automorphic Laplacian can be written as a product of determinants. We formulate this identities in the following theorem.
\begin{theorem} \label{goaway} For a Fuchsian group $\Gamma$ of the first kind with an $n$-dimensional unitary representation $\chi$, up to a nonzero holomorphic factor, the following identities hold
\begin{equation}
\det(A-s(1-s))=(s-\dfrac{1}{2})^{-K_0}\det\Phi(s)\det(1-\mathcal H(s;\Gamma;\chi))^2.
\end{equation}
In the case of a congruence subgroup $\Gamma$ of finite index in $PSL(2,\mathbb Z)$ we have also
\begin{eqnarray}&&
\det(A-s(1-s))=(s-\dfrac{1}{2})^{-K_0-k(\Gamma;\chi)}\det\Phi(s)\det(L_2+s)^{-\frac{n\vert F\vert}{\pi}}\nonumber\\&&\prod_{\left\lbrace R\right\rbrace _\Gamma}\prod_{l=0}^{\nu-1} \det(H_1+\dfrac{s+l}{\nu_R})^{e(\Gamma;\chi)}\det(H_1+s+\dfrac{1}{2})^{2k(\Gamma;\chi)}\det(1-\mathcal L_s^{\Gamma,\chi})^2.\nonumber\\&&
\end{eqnarray}
Here $\Phi(s)=\Phi(s;\Gamma;\chi)$ is the scattering matrix, $K_0=tr(\Phi(\frac{1}{2};\Gamma;\chi))$, $k(\Gamma;\chi)$ is the degree of non-singularity of $\chi$, $\left\lbrace R\right\rbrace _\Gamma$ denotes a set of all representatives $R$ with order $\nu$ of the elliptic conjugacy classes of $\Gamma$, $\alpha(R,l)$ and $c(n,h)$ have been defined in (\ref{fouz}) and (\ref{ccoocc}), respectively. Moreover,
\begin{equation}
e(\Gamma;\chi):=2\frac{-\dim \chi(\nu-1)+\alpha(R,l)}{\nu}.
\end{equation}
\end{theorem}
\subsection{Determinant identities for commensurable groups}
In this subsection we derive an identity, connecting the determinants of automorphic Laplacians for different manifolds. In general we are interested in determinant functor for space time category \cite{Reschetikhin}. To this end we need the following theorem \cite{Alexei}:
\begin{theorem}
Let $\chi_1$ and $\chi_2$ be unitary representations of a Fuchsian group $\Gamma$ of the first kind. Further, let $\Delta$ be a Fuchsian group of the first kind and $\Gamma\subset\Delta$ be of finite index in $\Delta$. Then the following identities hold
\begin{itemize}
\item $Z(s;\Gamma;\chi_1\oplus\chi_2)=Z(s;\Gamma;\chi_1)Z(s;\Gamma;\chi_2)$
\item $Z(s;\Gamma;\chi)=Z(s;\Delta;u^\chi)$
\end{itemize}
where $u^\chi$ is the representation of $\Delta$ induced by the representation $\chi$ of $\Gamma$.
\end{theorem}

For a normal subgroup $\Gamma\subset\Delta$ with trivial representation $\chi=1$ the second assertion of the theorem above can be expressed as (see \cite{Alexei}, page 50),
\begin{equation}\label{ffookk}
Z(s;\Gamma;1)=\prod_{\psi\in(\Gamma\setminus\Delta)^*}Z(s;\Delta;\psi)^{\dim\psi}
\end{equation}
where $\psi$ runs over the set of irreducible pairwise non-equivalent representations of the finite group $\Gamma\setminus\Delta$.

Now we consider commensurable groups $\Gamma_1$ and $\Gamma_2$ and we put $\Gamma_3:=\Gamma_1\cap\Gamma_2$. Assume that $\Gamma_3$ is a normal subgroup of $\Gamma_1$ and $\Gamma_2$. The formula (\ref{ffookk}) for $\Gamma_3$ as a subgroup of $\Gamma_1$ and $\Gamma_2$ is written respectively as,
\begin{equation}
Z(s;\Gamma_3;1)=\prod_{\psi_1\in(\Gamma_3\setminus\Gamma_1)^*}Z(s;\Gamma_1;\psi_1)^{\dim\psi_1}
\end{equation} 
and
\begin{equation}
Z(s;\Gamma_3;1)=\prod_{\psi_2\in(\Gamma_3\setminus\Gamma_2)^*}Z(s;\Gamma_2;\psi_2)^{\dim\psi_2}.
\end{equation} 
Hence, we get the following identity,
\begin{equation}\label{venven}
\prod_{\psi_1\in(\Gamma_3\setminus\Gamma_1)^*}Z(s;\Gamma_1;\psi_1)^{2\dim\psi_1}=\prod_{\psi_2\in(\Gamma_3\setminus\Gamma_2)^*}Z(s;\Gamma_2;\psi_2)^{2\dim\psi_2}.
\end{equation}

On the other hand from (\ref{determinant expression of Laplacian}) we have
\begin{equation}
Z^2(s;\Gamma;\chi)=e^{-(c_1s(s-1)-c_2)}(s-\dfrac{1}{2})^{K_0}\dfrac{{\det}(A(\Gamma;\chi)-s(1-s))}{\det\Phi(s)Z_I^2(s;\Gamma;\chi)Z_E^2(s;\Gamma;\chi)Z_P^2(s;\Gamma;\chi)}.
\end{equation}
Then by inserting the determinant expressions of the different zeta functions from (\ref{sol123}), (\ref{sol1234}), and (\ref{sol12345}), we get
\begin{eqnarray}\label{gharash}
&&Z^2(s;\Gamma;\chi)=f(\Gamma;\chi;s)\det(A(\Gamma;\chi)-s(1-s))\times\nonumber\\&&{\det}^{-1}\Phi(\Gamma;\chi;s)
\left[ \det(L_2+s)\right]^{\frac{n\vert F\vert}{\pi}}\times\nonumber\\&&\det(H_1+s+\dfrac{1}{2})^{-2k(\Gamma;\chi)}\prod_{\left\lbrace R\right\rbrace _\Gamma}\prod_{l=0}^{\nu-1}\left( \det(H_1+\dfrac{s+l}{\nu_R})\right)^{2\frac{-n(\nu-1)+\alpha(R,l)}{\nu}} 
\end{eqnarray}
where $f(\Gamma;\chi;s)$ is certain holomorphic function.
Now inserting (\ref{gharash}) for the corresponding groups and representations in both sides of (\ref{venven}), we get a relation between different determinants. We formulate this relation in the following theorem:
\begin{theorem}\label{koli76}
Let $\Gamma_1$ and $\Gamma_2$ be commensurable Fuchsian groups of the first kind such that $\Gamma_3:=\Gamma_1\cap\Gamma_2$ is a normal divisor of both $\Gamma_1$ and $\Gamma_2$. Then, up to a nonzero holomorphic factor, the following identity holds
\begin{eqnarray}
&&\prod_{\psi_1\in(\Gamma_3\setminus\Gamma_1)^*}(s-\dfrac{1}{2})^{a(\Gamma_1;\psi_1)} \det(A(\Gamma_1;\psi_1)-s(1-s))^{\dim\psi_1}\times\nonumber\\&&{\det}
\Phi(\Gamma_1;\psi_1;s)^{-\dim\psi_1}
\det(L_2+s)^{\frac{(\dim \psi_1)^2\vert F_1\vert}{\pi}}\times\nonumber\\&&\det(H_1+s+\dfrac{1}{2})^{-2k(\Gamma_1;\psi_1)\dim\psi_1}\prod_{\left\lbrace R\right\rbrace _{\Gamma_1}}\prod_{l=0}^{\nu-1} \det(H_1+\dfrac{s+l}{\nu_R})^{e(\Gamma_1;\psi_1)\dim\psi_1}\nonumber\\&&=\prod_{\psi_1\in(\Gamma_3\setminus\Gamma_2)^*}(s-\dfrac{1}{2})^{a(\Gamma_2;\psi_2)}{\det}(A(\Gamma_2;\psi_2)-s(1-s))^{\dim\psi_2}\times\nonumber\\&&{\det}
\Phi(\Gamma_2;\psi_2;s)^{-\dim\psi_2}
\left[ \det(L_2+s)\right]^{\frac{(\dim \psi_2)^2\vert F_2\vert}{\pi}}\times\nonumber\\&&\det(H_1+s+\dfrac{1}{2})^{-2k(\Gamma_2;\psi_2)\dim\psi_2}\times\nonumber\\&&\prod_{\left\lbrace R\right\rbrace _{\Gamma_2}}\prod_{l=0}^{\nu-1} \det(H_1+\dfrac{s+l}{\nu_R})^{e(\Gamma_2;\psi_2)\dim\psi_2}.
\end{eqnarray}
Here $\psi_1$ and $\psi_2$ run over the set of irreducible pairwise non-equivalent representations of the finite group $\Gamma_3\setminus\Gamma_1$ and $\Gamma_3\setminus\Gamma_2$, respectively. The term $a(\Gamma;\chi)$ is given by
\begin{eqnarray}\label{erij}
&&a(\Gamma;\chi):=\left[ K_0(\Gamma;\chi)+k(\Gamma;\chi)\right] \dim\chi\end{eqnarray}
where $K_0=tr(\Phi(\frac{1}{2};\Gamma;\chi))$, $k(\Gamma;\chi)$ is the degree of non-singularity of $\chi$, $\left\lbrace R\right\rbrace _\Gamma$ denotes a set of representatives $R$ with order $\nu$ of the elliptic conjugacy classes of $\Gamma$, $\alpha(R,l)$ has been defined in (\ref{fouz}), and
\begin{equation}
e(\Gamma;\chi):=2\frac{-\dim \chi(\nu-1)+\alpha(R,l)}{\nu}.
\end{equation}
\end{theorem}
\section{Jacquet-Langlands correspondence}
In this section we only consider the trivial representation $\chi=1$. In \cite{hejhal} an explicit integral operator lift with the Siegel theta function as kernel, between Maass forms of the unit group of an indefinite quaternion division algebra and a congruence subgroup of the modular group is constructed. This is indeed a special case of the Jacquet-Langlands correspondence which Hejhal reproved by using classical arguments \cite{hejhal}. Following \cite{hejhal}, \cite{bolte1}, and \cite{bolte} we illustrate this correspondence. We need first to recall the unit group of a quaternion algebra.
\subsection{Unit group of quaternion algebra}
Here we follow \cite{miyake}. A ring $B$ with unity is called an algebra of dimension $n$ over a field $F$, if the following three conditions are satisfied:
\begin{itemize}
\item[1]$F\subset B$ and the unity of $F$ coincides with the unity of $B$;
\item[2]the elements of $F$ commute with the elements of $B$;
\item[3]$B$ is a vector space over $F$ of dimension $n$. 
\end{itemize}
Let $B$ be an algebra over $F$ with center $Z(B)$.
The algebra $B$ is called a \textit{central} algebra if $Z(B)=F$.
The algebra $B$ is called \textit{simple} if it is simple as a ring that is, $B$ has no two-sided ideals except for $\left\lbrace 0\right\rbrace $ and $B$ itself.
The algebra $B$ is a \textit{division} algebra if every nonzero element of $B$ is invertible.
\begin{definiton}\label{def:QA}
A central simple algebra $B$ of dimension $4$ over a field $F$ is called a quaternion algebra over $F$.
\end{definiton}
Furthermore, if $B$ is a division algebra, we call $B$ a division quaternion algebra. 

Let $B$ be a quaternion algebra over a field $F$. Then there are only two possibilities (see \cite{lewis}, page 43)
\begin{itemize}
\item[1] either $B$ is a division quaternion algebra
\item[2] or $B$ splits over $F$ that is, $B$ is  isomorphic to $M_2(F)$, the algebra of all $2\times2$ matrices with entries from $F$.
\end{itemize}
Let $K$ be an extension of $F$ then one says that $B$ is ramified respectively splits over $K$ if $B\otimes_FK$, the tensor product of $B$ and $K$ as algebras over $F$, is a division quaternion algebra or is isomorphic to $M_2(K)$. 

A norm on $B$ is defined by using the following results \cite{weil}: 
\begin{itemize}
\item[1] If $F$ is algebraically closed, that is if every one variable polynomial of degree at least $1$ with coefficients in $F$ has a root in $F$ (see \cite{Lang}, page 178), then $M_2(F)$ is the unique quaternion algebra over $F$ up to isomorphism.
\item[2] If $K$ is any extension over $F$, then $B\otimes_FK$ is a quaternion algebra over $K$. 
\end{itemize}
Let $B$ be a quaternion algebra over $F$, and let $\overline{F}$ be the algebraic closure of $F$. According to the results above $B\otimes_{F}\overline{F}$ is a quaternion algebra over the algebraically closed field $\overline{F}$ and hence $B\otimes_{F}\overline{F}$ is isomorphic to $M_2(\overline{F})$. Now we can define the (reduced) norm and the (reduced) trace of elements of $B$ by
\begin{equation}
N_B(\beta)=\det(\beta),\quad tr_B(\beta)=tr(\beta),
\end{equation}
where $\det(\beta)$ and $tr(\beta)$ are the determinant and the trace of $\beta$ as an element in $M_2(\overline{F})$. In \cite{weil} it was shown that both $N_B(\beta)$ and  $tr_B(\beta)$ belong to $F$.

\begin{definiton}For an algebra $B$ not necessarily of quaternion type over the field of rational numbers $F=\mathbb Q$ or its p-adic extensions $F=\mathbb Q_p$ an order is defined as a subset $\mathcal O$ of $B$ satisfying the two conditions
\begin{itemize}
\item[1] $\mathcal O$ is a subring containing the unity of $B$
\item[2] $\mathcal O$ is finitely generated over $\mathbb Z$ or $\mathbb Z_p$ and contains a basis of $B$ over $F$.
\end{itemize}
\end{definiton}
An order of $B$ is called maximal if it is maximal with respect to inclusion. All maximal orders of an algebra $B$ are conjugate (see for example \cite{hejhal}, page 135).

From now on we assume $B$ to be a quaternion algebra over the field of rationals $\mathbb Q$. The algebra $B$ is characterized up to isomorphism by a positive integer $d(B)$ called the (reduced) discriminant which is defined to be the product of primes $p$ where $B$ is ramified over $\mathbb Q_p$ that is $B\otimes_{\mathbb Q}\mathbb Q_p$ is a division quaternion algebra. Note that some authors define discriminant as the square of what we presented as the definiton of $d(B)$ \cite{miyake, hejhal}. The discriminant is also defined for orders of the algebra $B$: the (reduced) discriminant $d(\mathcal O)$ of an order $\mathcal O$ of the algebra $B$ is defined by (see \cite{bolte1} and references there)
\begin{equation}
 d(\mathcal O)=\sqrt{\vert det[tr_B(\xi_j\xi_k)]\vert}
\end{equation}
where $\xi_1$, $\xi_2$, $\xi_3$, $\xi_4$ is any $\mathbb Z$-basis of $\mathcal O$ and
\begin{equation}
[tr_B(\xi_j\xi_k)]=\left( \begin{array}{cccc}
tr_B(\xi_1\xi_1)&tr_B(\xi_1\xi_2)&tr_B(\xi_1\xi_3)&tr_B(\xi_1\xi_4)\\
tr_B(\xi_2\xi_1)&tr_B(\xi_2\xi_2)&tr_B(\xi_2\xi_3)&tr_B(\xi_2\xi_4)\\
tr_B(\xi_3\xi_1)&tr_B(\xi_3\xi_2)&tr_B(\xi_3\xi_3)&tr_B(\xi_3\xi_4)\\
tr_B(\xi_4\xi_1)&tr_B(\xi_4\xi_2)&tr_B(\xi_4\xi_3)&tr_B(\xi_4\xi_4)\\
\end{array}
\right).
\end{equation}
For a maximal order $\mathcal O_{max}$ the discriminant is equal to the discriminant of the algebra $B$ that is $d(B)=d(\mathcal O_{max})$ (see \cite{stroembergsson} and references there).

\begin{definiton}The algebra $B$ is called indefinite or definite according to $B\otimes_{\mathbb Q}\mathbb R$ being isomorphic to $M_2(\mathbb R)$ or being a division quaternion algebra (\cite{miyake}, page 201).
\end{definiton}
We recall the following remark from \cite{stroembergsson}:
\begin{remark}
\begin{itemize}
\item For each square free number $d\in \mathbb Z^{+}$, there is exactly one quaternion algebra $B$ over $\mathbb Q$ up to isomorphisms with $d(B)=d$,
\item $d(B)>1$ if and only if $B$ is a division algebra,
\item $B$ being indefinite means that $d(B)$ has an even number of prime factors.
\end{itemize}
\end{remark}

Let $B$ be an indefinite quaternion algebra over $\mathbb Q$. We fix an isomorphism of $B\otimes_{\mathbb Q}\mathbb R$ and $M_2(\mathbb R)$ and consider $B$ as a subalgebra of $M_2(\mathbb R)$ through this isomorphism. Then the norm $N_B(\beta)$ of an element $\beta$ of $B$ is nothing but the determinant of $\beta$ as a matrix. Let $\mathcal O$ be an order of $B$. Then the unit group (of norm $1$) of $\mathcal O$ is defined by (\cite{miyake}, page 209)
\begin{equation}
\mathcal O^1=\left\lbrace \beta\in \ \mathcal O \vert\ N_B(\beta)=1\right\rbrace \subset SL(2,\mathbb R).
\end{equation}
Now we recall the following well known result (see \cite{miyake}, page 209):
\begin{theorem}
Let $B$ be an indefinite quaternion algebra over $\mathbb Q$, and $\mathcal O$ be an order of $B$. Then $\mathcal O^1$ is a Fuchsian group of the first kind. Moreover, if $B$ is a division quaternion algebra, then $\mathcal O^1 \backslash H$ is compact.
\end{theorem}
Because of this theorem from now on we restrict ourself to an order $\mathcal O$ of an indefinite quaternion algebra  over the field of rational numbers $\mathbb Q$ and its unit group $\mathcal O^1$.
\subsection{Siegel theta function}
In this subsection we introduce the Siegel theta function by following \cite{hejhal}, \cite{bolte1}, and \cite{bolte}. Let $\mathcal O$ be an order of an indefinite quaternion algebra over $\mathbb Q$. Since $\mathcal O$ is isomorphic to $\mathbb Z$, one can fix a basis $e_i,\ 1\leq i\leq4$ of $\mathcal O$ over $\mathbb Z$ such that
\begin{equation}
\mathcal O=e_1\mathbb Z\oplus e_2\mathbb Z\oplus e_3\mathbb Z\oplus e_4\mathbb Z.
\end{equation}
For an element $q\in \mathcal O$, let  
\begin{equation}
k_q=\left( \begin{array}{cccc}
k_1\\
k_2\\
k_3\\
k_4\\
\end{array}\right)
\end{equation}
be the vector representation of $q$ in the given basis.
Furthermore, since $\mathcal O$ is indefinite, one can fix an emmbeding
\begin{equation}\label{emm2}
\sigma:\mathcal O\longrightarrow M_2(\mathbb R).
\end{equation}
There is a unique $B\in GL_4(\mathbb R)$ describing this emmbeding in the following way:
\begin{equation}
\sigma(q)=\left(\begin{array}{cc}
\alpha&\beta\\
\gamma&\delta
\end{array}\right)\quad\text{whenever}\quad Bk_q=\left( \begin{array}{cccc}
\alpha\\
\beta\\
\gamma\\
\delta\\
\end{array}\right).
\end{equation}
Let 
\begin{equation}
S=\left( \begin{array}{cccc}
0&0&0&1\\
0&0&-1&0\\
0&-1&0&0\\
1&0&0&0\\
\end{array}\right). 
\end{equation}
Then by a simple calculation it can be shown
\begin{equation}
n(q)=\det(\sigma_q)=(\alpha\delta-\beta\gamma)=\frac{1}{2}(Bk_q)^tSBk_q.
\end{equation}
Thus the norm $n(q)$ of $q$ as a four dimensional vector space over $\mathbb Z$ is defined by the quadratic form
\begin{equation}
n(q)=\frac{1}{2}k_q^tS'k_q
\end{equation}
where $S'$ is a $4\times 4$ symmetric matrix given by
\begin{equation}
S'=B^tSB.
\end{equation}
For $L_1,\ L_2\in SL(2,\mathbb R)$ such that 
\begin{equation}
\left( \begin{array}{cccc}
\alpha_1&\beta_1\\
\gamma_1&\delta_1\\
\end{array}\right)=L_1\left( \begin{array}{cccc}
\alpha_2&\beta_2\\
\gamma_2&\delta_2\\
\end{array}\right)L_2^{-1}
\end{equation}
an element $A(L_1,L_2)\in M_4(\mathbb R)$ is defined by
\begin{equation}
\left( \begin{array}{cccc}
\alpha_1\\
\beta_1\\
\gamma_1\\
\delta_1\\
\end{array}\right)=A(L_1,L_2)\left( \begin{array}{cccc}
\alpha_2\\
\beta_2\\
\gamma_2\\
\delta_2\\
\end{array}\right).
\end{equation}
\begin{definiton}
For a symmetric matrix $\mathcal S$, the majorant $P$ is defined to be a positive definite symmetric matrix such that $P\mathcal S^{-1}P=\mathcal S$.
\end{definiton}
Let 
\begin{equation}
M_z:=\left( \begin{array}{cccc}
y^{\frac{1}{2}}&xy^{-\frac{1}{2}}\\
0&y^{-\frac{1}{2}}\\
\end{array}\right)\in SL(2,\mathbb R),\qquad z=x+iy\in \mathbb H.
\end{equation}
Then
\begin{equation}
P_{zw}:=A(M_z^{-1},M_w^{-1})^tA(M_z^{-1},M_w^{-1})
\end{equation}
is a majorant of $S$ (see \cite{bolte1}, page 14).
If $P$ is a majorant of $S$ then $B^tPB$ is a majorant of $B^tSB$ (see \cite{bolte1}, page 14).
According to this fact, 
\begin{equation}
P_{zw}':=B^tP_{zw}B
\end{equation}
is the majorant of $S'=B^tSB$ that is $P_{zw}'$ is symmetric positive definite. 

Now fix $z_0\in \mathbb H$ and take $\tau=u+iv\in\mathbb H,\ z=x+iy\in \mathbb H$. With $R:=uS'+ivP_{zz_0}'$, the Siegel theta function $\theta(z;\tau)$ is defined as
\begin{equation}
\theta(z;\tau):=Im(\tau) \sum_{k\in \mathbb Z^4}e^{\pi i k^tRk}=Im(\tau)\sum_{q\in \mathcal O}e^{\pi ik_q^tRk_q}.
\end{equation}
The Siegel theta function has the following transformation properties which is crucial for the application in the next subsection. (for the proof see \cite{bolte1}, page 16). 
\begin{theorem}
Let $\mathcal O$ be an order in an indefinite quaternion algebra over $\mathbb Q$, with (reduced) discriminant $d$. Then
\begin{itemize}
\item[(1)] $\theta(\sigma_qz;\tau)=\theta(z;\tau),\ \forall q\in \mathcal O^1$
\item[(2)] $\theta(z;g\tau)=\theta(z;\tau),\ \forall g\in \Gamma_0(d)$.
\end{itemize}
Note that for $q\in \mathcal O^1$, $\sigma_q=\sigma(q)\in SL(2,\mathbb R)$.
\end{theorem}
\subsection{Theta-lifts}
In this subsection we recall two integral transformations providing a lift between Maass forms for congruence subgroups and nonconstant square integrable automorphic eigenfunctions of the hyperbolic Laplacian for the unit group of quaternions.\
As before, let $\mathcal O$ be an order with discriminant $d(\mathcal O)$ in an indefinite quaternion division algebra over $\mathbb Q$ and $\mathcal O^1$ be the corresponding unit group. We also put $X_d:=\Gamma_0(d)\backslash \mathbb H$ and $X_{\mathcal O}:=\mathcal O^1\backslash \mathbb H$. The following theorem was proved in \cite{bolte1}:
\begin{theorem}
For a fixed reference point $z_0\in\mathbb H$ in the Siegel theta function the maps 
\begin{equation}
\Theta:L_0^2(X_{\mathcal O})\longrightarrow \mathcal C_{d(\mathcal O)}\
\text{and} \ \overset{\sim}{\Theta}:\mathcal C_{d(\mathcal O)}\longrightarrow L_0^2(X_{\mathcal O})
\end{equation}
given by
\begin{equation}
\Theta\varphi(\tau):=\int_{\mathcal F_{\mathcal O^1}}\theta(z;\tau)\varphi(z)d\mu(z)
\end{equation}
and
\begin{equation}
\overset{\sim}{\Theta}g(z):=\int_{\mathcal F_d}\overline{\theta(z;\tau)}g(\tau)d\mu(\tau)
\end{equation}
define bounded linear operators preserving Laplace and also Hecke eigenvalues.\\Here $L_0^2(X_{\mathcal O})$  denotes the space of non-constant square integrable automorphic functions on $X_{\mathcal O}:=\mathcal O^1\backslash \mathbb H$ and $\mathcal C_{d(\mathcal O)}$  denotes the space of cusp forms for the congruence subgroup $\Gamma_0(d(\mathcal O))$ whose level is equal to the discriminant of the order $\mathcal O$. $\mathcal F_{\mathcal O^1}$ and $\mathcal F_d$ are the corresponding fundamental domains and $\theta(z;\tau)$ is the Siegel theta function.
\end{theorem}
According to this theorem it follows that \cite{bolte1}:
\begin{theorem}
All eigenvalues of the hyperbolic Laplacian on $L^2(X_{\mathcal O})$ also occur as eigenvalues of the hyperbolic Laplacian on $L^2(X_d)$ where $d=d(\mathcal O)$.
\end{theorem}
For a maximal order $\mathcal O_{max}$ it is shown that the right hand side of Selberg's trace formula for $\mathcal O_{max}^1$ with the trivial representation $\chi=1$ coincides with the right hand side of Selberg's trace formula for the new forms (see formula \eqref{bnj}) of a congruence subgroup of level equal to the discriminant of $\mathcal O_{max}$ with the trivial representation $\chi=1$ (see Theorem \ref{zuio}). Therefore one gets equality of the left hand sides of these trace formulas:
\begin{equation}
\sum_{\varphi_k\in L^2(X_{\mathcal O_{max}})}h(\lambda_k)=\sum_{g_k\in\mathbb C\oplus \mathcal C_d^{new}}h(\lambda_k)
\end{equation}
where $\mathcal C_d^{new}$ denotes the space of new Maass cusp forms for the congruence subgroup of level $d=d(\mathcal O_{max})$. This together with the last theorem leads to 
\begin{theorem}
For a maximal order $\mathcal O_{max}$, the eigenvalues of the hyperbolic Laplacian, including multiplicities, on $X_{\mathcal O_{max}}$ coincide with the Laplace spectrum on the space of new Maass forms for the congruence subgroup $\Gamma_0(d)$, where $d$ is the discriminant of the maximal order $\mathcal O_{max}$.
\end{theorem}
\subsection{Selberg trace formula for new forms}
Consider the congruence subgroup $\Gamma_0(n)\subset SL(2,\mathbb Z)$ with the trivial representation $\chi=1$. The space $\mathcal C_n(\lambda)$ of Maass cusp forms with eigenvalue $\lambda$ can be decomposed into two subspaces of new and old forms $\mathcal C_n(\lambda)=\mathcal C_n^{old}(\lambda)\oplus \mathcal C_n^{new}(\lambda)$. The space $\mathcal C_n^{old}(\lambda)$ is the linear span of all forms with eigenvalue $\lambda$ coming from all overgroups $\Gamma_0(m)\supset\Gamma_0(n)$ with $m\vert n$ and $ \mathcal C_n^{new}(\lambda)$ is defined to be the orthogonal complement of $\mathcal C_n^{old}(\lambda)$. Let us denote the dimension of $\mathcal C_n(\lambda)$ and $\mathcal C_n^{new}(\lambda)$ by $\delta(n,\lambda)$ and $\delta^{new}(n,\lambda)$ respectively. Then the following identity holds \cite{atkin}
\begin{equation}\label{dim-formula}
\delta^{new}(n,\lambda)=\sum_{m\vert n}\beta(\frac{n}{m})\delta(m,\lambda)
\end{equation}
with
\begin{equation}\label{dolom}
\beta(a)=\sum_{l\vert a}\mu(l)\mu(\frac{a}{l}),
\end{equation}
where $\mu(n)$ is the Moebius function, defined on $\mathbb N$ by (see for example \cite{Handbook}, page 639)
\begin{equation}
\mu(n)=\begin{cases}
1&\text{if $n$ is a square-free positive integer}
\\&\text{with an even number of prime factors},\\-1&\text{if $n$ is a square-free positive integer}\\&\text{with an odd number of prime factors},\\0&\text{if n is not square-free}.
\end{cases}
\end{equation}
Identity (\ref{dim-formula}) leads to the following formula \cite{bolte}:
\begin{equation}
\sum_{u_k\in \mathcal C_n^{new}} h(\lambda_k)=\sum_{m\vert n}\beta(\frac{n}{m})\sum_{u_k\in \mathcal C_m}h(\lambda_k)
\end{equation}
This suggests to take the sum  
\begin{equation}\label{rhs tr}
\sum_{u_k\in \mathcal C_n^{new}} h(\lambda_k)+\sum_{m\vert n}\beta(\frac{n}{m})h(\lambda_0)=\sum_{m\vert n}\beta(\frac{n}{m})\sum_{u_k\in \mathbb C\oplus\mathcal C_m}h(\lambda_k)
\end{equation}
for defining the left hand side of the Selberg trace formula for the new forms. 
Thus we arrive at the following definition of a trace formula for new forms for the congruence subgroup $\Gamma_0(n)$ (see also \cite{stroembergsson}):
\begin{equation}\label{bnj}
\sum_{u_k\in \mathcal C_n^{new}} h(\lambda_k)+\sum_{m\vert n}\beta(\frac{n}{m})h(\lambda_0)=I^{new}_n+H^{new}_n+E^{new}_n+P^{new}_n-C^{new}_n.
\end{equation}
The terms  $I^{new}_n$, $H^{new}_n$, $E^{new}_n$ and $P^{new}_n$ are given by
\begin{equation}\label{fff1}
I^{new}_n=\sum_{m\vert n}\beta(\frac{n}{m})I_m
\end{equation}
\begin{equation}\label{f35}
H^{new}_n=\sum_{m\vert n}\beta(\frac{n}{m})H_m
\end{equation}
\begin{equation}\label{fff2}
E^{new}_n=\sum_{m\vert n}\beta(\frac{n}{m})E_m
\end{equation}
\begin{equation}\label{fff3}
P^{new}_n=\sum_{m\vert n}\beta(\frac{n}{m})P_m
\end{equation}
\begin{equation}\label{fff4}
C^{new}_n=\sum_{m\vert n}\beta(\frac{n}{m})C_m
\end{equation}
where $I_m$, $H_m$, $E_m$ and $P_m$ denote the contributions of identity, hyperbolic, elliptic, and parabolic conjugacy classes for the congruence subgroup $\Gamma_0(m)$ and the trivial representation $\chi=1$. The term $C_m$ refers to the contribution of the continuous spectrum for the corresponding group $\Gamma_0(m)$.
Next we recall the following theorem which is proved in \cite{bolte}:
\begin{theorem}\label{zuio}
Let $\mathcal O_{max}$ be a maximal order with discriminant $d(\mathcal O_{max})$ in an indefinite quaternion division algebra over the field of rationals. Then the right hand side of the new form Selberg trace formula for the congruence subgroup $\Gamma_0(n)$ with $n=d(\mathcal O_{max})$ coincides with the right hand side of the Selberg trace formula for the unit group of $\mathcal O_{max}$ that is
\begin{equation}
I_{\mathcal O_{max}^1}=I^{new}_n, \ \ \ \ \ \ E_{\mathcal O_{max}^1}=E^{new}_n, \ \ \ \ \ \ P_n^{new}=C_n^{new}=0
\end{equation}
\begin{equation}
H_{\mathcal O_{max}^1}=H^{new}_n
\end{equation}
where $I_{\mathcal O_{max}^1},\ E_{\mathcal O_{max}^1},\ H_{\mathcal O_{max}^1} $ denote the contributions of the identity, elliptic, and hyperbolic elements in the right hand side of Selberg's trace formula for the unit group of quaternions $\mathcal O_{max}^1$ with the trivial representation $\chi=1$.
\end{theorem}
\subsection{Determinant identities }
For a congruence subgroup $\Gamma_0(n)$, \textit{new form zeta functions}, $Z^{new}_{I,n}(s)$, $Z^{new}_{E,n}(s)$, $Z^{new}_{P,n}(s)$, and $Z^{new}_{H,n}(s)$ respectively corresponding to the identity, elliptic,\\ parabolic, and hyperbolic contributions in the Selberg trace formula for the new forms are defined as a solution of the following differential equations,
\begin{equation}\label{netman1}
\dfrac{d}{ds}I^{new}_n(s)=\dfrac{d}{ds}\dfrac{1}{2s-1}\dfrac{d}{ds}\log Z_{I,n}^{new}(s),
\end{equation}
\begin{equation}
\dfrac{d}{ds}E^{new}_n(s)=\dfrac{d}{ds}\dfrac{1}{2s-1}\dfrac{d}{ds}\log Z_{E,n}^{new}(s),
\end{equation}
\begin{equation}
\dfrac{d}{ds}P^{new}_n(s)=\dfrac{d}{ds}\dfrac{1}{2s-1}\dfrac{d}{ds}\log Z_{P,n}^{new}(s),
\end{equation}
\begin{equation}\label{netman4}
\dfrac{d}{ds}H^{new}_n(s)=\dfrac{d}{ds}\dfrac{1}{2s-1}\dfrac{d}{ds}\log Z_{H,n}^{new}(s)
\end{equation}
where $I^{new}_n(s)$, $E^{new}_n(s)$, $P^{new}_n(s)$ and $H^{new}_n(s)$ are the contributions of the identity, elliptic, parabolic, and hyperbolic elements in the Selberg trace formula for the new form with the test function
\begin{equation}\label{peri}
h(r^2+\frac{1}{4})=\dfrac{1}{r^2+\frac{1}{4}+s(s-1)}-\dfrac{1}{r^2+\beta^2}  \ \ \ \ \ \  \beta>\frac{1}{2}, \ \  s\in \mathbb C.
\end{equation}
Moreover, by definition we have
\begin{equation}\label{copol1}
\dfrac{d}{ds}I_{\mathcal O_{max}^1}=\dfrac{d}{ds}\dfrac{1}{2s-1}\dfrac{d}{ds}\log Z_{I,\mathcal O_{max}^1}(s),
\end{equation}
\begin{equation}
\dfrac{d}{ds}E_{\mathcal O_{max}^1}=\dfrac{d}{ds}\dfrac{1}{2s-1}\dfrac{d}{ds}\log Z_{E,\mathcal O_{max}^1}(s),
\end{equation}
\begin{equation}\label{copol3}
\dfrac{d}{ds}H_{\mathcal O_{max}^1}=\dfrac{d}{ds}\dfrac{1}{2s-1}\dfrac{d}{ds}\log Z_{H,\mathcal O_{max}^1}(s)
\end{equation}
where $I_{\mathcal O_{max}^1}$, $E_{\mathcal O_{max}^1}$, and $H_{\mathcal O_{max}^1}$
are the continuations of the identity, elliptic, and hyperbolic elements in Selberg's trace formula for $\mathcal O_{max}^1$ twisted with the trivial representation, with the test function given in (\ref{peri}).

From identities (\ref{fff1}-\ref{fff3}) it follows that
\begin{equation}\label{zeta new I}
Z_{I,n}^{new}(s):=\Pi_{m\vert n}Z_{I,m}^{\beta(\frac{n}{m})}(s),
\end{equation}
\begin{equation}\label{zeta new E}
Z_{E,n}^{new}(s):=\Pi_{m\vert n}Z_{E,m}^{\beta(\frac{n}{m})}(s),
\end{equation}
\begin{equation}\label{zeta new P}
Z_{P,n}^{new}(s):=\Pi_{m\vert n}Z_{P,m}^{\beta(\frac{n}{m})}(s),
\end{equation}
\begin{equation}\label{zeta new}
Z_{H,n}^{new}(s):=\Pi_{m\vert n}Z_m^{\beta(\frac{n}{m})}(s),
\end{equation}
where $Z_{I,m}(s)$, $Z_{E,m}(s)$, $Z_{P,m}(s)$, and $Z_{m}(s)$ denote the zeta function for the congruence subgroup $\Gamma_0(m)$ with the trivial representation corresponding to the contributions $I$, $E$, $P$, and $H$.

On the other hand Theorem \ref{zuio} together with formulae (\ref{netman1}-\ref{netman4}), (\ref{copol1}-\ref{copol3}), and (\ref{zeta new I}-\ref{zeta new}), up to a nonzero holomorphic factor, lead to
\begin{theorem}
The following formulae hold
\begin{equation}\label{zeta new I1}
Z_{I,\mathcal O_{max}^1}(s)=Z_{I,n}^{new}(s)=\Pi_{m\vert n}Z_{I,m}^{\beta(\frac{n}{m})}(s),
\end{equation}
\begin{equation}\label{zeta new E1}
Z_{E,\mathcal O_{max}^1}(s)=Z_{E,n}^{new}(s)=\Pi_{m\vert n}Z_{E,m}^{\beta(\frac{n}{m})}(s),
\end{equation}
\begin{equation}\label{zeta new H1}
Z_{H,\mathcal O_{max}^1}(s)=Z_{H,n}^{new}(s)=\Pi_{m\vert n}Z_m^{\beta(\frac{n}{m})}(s)
\end{equation}
where  $n=d(\mathcal O_{max})$.
\end{theorem}
Finally, we can also connect the automorphic Laplacians for the different groups by their regularized determinants. Indeed from (\ref{determinant expression of Laplacian}) and the previous theorem we get
\begin{theorem}\label{zzz45}
For the congruence subgroup $\Gamma_0(m)$ with the trivial representation, let $h_m$ be the number of inequivalent cusps, $\Phi_m(s)$ be the scattering matrix, $K_m=tr\Phi_m(\frac{1}{2})$. Also let $\mathcal O_{max}$ be a maximal order with discriminant $d(\mathcal O_{max})$ in an indefinite quaternion division algebra over the field of rationals.
Then the following identity, up to a nonzero holomorphic factor, holds
\begin{equation}
F(s){\det}(A(\mathcal O_{max}^1)-s(1-s))=\prod_{m\vert n}{\det}(A(\Gamma_0(m))-s(1-s))^{\beta(\frac{n}{m})}
\end{equation}
where 
\begin{equation}
F(s)=(s-\frac{1}{2})^{-\sum_{m\vert n}K_m\beta(\frac{n}{m})}\prod_{m\vert n}\det\Phi_m(s)^{\beta(\frac{n}{m})},
\end{equation}
$n=d(\mathcal O_{max})$, and $\beta(u)$ is given in (\ref{dolom}).
\end{theorem}
\appendix
\section{Spectral theory of automorphic functions}
In this appendix we introduce briefly some definitions and results in the spectral theory of the automorphic Laplacian which we need in this paper.
\paragraph{Hyperbolic plane}
One of the models of the hyperbolic plane $\mathbb H$ is the upper half plane,
\begin{equation}\label{upper half plane}
\left\lbrace  x + i y\in\mathbb C \ | \ y > 0 \right\rbrace 
\end{equation} 
equipped with the Poincar\'e metric,
\begin{equation}\label{metric}
 ds^2=\dfrac{dx^2+dy^2}{y^2}
\end{equation}
and the Poincar\'e measure,
\begin{equation}\label{measure}
d\mu(z)=\dfrac{dxdy}{y^2}.
\end{equation} 
The Laplace operator $L$ on $\mathbb H$ associated to the Poincar\'e metric is called the hyperbolic Laplacian. It has in Cartesian coordinates the following explicit form,
\begin{equation}\label{hyperbolic Laplacian}
L=-y^2(\dfrac{\partial^2}{\partial x^2}+\dfrac{\partial^2}{\partial y^2})
\end{equation}
where we defined the Laplacian with minus sign.
\paragraph{Isometries of $\mathbb H$}
The group of all orientation preserving isometries of $\mathbb H$ is identified with the group $G=PSL(2,\mathbb R)$,
\begin{equation}
G=\left\lbrace g=\left( \begin{array}{cc}
a&b\\
c&d\\
\end{array}
\right)\ \vert \ \det g=1,\quad a,b,c,d\in\mathbb R\right\rbrace/\left\lbrace \pm 1\right\rbrace , 
\end{equation}
acting on $\mathbb H$ by linear fractional transformations,
\begin{eqnarray}
&G\times\mathbb H\longrightarrow \mathbb H\nonumber&\\&(g,z):=g z=\dfrac{az+b}{cz+d}.&
\end{eqnarray}

The elements of $G$ except the identity are classified in three disjoint classes according to their traces as a matrix. An element $g \in G$ is called elliptic, parabolic or hyperbolic if $\vert\,tr(g)\,\vert<2$, $\vert\,tr(g)\,\vert=2$ or $\vert\,tr(g)\,\vert>2$, respectively.\\ 
Since the action of $G$ can be extended by continuity to $\mathbb H\cup\mathbb R\cup\left\lbrace \infty\right\rbrace $, this classification can be reformulated as the following: elliptic elements have only one fixed point on $\mathbb H$, the parabolic ones have a unique fixed point on $\mathbb R \cup\left\lbrace \infty\right\rbrace $ and the hyperbolic elements have two distinct fixed points on $\mathbb R \cup \left\lbrace \infty\right\rbrace $.
\paragraph{Fuchsian groups}
A discrete subgroup $\Gamma\subset PSL(2,\mathbb R)$ is called a Fuchsian group.\\
A cusp of a Fuchsian group $\Gamma$ is defined to be the fixed point of a parabolic element of $\Gamma$.\\
A fundamental domain $F$ for a Fuchsian group $\Gamma$ is defined to be the closure of a domain $U\in\mathbb H$, including all non $\Gamma$-equivalent points of $\mathbb H$ such that 
\begin{displaymath}
\mathbb H=\underset{\gamma\in\Gamma}{\cup}\gamma F.
\end{displaymath}
The volume of the quotient space $\Gamma \setminus\mathbb{H}$, represented by the fundamental domain $F$, is given by,
\begin{equation}
vol(\Gamma \setminus\mathbb{H}):=\vert F\vert:=\int_Fd\mu(z).
\end{equation}
A Fuchsian group $\Gamma$ for which the volume of $\Gamma \setminus \mathbb{H}$ is finite, is called Fuchsian group of the first kind (or cofinite).\\
If the surface $\Gamma \setminus\mathbb{H}$ is compact, the group $\Gamma$ is called cocompact.\\
A Fuchsian group  $\Gamma$ of the first kind is determined by \cite{Alexei}
\begin{itemize}
\item[1] a finite number $2g$ of hyperbolic generators, $A_1,B_1,\ldots, A_g,B_g$
\item[2] a finite number $l$ of elliptic generators, $R_1,\ldots, R_l$
\item[3] a finite number $h$ of parabolic generators, $S_1,\ldots, S_h$
\end{itemize}
such that the following relations hold,
\begin{eqnarray}
&[A_1,B_1]\ldots[A_g,B_g]S_1\ldots S_hR_1\ldots R_l=E,\nonumber&\\&R_1^{m_1}=Id,\ldots ,R_l^{m_l}=Id.&
\end{eqnarray}
Here $[,]$ denotes the commutator and $m_j\in \mathbb N\cup\left\lbrace 0\right\rbrace $ is the order of the elliptic element $R_j$. The signature of a group $\Gamma$, determined by a set of generators is defined to be the set of numbers,
\begin{equation}\label{signature}
(g;m_1,\ldots,m_l;h)
\end{equation}
which is a topological invariant of the group as is the fundamental group of the corresponding surface. We note that $g$ is the genus of the surface $\Gamma\setminus\mathbb H$ and $h$ is the number of cusps of the surface. Moreover the group $\Gamma$ is cocompact if and only if $h=0$. 

For a Fuchsian group of the first kind with signature as in (\ref{signature}), the volume of $\Gamma \setminus\mathbb{H}$ is given by the Gauss-Bonnet formula \cite{Alexei},
\begin{equation}
\vert F\vert=2\pi\left( 2g-2+\sum_{j=1}^l(1-\frac{1}{m_j})+h\right).
\end{equation}

Modular group and its congruence subgroups are examples of the Fuchsian groups of the first kind. The modular group is defined by
\begin{equation}
SL(2,\mathbb Z)=\left\lbrace \left( \begin{array}{cc}
a&b\\
c&d\\
\end{array}
\right)\vert \ ad-bc=1, \ a,b,c,d\in \mathbb Z\right\rbrace.
\end{equation}
This group has the signature $(0,3,2,1)$ and it is generated by the parabolic element 
\begin{equation}
T=\left( \begin{array}{cc}
1&1\\
0&1\\
\end{array}
\right)
\end{equation}
and the elliptic element
\begin{equation}
Q=\left( \begin{array}{cc}
0&1\\
-1&0\\
\end{array}
\right).
\end{equation}
The principal congruence group of level $N\in\mathbb N$ is a subgroup of modular group defined by \cite{Rankin}
\begin{equation}
\Gamma(N)=\left\lbrace g\in SL(2,\mathbb Z)\,\,\vert\,\, g\equiv\left( \begin{array}{cc}
1&0\\
0&1\\
\end{array}
\right)\mod N\right\rbrace.
\end{equation}
A subgroup of $SL(2,\mathbb Z)$ containing $\Gamma(N)$ is called a congruence group \cite{Rankin}. The Hecke congruence group of level $N$ is an example of congruence subgroups defined by \cite{Rankin}
\begin{equation}
\Gamma_0(N)=\left\lbrace g\in SL(2,\mathbb Z)\,\,\vert\,\, g_{21}\equiv 0\mod N\right\rbrace.
\end{equation}
\paragraph{Automorphic Laplacian}
The Fuchsian groups allow us to define the notion of automorphy of functions and operators on $\mathbb H$.\\
Let $V$ denote a Hermitian space of dimension $n:=\dim V$ and $\chi$ be a unitary representation of $\Gamma$ on $V$.\\
Then a vector valued function $f:\mathbb H\rightarrow V$ such that,
\begin{equation}
f(\gamma z)=\chi(\gamma)f(z), \ \gamma \in\Gamma,
\end{equation}
is called an automorphic function with respect to $\Gamma$ and $\chi$.
We denote by $\mathfrak{H}=\mathfrak{H}(\Gamma;\chi)$, the Hilbert space of automorphic functions with respect to $\Gamma$ and $\chi$, square integrable on the fundamental domain $F$
with the scalar product given by,
\begin{equation}
( f, h )=\int_F\left\langle f(z),h(z)\right\rangle_V d\mu(z),\ \ f,h\in\mathfrak{H}
\end{equation}
where $\left\langle ,\right\rangle $ denotes the Hermitian form on $V$.

Let $\mathfrak{D}$ be the dense domain in $\mathfrak{H}(\Gamma;\chi)$ consisting of the functions $g$ which are restrictions of functions $f$ on $\mathbb H$ to $F$ such that
\begin{itemize}
  \item  $f\in \oplus_{i=1}^{\dim V}C^{\infty}(\mathbb{H})$
  
  \item $f(z)=\chi(\gamma)f(\gamma z), \ \gamma \in\Gamma$
  
  \item $f$ and $Lf$ belong to $\mathfrak{H}(\Gamma;\chi)$
\end{itemize}
An operator $\tilde{A}$ is defined by
\begin{equation}
\widetilde{A}f=Lf,\ f\in\mathfrak{D}
\end{equation}
which is symmetric and non-negative. It turns out that this operator
is essentially self adjoint \cite{Alexei2, Faddeev}. The automorphic Laplacian $A$ is defined
as the unique self-adjoint extension (Fredrichs extension) of the  operator
$\widetilde{A}$ on $\mathfrak{H}$. Thus the automorphic
Laplacian $A$ is a self-adjoint non-negative unbounded operator.
For more detail and proofs of the assertions see \cite{Alexei2, Faddeev}.
\paragraph{Spectral decomposition}
Let $\Gamma$ be a Fuchsian group of the first kind with inequivalent cusps $x_\alpha$, $1\leq \alpha\leq h$. An element $\gamma\in\Gamma$ is primitive if it can not be written as a power of another element of the group. Let $S_\alpha$ be a primitive parabolic element leaving the cusp $x_\alpha$ invariant $S_\alpha x_\alpha=x_\alpha$. Then $S_\alpha$ is the generator of the maximal stabilizer group of the cusp $x_\alpha$, denoted by $\Gamma_{\alpha}$. 
\begin{definiton}
A finite dimensional representation $\chi$ of $\Gamma$ is called singular in the cusp $x_\alpha$ if
\begin{equation}
\dim\ker(\chi(S_\alpha)-1_V)=0
\end{equation}
where $1_V$ is the identity operator in $V$. Otherwise we say that the representation $\chi$ is non-singular in the cusp $x_\alpha$.
\end{definiton}
\begin{definiton}
If $\chi$ is non-singular at a cusp $x_\alpha$, we say that the cusp is open.
\end{definiton}
\begin{definiton}
A representation $\chi$ is called singular if it is singular in all cusps. 
\end{definiton}
\begin{definiton}
The representation $\chi$ is called non-singular if it is non-singular at least in one cusp, that is if at least one cusp is open.
\end{definiton}
\begin{remark}
The definition of singularity of a representation is opposite to Selberg's definition (see \cite{Selberg}, \cite{Alexei}) but it is more reasonable from the point of view of mathematical physics and it is due to E. Balslev \footnote{Private communaication between Alexei Venkov and Erik Balslev}. If the spectrum of the hyperbolic Laplacian on a non-compact surface is purely discrete, that means the situation is singular. In non-singular situation, we have a continuous spectrum and may be a discrete one. 
\end{remark}
For cocompact groups and also for non-cocompact groups with singular representation, the automorphic Laplacian $A=A(\Gamma;\chi)$ has only a purely discrete spectrum in $\mathfrak{H}=\mathfrak{H}(\Gamma;\chi)$, spanned by the corresponding discrete set of eigenfunctions of $A$ (see \cite{Alexei} pages 17 and 18).

For non-cocompact groups with a non-singular representation $\chi$ the automorphic Laplacian $A(\Gamma;\chi)$ in $\mathfrak{H}=\mathfrak{H}(\Gamma;\chi)$ has a continuous spectrum and may be a discrete one. The continuous spectrum is described by the Eisenstein series analytically continued to the spectrum \cite{Alexei, Alexei2}. 

Before proceeding further, we need some notations. The subspace $V_\alpha\subset V$ for $x_\alpha$ a cusp with $S_\alpha x_\alpha=x_\alpha$ is defined by
\begin{equation}
V_\alpha:=\left\lbrace v\in V \ \vert \ \chi(S_\alpha)v=v\right\rbrace
\end{equation}
and we put $k_\alpha:=\dim V_\alpha$. We denote an orthonormal basis of $V_\alpha$ by $\left\lbrace e_l(\alpha)\right\rbrace_{l=1}^{k_\alpha}$. The degree $k$ of non-singularity of the representation $\chi$ is defined as
\begin{equation}\label{ksk}
k:=k(\Gamma;\chi):=\sum_{\alpha=1}^hk_\alpha.
\end{equation}
We denote by $P_\alpha$ the orthogonal projection of $V$ onto $V_\alpha$.\\
For every open cusp $x_\alpha$, $1\leq \alpha\leq h$, the Eisenstein series $E_\alpha(.,s):\mathbb H\rightarrow V_\alpha$ is a $k_\alpha$ dimensional vector whose components $E_{\alpha, l}$ are defined as an absolutely convergent series in the domain $\text{Re}(s)>1$, uniformly convergent in $z$ on any compact subsets of $\mathbb H$ by
\begin{equation}
E_{\alpha, l}(z,s)=\sum_{\sigma\in\Gamma_\alpha\backslash\Gamma}(\text{Im}(\sigma_\alpha^{-1}\sigma z))^s\chi^*(\sigma)e_l(\alpha),\quad 1\leq l\leq k_\alpha\quad\text{Re}(s)>1
\end{equation}
where $\chi^*$ denotes the adjoint of $\chi$ as operators in the Hermitian space $V$, $\sigma_\alpha\in PSL(2,\mathbb R)$ denotes the element such that $\sigma_\alpha \infty=x_\alpha$ and $e_l(\alpha)$ is an element of the orthonormal basis of $V_\alpha$. In the domain $\text{Re} s>1$, the Eisenstein series has the following properties \cite{Alexei, Alexei2}:
\begin{itemize}
\item[\textbf{1}] $E_\alpha(z,s)$ is holomorphic in $s$.
\item[\textbf{2}] For fixed $s$, $LE_\alpha(z,s)=s(1-s)E_\alpha(z,s)$.
\item[\textbf{3}] For fixed $s$, $E_\alpha(z,s)$ is automorphic relative to $\Gamma$ and $\chi$.
\item[\textbf{4}] The zero-th order term of the Fourier expansion of the components \\$E_{\alpha,l}(z,s)$ of $E_{\alpha}(z,s)$ at a cusp $x_\beta$ is given by,
\begin{equation}
\delta_{\alpha,\beta}y^se_l(\alpha)+\phi_{\alpha l,\beta}(s)y^{1-s}.
\end{equation}
\end{itemize}
The elements of the automorphic scattering matrix 
\begin{equation}\label{sctnat}
\Phi(s)=\Phi(s;\Gamma;\chi):=\left\lbrace\Phi_{bd}(s)\right\rbrace_{b,d=1}^{k(\Gamma;\chi)}
\end{equation}
are given by
\begin{equation}
\Phi_{bd}(s)=\Phi_{\alpha l,\beta k}(s)=\langle e_k(\beta),\phi_{\alpha l,\beta}(s)\rangle_V
\end{equation}
where $\langle.\rangle_V$ denotes the inner product in $V$. The indexes are defined by 
\begin{equation}
b=k_1+k_2+\ldots+k_{\alpha-1}+l,\quad d=k_1+k_2+\ldots+k_{\beta-1}+k
\end{equation}
such that
\begin{equation}
1\leq\alpha,\beta\leq h,\quad 1\leq l\leq k_\alpha,\quad 1\leq k\leq k_\beta.
\end{equation}
Then the following Theorem holds \cite{Alexei, Alexei2}
\begin{theorem}
The following assertions hold,
\begin{itemize}
\item[\textbf{a}] The scattering matrix $\Phi(s)$ and the Eisenstein series $E_\alpha(z,s)$ admit meromorphic continuations to the entire $s$-plane, the order of these meromorphic functions is not greater than four.
\item[\textbf{b}] In the half plane $\text{Re}(s)\geq\frac{1}{2}$, $\Phi(s)$ and  $E_\alpha(z,s)$ have only a finite number of common simple poles $s_j\in(\frac{1}{2},1]$ such that $\lambda_j:=s_j(1-s_j)$ is a real eigenvalue of $A(\Gamma;\chi)$.
\item[\textbf{c}] The scattering matrix fulfills the functional equations
\begin{equation}
\Phi(s)=\overline{\Phi(s)}^T,\quad \Phi(s)\Phi(1-s)=Id.
\end{equation}
\item[\textbf{d}] The Eisenstein series satisfies the following functional equation,
\begin{equation}
E(z,s)=\Phi(s)E(z,1-s)
\end{equation}
where
\begin{equation}
E(z,s)=\left( E_1(z,s),\ldots,E_h(z,s)\right)^T.
\end{equation}
\end{itemize}
\end{theorem}

An eigenfunction of the automorphic Laplacian,
\begin{equation}\label{ldr}
A(\Gamma;\chi)f=s(1-s)f,\quad f\in\mathfrak{H}(\Gamma;\chi)
\end{equation}
with vanishing constant term of the Fourier expansion at each cusp is called a cusp form with spectral parameter $s$ \cite{Alexei}.
A cusp form $f$ decays exponentially fast in all cusps and $f$ is orthogonal to the Eisenstein series \cite{Alexei}. The spectral parameters of the cusp forms are a discrete set of points $s$, lying on the line $\text{Re}s=\frac{1}{2}$ and in the interval $\left( \frac{1}{2},1\right]$ \cite{Alexei, Alexei2, Phillips-Lax}. We denote the space of all cusp forms by $\mathfrak H_0$.

Let $\Theta_0$ be the finite dimensional space, spanned by the residues of Eisenstein series at finitely many poles in $(\frac{1}{2},1]$ and let $\Theta_1$ be the orthogonal complement of $\mathfrak H_0\oplus\Theta_0$ in $\mathfrak H$.
The automorphic Laplacian $A(\Gamma;\chi)$ splits the space $\mathfrak H$ into three invariant subspaces defined above \cite{Alexei},
\begin{equation}
\mathfrak H=\mathfrak H_0\oplus\Theta_0\oplus\Theta_1.
\end{equation}
The automorphic Laplacian $A=A(\Gamma;\chi)$ has a purely discrete spectrum on the space $\mathfrak H_0\oplus\Theta_0\subset\mathfrak H$ \cite{Alexei}.
The spectrum of $A(\Gamma;\chi)$ on $\Theta_1$ is absolutely continuous, filling up the semi-axis $\lambda\geq\frac{1}{4}$ with multiplicity $k(\Gamma;\chi)$. In this case the eigenfunctions are described by $E_\alpha(z,s=\frac{1}{2})$ which are not in $\Theta_1$ \cite{Alexei}.
\paragraph{Selberg trace formula}
The Selberg trace formula is an identity connecting the spectrum of the automorphic Laplacian on $\Gamma\backslash\mathbb{H}$ to the geometry of this surface. We recall it from \cite{Alexei}. First we need some notations. Let $\left\lbrace e_l(\alpha)\right\rbrace_{l=1}^n$ be a basis of $V$ in which $\chi(S_\alpha)(1_V-P_\alpha)$ is diagonal,
\begin{equation}
\chi(S_\alpha)(1_V-P_\alpha)e_l(\alpha)=\nu_{\alpha l}e_l(\alpha).
\end{equation}
Then the following holds
\begin{equation}
\nu_{\alpha l}=\begin{cases}
0&e_l(\alpha)\in V_\alpha,\\ \exp(2\pi i\theta_{\alpha l})& e_l(\alpha)\in V\ominus V_\alpha.
\end{cases}
\end{equation}
where $0<\theta_{\alpha l}<1$. 

\begin{theorem}\label{Sel-tr}
Let $\overset{\sim}{h}(r):=h(r^2+\frac{1}{4})$ be a function of a complex variable $r$ which satisfies the following conditions:

\begin{itemize}
\item As a function of $r$, $\overset{\sim}{h}(r)$ is holomorphic in the strip \\$\left\lbrace r\in \mathbb C:\vert Im(r)\vert<\frac{1}{2}+\varepsilon\right\rbrace$ for some $\varepsilon>0$. 
\item In that strip, $\overset{\sim}{h}(r)=O((1+\vert r^2\vert)^{-1-\varepsilon})$ and all the series and integrals appearing below converge absolutely.
\end{itemize}

Then the following identity holds
\begin{equation}\label{trace formula}
\sum_{k=0}^{\infty}h(\lambda_k)+C=I+H+E+P
\end{equation}
where $\left\lbrace \lambda_n \ \vert\ 0=\lambda_0<\lambda_1\leq \lambda_2\leq \ldots\right\rbrace$ are the discrete eigenvalues of $A(\Gamma;\chi)$. Here $C$ corresponds to the continuous part of the spectrum given by
\begin{equation}
C=C(\overset{\sim}{h}(r);\Gamma;\chi)=-\frac{1}{4\pi}\int_{-\infty}^{\infty}\dfrac{\varphi'}{\varphi}(\frac{1}{2}+ir;\Gamma;\chi)h(r^2+\frac{1}{4})dr+\dfrac{K_0}{4} h(\dfrac{1}{4})
\end{equation}
where $\varphi$ denotes the determinant of the scattering matrix \ $\Phi(s)$ and \\$K_0=tr(\Phi(\frac{1}{2};\Gamma;\chi))$. On the right hand side of (\ref{trace formula})  $I$ corresponds to the contribution of the identity element of the group which is given by
\begin{equation}\label{contribution of identity}
I=I(\overset{\sim}{h}(r);\Gamma;\chi)=\dfrac{n\vert F\vert}{4\pi}\int_{-\infty}^\infty r \tanh(\pi r)\ h(r^2+\frac{1}{4})dr
\end{equation}
The term $H$ denotes the contribution of the hyperbolic conjugacy classes and is given by
\begin{equation}\label{HHH}
H=H(\overset{\sim}{h}(r);\Gamma;\chi)=\sum_{\left\lbrace P\right\rbrace _\Gamma}\sum_{m=1}^\infty \dfrac{\text{tr}_V\chi^m(P)logN(P)}{N(P)^{\frac{m}{2}}-N(P)^{-\frac{m}{2}}}g(m \ logN(P))
\end{equation}
where $\left\lbrace P\right\rbrace _\Gamma$ denotes the primitive hyperbolic conjugacy classes and the function $g$ appears through the Selberg transformation:
\begin{equation}\label{g(u)}
g(u)=\frac{1}{2\pi}\int_{-\infty}^{\infty} e^{-iru}h(r^2+\frac{1}{4})dr
\end{equation}
The next term $E$ refers to the contribution of the elliptic elements and is given by a summation over primitive elliptic conjugacy classes $\left\lbrace R\right\rbrace _\Gamma$ of order $\nu$ 
\begin{equation}\label{contribution of elliptic}
E=E(\overset{\sim}{h}(r);\Gamma;\chi)=\frac{1}{2}\sum_{\left\lbrace R\right\rbrace _\Gamma}\sum_{m=1}^{\nu-1}\dfrac{\text{tr}_V\chi^k(R)}{\nu \sin\pi m/\nu}\int_{-\infty}^\infty \dfrac{exp(-2\pi rm/\nu)}{1+exp(-2\pi r)}h(r^2+\frac{1}{4})dr
\end{equation}
Finally, the last term comes from the parabolic conjugacy classes given by
\begin{eqnarray}\label{contribution of parabolic}
&&P=P(\overset{\sim}{h}(r);\Gamma;\chi)=-(k(\Gamma; \chi)\ln2+\sum_{\alpha=1}^h\sum_{l=k_\alpha+1}^n\ln\vert 1-\exp(2\pi i\theta_{\alpha l})\vert) g(0)\nonumber\\&&
-\dfrac{k(\Gamma;\chi)}{2\pi}\int_{-\infty}^\infty \psi(1+ir)h(r^2+\frac{1}{4})dr+\dfrac{k(\Gamma;\chi)}{4}h(\dfrac{1}{4})
\end{eqnarray}
where $h$ is the number of cusps and $\psi$ is the di-gamma function.
\end{theorem}

In the case of cocompact groups like the unit group of quaternion algebras there is no continuous spectrum and no parabolic element and the trace formula (\ref{trace formula}) reduces to 
\begin{equation}\label{Selberg trace formula for cocompact groups}
\sum_{k=0}^{\infty}h(\lambda_k)=I+H+E.
\end{equation}
\paragraph{Weyl-Selberg Formula}
In this subsection we recall briefly the Weyl-Selberg formula which clarifies the asymptotics of the distribution of the eigenvalues of the automorphic Laplacian.

The discrete eigenvalues $\left\lbrace \lambda_n \ \vert\ 0=\lambda_0<\lambda_1\leq \lambda_2\leq \ldots\right\rbrace$ of the automorphic Laplacian $A(\Gamma;\chi)$ can be represented as $\lambda_j=\frac{1}{4}+r_j^2$, $r_j\in \mathbb R$. Then the distribution function for the eigenvalues of $A(\Gamma;\chi)$ is defined by 
\begin{equation}\label{malbar}
N(T;\Gamma;\chi)=\sharp\left\lbrace \lambda_j\,\vert\, \vert r_j\vert<T\right\rbrace\geq0.
\end{equation}
The continuous spectrum is measured by \cite{Iwaniec}
\begin{equation}\label{halbar}
M(T;\Gamma;\chi)=\frac{1}{4\pi}\int_{-T}^T-\dfrac{\varphi'}{\varphi}(\frac{1}{2}+ir;\Gamma;\chi)dr\geq0.
\end{equation}
Weyl-Selberg formula is given by (\cite{Alexei} page 52)
\begin{equation}
\begin{split}
N(T;\Gamma;\chi)+M(T;\Gamma;\chi)=\dfrac{\vert F\vert\dim V}{4\pi}T^2-\dfrac{k(\Gamma;\chi)}{\pi}T\log T+\\\frac{1}{\pi}\left[ k(\Gamma;\chi)(1-\log2)-\sum_{\alpha=1}^h\sum_{l=k_\alpha+1}^{\dim V}\log\vert1-\exp(2\pi i\theta_{\alpha l}\vert)\right]T\\+O(\dfrac{T}{\log T}),\qquad T\rightarrow\infty
\end{split}
\end{equation}
where all notations are introduced in the previous section. We note that if the group is cocompact or the representation $\chi$ is singular then the second term in the left hand side of the equality is absent. 

The functions $N(T;\Gamma;\chi)$ and $M(T;\Gamma;\chi)$ can be estimated separately as follows (see \cite{Fischer}, page 138, Corollary 3.3.14),
\begin{equation}\label{mnb4}
N(T;\Gamma;\chi)=O(T^2),\qquad T\rightarrow\infty
\end{equation}
and
\begin{equation}\label{mnb5}
M(T;\Gamma;\chi)=O(T^2),\qquad T\rightarrow\infty.
\end{equation}
In the case of congruence subgroups with trivial representation there are more precise estimates, that is (see \cite{Iwaniec}, page 159),
\begin{equation}\label{asn}
N(T;\Gamma;\chi)=\dfrac{\vert F\vert}{4\pi}T^2+O(T\log T),\qquad T\rightarrow\infty
\end{equation}
and
\begin{equation}\label{asm}
M(T;\Gamma;\chi)=O(T\log T),\qquad T\rightarrow\infty.
\end{equation}


\begin{thebibliography}{1}
\bibitem{Albin}
P. Albin,C. L. Aldana, and F. Rochon, \emph{Ricci flow and the determinant of the Laplacian on non compact surfaces}, arxiv:0909.0807 (math.DG)


\bibitem{Aros}
R. Aros, D. E. Diaz, \emph{Functional determinants,generalized BTZ geometries and Selberg zeta function}, J.Phys.A.Math.Theor. \textbf{43}, (2010), 205402 (16pp)

\bibitem{atkin}
  A.O.L.Atkin, J.Lehner, \emph{Hecke Operators on $\Gamma_0(m)$},
  Math. Ann. \textbf{185} (1970), 134-160.


\bibitem{Balslev-venkov3}
E. Balslev, A. Venkov, \emph{Spectral theory of Laplacians for Hecke groups with primitive character}, Acta Mathematica, \textbf{186}, Nr. 2, (2001), 155-217.


\bibitem{Balslev-venkov4}
E. Balslev, A. Venkov, \emph{Correction to "Spectral theory of Laplacians for Hecke groups with primitive character},Acta Mathematica, \textbf{192}, Nr. 1, (2004), 1-3.


\bibitem{Balslev}
E. Balslev, \emph{Spectral deformation of Laplacian on hyperbolic manifolds}, Comm. Analysis and Geometry, \textbf{5}, no.2 (1997), 213-247.

\bibitem{Balslev-venkov}
  E. Balslev, A. Venkov, \emph{The Weyl law for subgroups of the modular group},
  Geom. Funct. Anal.(GAFA),\textbf{8} (1998), 437-465.


\bibitem{Balslev-venkov2}
  E. Balslev, A. Venkov, \emph{On the Relative Distribution of Eigenvalues of Exceptional Hecke Operators and Automorphic Laplacians},
  Saint Petersburg Mathematical Journal),\textbf{17}, Nr. 1, (2006), 1-37.


\bibitem{Barnes}
E. W. Barnes, \emph{The theory of double gamma function}, Philos. Trans. Roy. Soc. A \textbf{169}, (1901), 265-388.


\bibitem{Bernardo}
J. M. Bernardo, \emph{Algorithm AS 103: Psi (Digamma) Function}, Journal of the Royal Statistical Society. Series C (Applied Statistics) \textbf{25}, No. 3, (1976), 315-317.




\bibitem{bolte}
  J.Bolte, S.Johansson, \emph{A Spectral correspondence for Maass waveforms},
  GAFA, Geom. funct. anal. \textbf{9} (1999), 1128--1155.


\bibitem{bolte1}
  J.Bolte, S.Johansson, \emph{Theta-lifts of Maass waveforms},
  in "Emerging applications of number theory" (D.A. Hejhal, F. Chung, J. Friedman, M. C. Gutzwiller, A. Odlyzko, eds.), IMA\textbf{ 109} , Springer-Verlag, New York (1998), 39--72.


\bibitem{Bolte-Steiner}
J. Bolte, F. Steiner, \emph{Determinants of Laplace like operators on Riemann surfaces}, Comm.Math.Phys. \textbf{130}, (1990), 581-597.

\bibitem{bowen}
  R.Bowen and C.Series, \emph{Markov maps associated with Fuchsian groups},
  Publ. IHES \textbf{50} (1979), 401--418.

\bibitem{Brocker}
U. Brocker, \emph{On Selberg zeta functions,topological zeroes and determinant formulas}, Preprint,March 1994

\bibitem{Bytsenko0}
A. A. Bytsenko, \emph{Heat-Kernel asymptotics of locally symmetric spaces of rank one and Chern-Simons Invariants}, Nuclear Physics B (Proc. Suppl), \textbf{104} (2002), 127-134


\bibitem{Bytsenko}
A. A. Bytsenko,M. E. X. Guimaraes,F. L. Williams, \emph{Remarks on the spectrum and truncated heat kernel of the BTZ black hole}, atxiv:hep-th/0609102


\bibitem{Bytsenko2}
A. A. Bytsenko, E. Elizalde, S. A. Sukhanov, \emph{Hyperbolic topological invariants and the black hole geometry}, arxiv:hep-th/0302134, Feb 2003

\bibitem{Bytsenko3}
A. A. Bytsenko, A. E. Goncalves, F. L. Williams, \emph{Chern-Simons invariants of closed hyperbolic 3-manifolds}, arxiv:hep-th/9908037, Aug 1999.


\bibitem{Cartier-Voros}
P. Cartier, A. Voros, \emph{Une nouvelle interpretation la formule des traces da Selberg}, Alexander Grothendieck Festschrift, \textbf{2}, (1990), 1-67.


\bibitem{Chang}
C.-H. Chang, D. Mayer, \emph{Thermodynamic Formalism and Selberg's zeta function for modular groups}, Regul Chaotic Dyn, \textbf{5} (2000), 281-312.

\bibitem{Colin1}
Y. Colin de Verdiere, \emph{Pseudo Laplacian \textbf{I}}, Ann. Inst. Fourier \textbf{32}, (1983).

\bibitem{Colin}
Y. Colin de Verdiere, \emph{Pseudo Laplacian \textbf{II}}, Ann. Inst. Fourier \textbf{33}, (1983).



\bibitem{Deitmar}
A. Deitmar, \emph{A determinant formula for the generalized Selberg zeta function}, Quart.J.Math.Oxford(2), \textbf{47}, (1996), 435-453.




\bibitem{Deshouillers}
J. M. Deshouillers, H. Iwaniec, R. Phillips, P. Sarnak, \emph{Maass cusp forms}, Proc. Nat. Acad. Sci. USA, \textbf{82}, (1985), 3533-3534.


\bibitem{Diaz}
D. E. Diaz, \emph{Holographic formula for the determinant of the scattering operator in thermal AdS}, J.Phys.AMath.Theor., \textbf{42}, (2009), 365401(11 pp).


\bibitem{DJGV}
R. Dijkgraaf, J. Maldacena, G. Moore, E. Verlinde \emph{A black hole Farey tail}, arXive:hep-th/0005003v3.


\bibitem{Efrat}
I. Efrat, \emph{Determinants of Laplacian on Surfaces of Finite Volume}, Commun. math. Phys., \textbf{324}, (1991), 443-451.


\bibitem{Faddeev}
Ludwig Faddeev, \emph{Expansion in Eigenfunctions of the Laplace operator on the fundamental domain of a discrete group on the Lobachovsky plane}, Truday Moscov. Mat. Obsc. \textbf{17}, (1967), 323-350; English transl. in Trans. Moscow Math. Soc. \textbf{17} (1967)


\bibitem{Faddeev-Popov}
L. D. Faddeev, V. N. Popov, \emph{Feynman diagrams for the Yang-Mills fields}, Phys.Lett.B \textbf{25}, (1967), 29-30.

\bibitem{Fischer}
Juergen Fischer, \emph{An Approach to the Selberg Trace Formula via the Selberg Zeta-Function}, Lecture Notes in Mathematics \textbf{1253}, Springer, (1987).


\bibitem{Fock}
V. A. Fock, \emph{no titel available}, Izvestiya Akad. Nauk USSR, OMEN, p557 (1937)



\bibitem{Freed}
D. S. Freed, \emph{Remarks on Chern-Simons Theory}, Bull. Amer. Math. Soc. \textbf{46} No.2 (2009), 221-254



\bibitem{Gilbert}
G.Gilbert, \emph{String theory path integral (genus two and higher)}, Nuclear Phys. B, \textbf{277}, (1986), 102-124.


\bibitem{Gradshteyn}
I. Gradshteyn, I. Ryzhik,  
 \emph{Table of integrals, series and products}, Academic Press, New York (1965)



\bibitem{Guillarmon}
C.Guillarmon, \emph{Generalized Krein formula,determinants,and Selberg zeta function in even dimension}, Amer.J.of Math., \textbf{131}, (2009), no.5, 1359-1417.



\bibitem{Hawking}
S. W. Hawking, \emph{Zeta function regularization of path integrals in curved space-time}, Comm.Math.Phys. \textbf{55}, (1977), 133--148.


\bibitem{hejhal}
 D. A. Hejhal, \emph{A classical approach to a well-known spectral correspondence on   quaternion groups}
 in Number theory, New York 1983-84, D. Chudnovsky, G. Chudnovsky, H. Cohen, and M. Nathanson, eds., Lecture Notes in Mathematics 1135, Berlin-Heidelberg-New York, 1985, Springer-Verlag.

\bibitem{Hoker}
E. D'Hoker, D. H. Phong \emph{Multiloop amplitudes for the bosonic Polyakov string}, Nucl. Phys. B, \textbf{269}, (1986), 205-234.


\bibitem{Hoker1}
E. D'Hoker, D. H. Phong \emph{On determinants of Laplacians on Riemann surfaces}, Comm. Math. Phys., \textbf{104}, (1986), 537-545.

\bibitem{Iwaniec}
Henryk Iwaniec \emph{Spectral Methods of Automorphic Forms}, Graduate Studies in Mathematics, \textbf{53}, (2002).


\bibitem{Koyama}
S.Y.Koyama, \emph{Determinant Expression of Selberg Zeta functions.$\mathbf I$}, Transactions of the American mathematical Society, \textbf{324}, (1991), no.1, 149-168.

\bibitem{Koyama2}
S.Y.Koyama, \emph{Determinant Expression of Selberg Zeta functions.$\mathbf II$}, Transactions of the American mathematical Society, \textbf{329}, (1992), no.2, 755-772.


\bibitem{Koyama3}
S.Y.Koyama, \emph{Determinant Expression of Selberg Zeta functions.$\mathbf III$}, Proceeding of the American Mathematical Society, \textbf{113}, (1991), no.2, 303-311.


\bibitem{Krein}
M. G. Krein, \emph{On perturbation determinants and a trace formula for unitary and self-adjoint operators}, Soviet.Math.Dokl. \textbf{3}, (1962), 707-710.

\bibitem{Kurokawa}
N. Kurokawa, M. Wakayama, \emph{Zeta regularizations}, Acta Appl.Math., \textbf{81}, (2004), 147--166.

\bibitem{Lang}
S. Lang, \emph{Algebra},Graduate text in mathematics, \textbf{211}, Springer-Verlag, (2002).



\bibitem{Phillips-Lax}
P. D. Lax, R. S. Phillips , \emph{Scattering theory for automorphic functions}, Annals of mathematics studies, \textbf{87}, Princeton University Press, (1976).


\bibitem{lewis}
  D.W.Lewis, \emph{Quaternion Algebras and the Algebraic Legacy of Hamilton's Quaternions},
  Irish Math. Soc. Bulletin \textbf{57} (2006), 41-64.


\bibitem{Maschot-Moore}
J. Manschot, Gregory W. Moore \emph{A Modern Faretail}, Commun. Num. Theor. Phys. \textbf{4}, (2010), 103-159.


\bibitem{Minakshisundaram}
S.Minakshisundaram, A. Pleijel, \emph{Some properties of the eigenvalues of the Laplace operator on Riemannian manifolds}, Canadian J.Math. \textbf{1}, (1949), 242-256.



\bibitem{miyake}
 T.Miyake, \emph{Modular forms}
 , Springer, 1989.





\bibitem{Arash}
A. Momeni, A. Venkov, \emph{An applicationof Jacquet-Langlands correspondence to transfer operator for geodesic flows on Riemann surfaces}, arxiv:0808202, 2008



\bibitem{Handbook}
F. W. J. Olver, D. W. Lozier, R. F. Boisvert, C. W. Clark, \emph{NIST Handbook of Mathematical Functions}, National Institute of Standards and Technology U.S. department of commerce and Cambridge University press, 2010.






\bibitem{Phillips-Sarnak}
R. Phillips and P. Sarnak, \emph{On cusp forms for co-finite subgroups of PSL(2,R)}, Invent.Math. \textbf{80}, (1985), 339-364. 



\bibitem{Rankin}
Robert A. Rankin, \emph{Modular forms and functions}, Cambridge university press (1977).


\bibitem{Ray}
D. B. Ray, I. M. Singer, \emph{R-torsion and the Laplacian on Riemannian manifolds}, Advances in Math. \textbf{7}, (1971), 145-210

\bibitem{Reschetikhin}
N. Reshetikhin, \emph{Lectures on quantization of gage systems}, Proceedings of the Summer School "New paths towards quantum gravity", Holbaek, Denmark; B. Booss-Bavnbek, G. Esposito and M. Lesch, eds. Springer, Berlin, (2010), 3-58.



\bibitem{Ruij}
S. N. M. Ruijsenaars, \emph{On Barnes' multiple zeta and gamma functions}, Advances in mathematics, \textbf{156}, (2000) 107--132.


\bibitem{Sarnak}
P. Sarnak,  \emph{Determinants of Laplacians}, Commun. Math. Phys., \textbf{110}, (1987), 113-120.


\bibitem{Schwinger}
J.Schwinger, \emph{On gauge invariance and vacuum polarization}, Phys. Rev. \textbf{82}, (1951), 664-679.


\bibitem{Selberg}
A. Selberg, \emph{Harmonic analysis and discontinuous groups in weakly symmetric Riemannian spaces with application to Dirichlet series}. J.Indian Math. Soc. \textbf{20} (1956), 47-87.



\bibitem{stroembergsson}
  A. Stroembergsson, \emph{Studies in the analytic and spectral theory of automorphic forms}.
 Phd Thesis, Uppsala University 2001, 226 pp.





\bibitem{Takhtajan}
L. A. Takhtajan, \emph{Quantum mechanics for mathematicians}, Graduate Studies in Mathematics, \textbf{95}, AMS, (2008).



\bibitem{Titchmarsh}
E. C. Titchmarsh, \emph{Theory of functions},
Oxford University Press, (1939).


\bibitem{Alexei}
A.B.Venkov, \emph{Spectral theory of automorphic functions and its applications},
Kluwer Academic Publishers, (1990).


\bibitem{Alexei2}
A.B.Venkov, \emph{Spectral theory of automorphic functions},
Pro. Steklov Math. Inst. \textbf{181}, (1983).


\bibitem{Alexei3}
A.B.Venkov, \emph{Selberg's trace formula for the Hecke operator generated by an involution, and the eigenvalues of the Laplace-Beltrami operator on the fundamental domain of the modular group $PSL(2,\mathbb Z)$},
Math. USSR Izvestija, \textbf{12}, No. 3 (1978).


\bibitem{vi}
  M.-F. Vigneras, \emph{Arithmetique des Algebres de Quaternions},
 Lecture Notes in Math. \textbf{800}, Springer-Verlag, Berlin, 1980. 


\bibitem{Vigneras}
Marie-France Vigneras, \emph{L'Equation fonctionnelle de la fonction zeta de Selberg du groupe modulaire $PSL(2, \mathbb Z)$},
Soc. Math. de France \textbf{61}, (1979) 235-249.


\bibitem{Voros}
A. Voros, \emph{Spectral functions, special functions and the Selberg zeta function},
Commun. Math. Phys. \textbf{110}, (1987) 439-465.


\bibitem{Watson}
E. T. Whittaker, G. N. Watson, \emph{A course of modern analysis},
Cambridge University Press, (1927).



\bibitem{weil}
  A. Weil, \emph{Basic Number Theory}.
  Springer, Berlin Heidelberg New York, 1967.
\end{thebibliography}
\end{document}